\documentclass[a4paper,10pt]{article}
\usepackage{amssymb}
\usepackage[dvips]{color}
\usepackage{graphicx}
\usepackage{amsmath}
\usepackage{enumerate}
\usepackage{a4wide}
\usepackage{tikz}
\usepackage{3dplot}

\newcommand{\R }{\ensuremath{\mathbb R}}

\newtheorem{theorem}{Theorem}

\newtheorem{corollary}[theorem]{Corollary}

\newtheorem{definition}[theorem]{Definition}
\newtheorem{example}[theorem]{Example}

\newtheorem{lemma}[theorem]{Lemma}
\newtheorem{notation}[theorem]{Notation}

\newtheorem{proposition}[theorem]{Proposition}

\newtheorem{remark}[theorem]{Remark}

\newenvironment{proof}[1][Proof]{\textbf{#1.} }{\ \rule{0.5em}{0.5em}}

\begin{document}

\title{Universal Arbitrage Aggregator in Discrete Time Markets under
Uncertainty \thanks{\textbf{Acknowledgements:} we wish to thank J.
Ob\l oj and F. Riedel for helpful discussions on this subject.}}

\author{Matteo Burzoni \\
{\small Milano University, email: matteo.burzoni@unimi.it} \and Marco
Frittelli \\
{\small Milano University, email: marco.frittelli@unimi.it} \and Marco
Maggis \thanks{%
The author have been supported by the Gruppo Nazionale per l'Analisi
Matematica, la Probabilit\`a e le loro Applicazioni (GNAMPA) of the Istituto
Nazionale di Alta Matematica (INdAM).} \qquad \\
{\small Milano University, email: marco.maggis@unimi.it}}
\maketitle

\begin{abstract}
In a model independent discrete time financial market, we discuss
the richness of the family of martingale measures in relation to
different
notions of Arbitrage, generated by a class $\mathcal{S}$ of \textit{%
significant} sets, which we call Arbitrage \emph{de la classe} $\mathcal{S}$%
. The choice of $\mathcal{S}$ reflects into the intrinsic
properties of the
class of polar sets of martingale measures. In particular: for $\mathcal{S}%
=\left\{ \Omega \right\} ,$ absence of Model Independent Arbitrage
is equivalent to the existence of a martingale measure; for
$\mathcal{S}$ being the open sets, absence of Open Arbitrage is
equivalent to the existence of full support martingale measures.
These results are obtained by adopting a technical filtration
enlargement and by constructing a universal aggregator of all
arbitrage opportunities. We further introduce the notion of market
feasibility and provide its characterization via arbitrage
conditions. We conclude providing a dual representation of Open
Arbitrage in terms of weakly open sets of probability measures,
which highlights the robust nature of this concept.
\end{abstract}



\section{Introduction: No Arbitrage under Uncertainty}

The introduction of Knightian Uncertainty in mathematical models
for Finance has recently renewed the attention on foundational
issues such as option pricing rules, super-hedging, and arbitrage
conditions.

\bigskip

We can distinguish two extreme cases:

\begin{enumerate}
\item We are completely sure about the reference probability
measure $P$. In this case, the classical notion of No Arbitrage or
NFLVR can be successfully applied (as in \cite{DMW90,DS94,DS98}).

\item We face complete uncertainty about any probabilistic model
and therefore we must describe our model independently by any
probability. In this case we might adopt a model independent
(weak) notion of No Arbitrage
\end{enumerate}

In the second case, a pioneering contribution was given in the
paper by Hobson \cite{Ho98} where the problem of pricing exotic
options is tackled under model mis-specification. In his approach
the key assumption is the existence of a martingale measure for
the market, consistent with the prices of some observed vanilla
options (see also \cite{BHR01,CO11,DOR14} for further
developments). In \cite{DH07}, Davis and Hobson relate the
previous problem to the absence of Model Independent Arbitrages,
by the mean of semi-static strategies. A step forward towards a
model-free version of the First Fundamental Theorem of Asset
Pricing in discrete time was formerly achieved by Riedel
\cite{Riedel} in a one period market and by Acciao at al.
\cite{AB13} in a more general setup.

\bigskip

Between cases 1. and 2., there is the possibility to accept that
the model could be described in a probabilistic setting, but we
cannot assume the knowledge of a specific reference probability
measure but at most of a set of priors, which leads to the new
theory of Quasi-sure Stochastic Analysis as in
\cite{Bion,Cohen,Denis,Denis2,Pe10,STZ11,STZ11a}.

The idea is that the classical probability theory can be
reformulated as far as the single reference probability $P$ is
replaced by a class of (possibly non-dominated) probability
measures $\mathcal{P}^{\prime }$. This is the case, for example,
of uncertain volatility (e.g. \cite{STZ11a}) where, in a general
continuous time market model, the volatility is only known to lie
in a certain interval $[\sigma _{m},\sigma _{M}]$. \newline In the
theory of arbitrage for non-dominated sets of priors, important
results were provided by Bouchard and Nutz \cite{Nutz2} in
discrete time. A
suitable notion of arbitrage opportunity with respect to a class $\mathcal{P}%
^{\prime }$, named $NA(\mathcal{P}^{\prime })$, was introduced and
it was shown that the no arbitrage condition is equivalent to the
existence of a family $\mathcal{Q}^{\prime }$ of martingale
measure having the same polar sets of $\mathcal{P}^{\prime }$. In
continuous time markets, a similar topic has been recently
investigated also by Biagini et al. \cite{BBKN14}.

\bigskip

Bouchard and Nutz \cite{Nutz2} answer the following question:
which is a good notion of arbitrage opportunity for all
\textbf{admissible} probabilistic models $P\in \mathcal{P}^{\prime
}$ (i.e. one single $H$ that works as an arbitrage for all
admissible models) ? To pose this question one has to know
\textbf{a priori} which are the admissible models, i.e. we have to
exhibit a subset of probabilities $\mathcal{P}^{\prime }$. \bigskip

\textit{In this paper our aim is to investigate arbitrage conditions and robustness properties
of markets that are described independently of any reference
probability or set of priors.}

\bigskip

We consider a financial
market
described by a discrete time adapted stochastic process $S:=(S_{t})_{t\in I}$%
, $I=\left\{ 0,\ldots ,T\right\} $, defined on $(\Omega ,\mathcal{F},\mathbb{%
F})$, $\mathbb{F}:=(\mathcal{F}_{t})_{t\in I}$, with $T<\infty $
and taking values in $\R^{d}$ (see Section \ref{setup}). Note we are not imposing any
restriction on $S$ so that it may describe generic financial
securities (for examples, stocks and/or options). Differently from
previous approaches in literature, in our setting the measurable
space $(\Omega ,\mathcal{F})$ and the price process $S$ defined on
it are given, and we investigate the properties of martingale
measures for $S$ induced by no arbitrage conditions. The class
$\mathcal{H}$
of admissible trading strategies is formed by all $\mathbb{F}$-predictable $%
d $-dimensional stochastic processes and we denote with
$\mathcal{M}$ the set of all probability measures under which $S$
is an $\mathbb{F}$-martingale
and with $\mathcal{P}$ the set of all probability measures on $(\Omega ,%
\mathcal{F}).$ We introduce therefore a flexible definition of Arbitrage which
allows us to characterize the richness of the set $\mathcal{M}$ in a unified framework.

\paragraph{Arbitrage \emph{de la classe} $\mathcal{S}$.}

We look for a single strategy $H$ in $\mathcal{H}$ which
represents an Arbitrage opportunity in some appropriate sense.
Let:
\begin{equation*}
\mathcal{V}_{H}^{+}=\left\{ \omega \in \Omega \mid \
V_{T}(H)(\omega )>0\right\} ,
\end{equation*}%
where $V_{T}(H)=\sum_{t=1}^{T}H_{t}\cdot (S_{t}-S_{t-1})$ is the
final value of the strategy $H$. It is natural to introduce
several notion of Arbitrage accordingly to the properties of the
set $\mathcal{V}_{H}^{+}$.

\begin{definition}
\label{classeCArb} Let $\mathcal{S}$ be a class of measurable subsets of $%
\Omega $ such that $\varnothing \notin \mathcal{S}$. A trading strategy $%
H\in \mathcal{H}$ is an Arbitrage de la classe $\mathcal{S}$ if

\begin{itemize}
\item[$\bullet $ ] $V_{0}(H)=0$, $V_{T}(H)(\omega )\geq 0\ \forall
\omega \in \Omega $ and $\mathcal{V}_{H}^{+}$ contains a set in
$\mathcal{S}$.
\end{itemize}
\end{definition}

The class $\mathcal{S}$ has the role to translate mathematically
the meaning of a \textquotedblleft true gain\textquotedblright .
When a probability $P$ is given (the \textquotedblleft reference
probability\textquotedblright ) then we agree on representing a
true gain as $P(V_{T}(H)>0)>0$ and therefore
the classical no arbitrage condition can be expressed as: no losses $%
P(V_{T}(H)<0)=0$ implies no true gain $P(V_{T}(H)>0)=0$. In a
similar fashion, when a subset $\mathcal{P}^{\prime }$ of
probability measures is given, one may replace the $P$-a.s.
conditions above with $\mathcal{P}$-q.s conditions, as in
\cite{Nutz2}. However, if we can not or do not want to rely on a
priory assigned set of probability measures, we may well use
another concept: there is a true gain if the set
$\mathcal{V}_{H}^{+}$ contains a set considered
\textit{significant}. This is exactly the role
attributed to the class $\mathcal{S}$ which is the core of Section \ref%
{sectionRobust}. Families of sets, not determined by some probability measures,
have been already used in the context of the first and second
fundamental theorem of
asset pricing respectively by Battig Jarrow \cite{BJ99} and Cassese \cite%
{C08} (see Section \ref{SubAlt} for a more specific comparison).

\bigskip

In order to investigate the properties of the martingale measures
induced by
No Arbitrage conditions of this kind we first study (see Section \ref%
{secGeometric}) the structural properties of the market adopting a
geometrical approach in the spirit of \cite{Pliska} but with
$\Omega $ being a general Polish space, instead of a finite sample
space. In particular, we characterize the class $\mathcal{N}$ of
the $\mathcal{M}$-polar sets i.e. those $B\subset \Omega $ such
that there is no martingale measure that can assign a positive
measure to $B$. In the model independent framework the set
$\mathcal{N}$ is induced by the market since the set of martingale
measure has not to withstand to any additional condition (such as
being equivalent to a certain $P$). Once these polar sets are
identified we explicitly build in Section \ref{main} a process
$H^{\bullet }$ which depends only on the price process $S$ and
satisfies:

\begin{itemize}
\item[$\bullet $ ] $V_{T}(H^{\bullet })(\omega )\geq 0\ \forall
\omega \in \Omega $

\item[$\bullet $ ] $N\subseteq \mathcal{V}_{H^{\bullet }}^{+}$ for every $%
N\in \mathcal{N}$.
\end{itemize}

This strategy is a measurable selection of a set valued process
$\mathbb{H}$, that we baptize \textbf{Universal Arbitrage
Aggregator} since for any $P$, which is not absolutely continuous with respect to $\mathcal{M}$, an
arbitrage
opportunity $H^P$ (in the classical sense) can be found among the values of $\mathbb{%
H}$. All the inefficiencies of the market are captured by the
process $H^{\bullet }$ but, in general, it fails to be $\mathbb{F}$-predictable.
To recover predictability we need to enlarge the natural
filtration of the process $S$ by considering a suitable
technical filtration $\widetilde{\mathbb{F}}:=\{\widetilde{\mathcal{F}}%
_{t}\}_{t\in I}$ which does not affect the set of martingale
measures, i.e. any martingale measure $Q\in \mathcal{M}$ can be
uniquely extended to a martingale measure $\widetilde{Q}$ on the
enlarged filtration.\newline This allows us to prove, in Section
\ref{main}, the main result of the paper:

\begin{theorem}
\label{C-polarEquiv} Let $(\Omega ,\widetilde{\mathcal{F}}_{T},\widetilde{%
\mathbb{F}})$ be the enlarged filtered space as in Section \ref%
{sectionFiltration} and let $\widetilde{\mathcal{H}}$ be the set of $d$%
-dimensional discrete time $\widetilde{\mathbb{F}}$-predictable
stochastic process. Then
\begin{equation*}
\text{No Arbitrage de la classe }\mathcal{S}\text{ in }\widetilde{\mathcal{H}%
}\Longleftrightarrow \mathcal{M}\neq \varnothing \text{ and }\mathcal{N}%
\text{ does not contain sets of }\mathcal{S}
\end{equation*}
\end{theorem}

In other words, properties of the family $\mathcal{S}$ have a dual
counterpart in terms of polar sets of the pricing functional.

In Section \ref{main} we further provide our version of the
Fundamental Theorem of Asset Pricing: the equivalence between
absence of Arbitrage de la classe $\mathcal{S}$ in
$\widetilde{\mathcal{H}}$ and the existence of martingale measures
$Q\in \mathcal{M}$ with the property that $Q(C)>0$ for all $C\in
\mathcal{S}.$

\paragraph{Model Independent Arbitrage.}

When $\mathcal{S}:=\left\{ \Omega \right\} $ then the Arbitrage de
la classe $\mathcal{S}$ corresponds to the notion of a Model
Independent Arbitrage. As $\Omega $ never belongs to the class of
polar sets $\mathcal{N}$, from Theorem \ref{C-polarEquiv} we
directly obtain the following result.

\begin{theorem}
\label{corollaryMI}%
\begin{equation*}
\text{No Model Independent Arbitrage in }\widetilde{\mathcal{H}}%
\Longleftrightarrow \mathcal{M}\neq \varnothing .
\end{equation*}
\end{theorem}

An analogous result has been obtained in \cite{AB13} when
considering a
single risky asset $S$ as the canonical process on the path space $\Omega=%
\mathbb{R}_+^T$, a possibly uncountable collection of options $%
(\varphi_\alpha)_{\alpha\in A}$ whose prices are known at time
$0$, and when trading is possible through semi-static strategies
(see also \cite{Ho11} for a detailed discussion). Assuming the
existence of an option $\varphi_0$ with a specific payoff,
equivalence in Theorem \ref{corollaryMI} is achieved in the
original measurable space
$(\Omega,\mathcal{F},\mathbb{F},\mathcal{H})$.
In our setup, although we are free to choose a $(d+k)$-dimensional process $%
S $ for modeling a finite number of options ($k$) on possibly
different underlying ($d$), the class $\widetilde{\mathcal{H}}$ of
admissible strategies are dynamic in every $S^i$ for $i=1,\ldots
d+k$. In order to incorporate the case of semi-static strategies
we would need to consider restrictions on
$\widetilde{\mathcal{H}}$ and for this reason the two results are
not directly comparable.

\paragraph{Arbitrage with respect to open sets.}

In the topological context, in order to obtain full support
martingale measures, the suitable choice for $\mathcal{S}$ is the
class of open sets. This selection determines the notion of Arbitrage
with respect to open sets, which we shorten as \textquotedblleft
Open Arbitrage\textquotedblright :

\begin{itemize}
\item[$\bullet $] \textbf{Open Arbitrage} is a trading strategy
$H\in \mathcal{H}$ such that $V_{0}(H)=0,$ $V_{T}(H)(\omega )\geq
0\ \forall \omega \in \Omega $ and $\mathcal{V}_{H}^{+}$ contains
an open set.
\end{itemize}
This concept admits the following dual reformulation (see Section \ref%
{secRob}, Proposition \ref{WS-characterization}).%
\begin{eqnarray}
&&\text{An Open Arbitrage consists in a trading strategy }H\in \mathcal{H}%
\text{ and a non empty }  \notag \\
&&\text{weakly open set }\mathcal{U}\subseteq \mathcal{P}\text{
such that
for all }P\in \mathcal{U}\text{, }V_{T}(H)\geq 0\text{ }P\text{-a.s. and }P(%
\mathcal{V}_{H}^{+})>0\text{.}  \label{open}
\end{eqnarray}

The robust feature of an open arbitrage is therefore evident from
this dual formulation, as a certain strategy $H$ satisfies
(\ref{open}) if it represents an arbitrage in the classical sense
for a whole open set of probabilities. In addition, if $H$ is such
strategy and we disregard any finite subset of probabilities then
$H$ remains an Open Arbitrage. Moreover every weakly open subset
of $\mathcal{U}$ contains a full support
probability $P$ (see Lemma \ref{LemWeakClosure}) under which $H$ is a $P$%
-Arbitrage in the classical sense. Full support martingale
measures can be efficiently used whenever we face model
mis-specification, since they have a well spread support that
captures the features of the sample space of events
without neglecting significantly large parts. In Dolinski and Soner \cite%
{DS14} the equivalence of a local version of NA and the existence
of full support martingale measures has been proved (see Section
2.5, \cite{DS14}) in a continuous time market determined by one
risky asset with proportional transaction costs.


\paragraph{Feasibility and approximating measures.}
In Section \ref{marketfeasible} we answer the question: which are
the markets that are feasible in the sense that the
properties of the market are nice for \textquotedblleft \textbf{most}%
\textquotedblright\ probabilistic models? Clearly this problem
depends on the choice of the feasibility criterion, but to this
aim we do not need to exhibit a priori a subset of probabilities.
On the opposite, given a market (described without reference
probability), the induced set of No Arbitrage models
(probabilities) for that market will determine if the market
itself is feasible or not. What is needed here is a good notion of
\textquotedblleft \textbf{most}" probabilistic models.
\\More precisely given the price process $S$ defined on $(\Omega ,\mathcal{F}),$ we introduce
 the set $\mathcal{P}_{0}$ of probability
measures that exhibit No Arbitrage in the classical sense:
\begin{equation}
\mathcal{P}_{0}=\left\{ P\in \mathcal{P}\mid \text{No Arbitrage
with respect to }P\right\} .  \label{NoA}
\end{equation}%
When%
\begin{equation*}
\overline{\mathcal{P}_{0}}^{\tau }=\mathcal{P}
\end{equation*}%
with respect to some topology $\tau $ the market is feasible in
the sense that any \textquotedblleft bad\textquotedblright\
reference probability can
be approximated by No Arbitrage probability models. We show in Proposition %
\ref{propNAfullSupp} that this property is equivalent to the
existence of a full support martingale measure if we choose $\tau
$ as the weak* topology.

One other contribution of the paper, proved in Section
\ref{marketfeasible}, is the following characterization of
feasible markets and absence of Open Arbitrage in terms of
existence of full support martingale measures. We denote with
$\mathcal{P}^{+}\subset \mathcal{P}$ the set of full support
probability measures.

\begin{theorem}
\label{mainEquivalence} The following are equivalent:

\begin{enumerate}
\item The market is feasible, i.e $\overline{\mathcal{P}_{0}}^{\sigma (%
\mathcal{P},C_{b})}=\mathcal{P}$;

\item There exists $P\in \mathcal{P}_{+}$ s.t. No Arbitrage w.r.to
$P$ (in the classical sense) holds true;

\item $\mathcal{M\cap P}_{+}\neq \varnothing $;

\item No Open Arbitrage holds with respect to admissible strategies $%
\widetilde{\mathcal{H}}$.
\end{enumerate}
\end{theorem}

Riedel \cite{Riedel} already pointed out the relevance of the
concept of full support martingale measures in a probability-free
set up. Indeed in a one period market model and under the
assumption that the price process is continuous with respect to
the state variable, he showed that the absence of a one point
arbitrage (non-negative payoff, with strict positivity in at least
one point) is equivalent to the existence of a full support
martingale measure. As shown in Section \ref{remarkCont}, this
equivalence is no longer true in a multiperiod model (or in a
single period model with non trivial initial sigma algebra), even
for price processes continuous in $\omega $. In
this paper we consider a multi-assets multi-period model without $\omega $%
-continuity assumptions on the price processes and we develop the
concept of open arbitrage, as well as its dual reformulation, that
allows for the equivalence stated in the above theorem.

Finally, we present a number of simple examples that point out:
the
differences between single period and multi-period models (examples \ref%
{multi}, \ref{multi1}, \ref{multi2}); the geometric approach to
absence of arbitrage and existence of martingale measures (Section
\ref{examples}); the
need in the multi-period setting of the disintegration of the atoms (example %
\ref{ex3d}); the need of the one period anticipation of some polar
sets (example \ref{svuotamento}). The consequences of our version of the FTAP for the robust
formulation of the superhedging duality will be analyzed in a
forthcoming paper.

\section{Financial Markets}

\label{setup}

We will assume that $(\Omega ,d)$ is a Polish space and $\mathcal{F}=%
\mathcal{B}(\Omega )$ is the Borel sigma algebra induced by the
metric $d$. The requirement that $\Omega $ is Polish is used in
Section \ref{sectionNull} to guarantee the existence of a proper
regular conditional probability, see Theorem \ref{SV}. We fix a
finite time horizon $T\geq 1$, a finite set of time indices
$I:=\left\{ 0,\ldots ,T\right\} $ and we set: $I_{1}:=\left\{
1,\ldots ,T\right\} $. Let $\mathbb{F}:=\{\mathcal{F}_{t}\}_{t\in
I}$ be a filtration with $\mathcal{F}_{0}=\left\{ \varnothing
,\Omega \right\} $\ and
$\mathcal{F}_{T}\subseteq \mathcal{F}$. We denote with $\mathcal{L}(\Omega ,%
\mathcal{F}_{t};\mathbb{R}^{d})$ the set of
$\mathcal{F}_{t}$-measurable
random variables $X:\Omega \rightarrow \mathbb{R}^{d}$ and with $\mathcal{L}%
(\Omega ,\mathbb{F};\mathbb{R}^{d})$ the set of adapted processes $%
X=(X_{t})_{t\in I}$ with $X_{t}\in \mathcal{L}(\Omega ,\mathcal{F}_{t};%
\mathbb{R}^{d})$.

The market consists of one non-risky asset $S_{t}^{0}=1$ for all
$t\in I,$ constantly equal to $1$, and $d\geq 1$ risky assets
$S^{j}=(S_{t}^{j})_{t\in I}$ , $j=1,\ldots ,d$, that are
real-valued adapted stochastic processes.
Let $S=[S^{1},\ldots ,S^{d}]\in \mathcal{L}(\Omega ,\mathbb{F};\mathbb{R}%
^{d})$ be the $d$-dimensional vector of the (discounted) price processes.%
\newline
In this paper we focus on arbitrage conditions, and therefore
without loss of generality we will restrict our attention to
self-financing trading strategies of zero initial cost. Therefore,
we may assume that a trading strategy $H=(H_{t})_{t\in I_{1}}$ is
a $\mathbb{R}^{d}$-valued predictable
stochastic process: $H=[H^{1},\ldots ,H^{d}],$ with $H_{t}\in \mathcal{L}%
(\Omega ,\mathcal{F}_{t-1};\mathbb{R}^{d}),$ and we denote with
$\mathcal{H}$
the class of all trading strategies. The (discounted) value process $%
V(H)=(V_{t}(H))_{t\in I}$ is defined by:
\begin{equation*}
V_{0}(H):=0,\quad V_{t}(H):=\sum_{i=1}^{t}H_{i}\cdot (S_{i}-S_{i-1}),\quad
t\geq 1.
\end{equation*}%
A (discrete time) financial market is therefore assigned, without
any
reference probability measure, by the quadruple [$(\Omega ,d);(\mathcal{B}%
(\Omega ),\mathbb{F});S;\mathcal{H}$%
] satisfying the previous conditions.

\begin{notation}
For $\mathcal{F}$-measurable random variables $X$ and $Y$, we
write $X>Y$ (resp. $X\geq Y,$ $X=Y)$ if $X(\omega )>Y(\omega )$
for all $\omega \in \Omega $ (resp. $X(\omega )\geq Y(\omega ),$
$X(\omega )=Y(\omega )$ for all $\omega \in \Omega ).$
\end{notation}

\subsection{Probability and martingale measures}

Let $\mathcal{P}:=\mathcal{P}(\Omega )$ be the set of all probabilities on $%
(\Omega ,\mathcal{F})$ and $C_{b}:=C_{b}(\Omega )$ the space of
continuous
and bounded functions on $\Omega $. Except when explicitly stated, we endow $%
\mathcal{P}$ with the weak$^{\ast }$ topology $\sigma
(\mathcal{P},C_{b}),$ so that $(\mathcal{P},\sigma
(\mathcal{P},C_{b}))$ is a Polish space (see \cite{Aliprantis}
Chapter 15 for further details). The convergence of $P_{n}$
to $P$ in the topology $\sigma (\mathcal{P},C_{b})$ will be denoted by $P_{n}%
\overset{w}{\rightarrow }P$ and the $\sigma
(\mathcal{P},C_{b})$-closure of
a set $\mathcal{Q}\subseteq \mathcal{P}$ will be denoted with $\overline{%
\mathcal{Q}}.$

We define the \emph{support} of an element $P\in \mathcal{P}$ as
\begin{equation*}
supp(P)=\bigcap \{C\in \mathcal{C}\mid P(C)=1\}
\end{equation*}%
where $\mathcal{C}$ are the closed sets in $(\Omega ,d)$. Under
our assumptions the support is given by
\begin{equation*}
supp(P)=\{\omega \in \Omega \mid P(B_{\varepsilon }(\omega
))>0\text{ for all }\varepsilon >0\},
\end{equation*}%
where $B_{\varepsilon }(\omega )$ is the open ball with radius
$\varepsilon $ centered in $\omega $.

\begin{definition}
We say that $P\in \mathcal{P}$ has full support if $supp(P)=\Omega
$ and we denote with
\begin{equation*}
\mathcal{P}_{+}:=\left\{ P\in \mathcal{P}\mid supp(P)=\Omega
\right\}
\end{equation*}%
the set of all probability measures having full support.
\end{definition}

Observe that $P\in \mathcal{P}_{+}$ if and only if $P(A)>0$ for
every open set $A$. Full support measures are therefore important,
from a topological point of view, since they assign positive
probability to all open sets.

\begin{definition}
The set of $\mathbb{F}$-martingale measures\ is%
\begin{equation}
\mathcal{M}(\mathbb{F})=\left\{ Q\in \mathcal{P}\mid S\text{ is a }(Q,%
\mathbb{F})\text{-martingale}\right\} .  \label{martingale}
\end{equation}%
and we set: $\mathcal{M}:=\mathcal{M}(\mathbb{F}),$ when the
filtration is not ambiguous, and
\begin{equation*}
\mathcal{M}_{+}=\mathcal{M}\cap \mathcal{P}_{+}.
\end{equation*}
\end{definition}

\begin{definition}
Let $P\in \mathcal{P}$ and $\mathcal{G}\subseteq \mathcal{F}$ be a
sub-sigma algebra of $\mathcal{F}$. The generalized conditional
expectation of a non negative $X\in \mathcal{L}(\Omega
,\mathcal{F},\mathbb{R})$ is defined by:
\begin{equation*}
E_{P}[X\mid \mathcal{G}]:=\lim_{n\rightarrow +\infty
}E_{P}[X\wedge n\mid \mathcal{G}],
\end{equation*}%
and for $X\in \mathcal{L}(\Omega ,\mathcal{F},\mathbb{R})$ we set $%
E_{P}[X\mid \mathcal{G}]:=E_{P}[X^{+}\mid
\mathcal{G}]-E_{P}[X^{-}\mid \mathcal{G}],$ where we adopt the
convention $\infty -\infty =-\infty $. All basic properties of the
conditional expectation still hold true (see for example
\cite{FKV09}). In particular if $Q\in \mathcal{M}$ and $H\in
\mathcal{H}$ then $E_{Q}[H_{t}\cdot (S_{t}-S_{t-1})\mid \mathcal{F}%
_{t-1}]=H_{t}\cdot E_{Q}[(S_{t}-S_{t-1})\mid \mathcal{F}_{t-1}]=0$
$Q$-a.s., so that $E_{Q}[V_{T}(H)]=0$ $Q$-a.s.
\end{definition}

\section{Arbitrage de la classe $\mathcal{S}$}

\label{sectionRobust}

Let $H\in \mathcal{H}$ \ and recall that
$\mathcal{V}_{H}^{+}:=\left\{ \omega \in \Omega \mid \
V_{T}(H)(\omega )>0\right\} $ and that $V_{0}(H)=0$.

\begin{definition}
Let $P\in \mathcal{P}.$ A \textbf{$P$-Classical Arbitrage} is a
trading strategy $H\in \mathcal{H}$ such that:

\begin{itemize}
\item[$\bullet $ ] $V_{T}(H)\geq 0$ $P-$a.s.$,$ and $P(\mathcal{V}%
_{H}^{+})>0 $
\end{itemize}

We denote with $NA(P)$ the absence of $P$-Classical Arbitrage.
\end{definition}

Recall the definition of Arbitrage de la classe $\mathcal{S}$
stated in the Introduction.

\begin{definition}
\label{def} Some examples of Arbitrage de la classe $\mathcal{S}$:

\begin{enumerate}
\item $H$ is a $1p$-Arbitrage when $\mathcal{S}=\left\{ C\in
\mathcal{F}\mid
C\neq \varnothing \right\} $. This is the weakest notion of arbitrage since $%
\mathcal{V}_{H}^{+}$ might reduce to a single point. The
$1p$-Arbitrage corresponds to the definition given in
\cite{Riedel}. This can be easily generalized to the
following notion of $n$ point Arbitrage: $H$ is an $np$-Arbitrage
when
\begin{equation*}
\mathcal{S}=\left\{C\in \mathcal{F} \mid C\text{ has at least
}n\text{ elements}\right\} ,
\end{equation*}%
and might be significant for $\Omega $ (at most) countable.

\item $H$ is an Open Arbitrage when $\mathcal{S}=\left\{ C\in \mathcal{B}%
(\Omega )\mid C\text{ open non-empty}\right\} $.

\item \label{4defS}$H$ is a $\mathcal{P}^{\prime }$-q.s. Arbitrage when $%
\mathcal{S}=\left\{ C\in \mathcal{F}\mid P(C)>0\text{ for some
}P\in \mathcal{P}^{\prime }\right\} ,$ for a fixed family
$\mathcal{P}^{\prime
}\subseteq \mathcal{P}$. Notice that $\mathcal{S}=(\mathcal{N}(\mathcal{P}%
^{\prime }))^{c}$, the complements of the polar sets of
$\mathcal{P}^{\prime }
$. Then there are No $\mathcal{P}^{\prime }$-q.s. Arbitrage if:%
\begin{equation*}
H\in \mathcal{H}\text{ such that }V_{T}(H)(\omega )\geq 0\ \forall
\omega \in \Omega \Rightarrow V_{T}(H)=0\text{
}\mathcal{P}^{\prime }\text{-q.s.}
\end{equation*}%
This definition is similar to the No Arbitrage condition in
\cite{Nutz2}, the only difference being that here we require
$V_{T}(H)(\omega )\geq 0\ \forall \omega \in \Omega $, while in
the cited reference it is only
required $V_{T}(H)\geq 0$ $\mathcal{P}^{\prime }$-q.s.. Hence No $\mathcal{P}%
^{\prime }$-q.s. Arbitrage is a condition weaker than No Arbitrage in \cite%
{Nutz2}.

\item $H$ is a $P$-a.s. Arbitrage when $\mathcal{S}=\left\{ C\in
\mathcal{F}\mid P(C)>0\right\}$ for a
fixed $P\in \mathcal{P}$. As in the previous example the
No $P$-a.s. Arbitrage is a weaker condition than the No
$P$-Classical Arbitrage condition, the only difference being that
here we require $V_{T}(H)(\omega )\geq 0\ \forall
\omega \in \Omega $, while in the classical definition it is only required $%
V_{T}(H)\geq 0$ $P$-a.s.

\item $H$ is a Model Independent Arbitrage when
$\mathcal{S}=\left\{ \Omega \right\} ,$ in the spirit of
\cite{AB13,DH07,CO11}.

\item $H$ is an $\varepsilon $-Arbitrage when $\mathcal{S}=\left\{
C_{\varepsilon }(\omega )\mid \omega \in \Omega \right\} ,$ where $%
\varepsilon >0$ is fixed and $C_{\varepsilon }(\omega )$ is the
closed ball in $(\Omega ,d)$ of radius $\varepsilon $ and centered
in $\omega .$
\end{enumerate}
\end{definition}

Obviously, for any class $\mathcal{S}$,
\begin{equation}
\text{No $1p$-Arbitrage}\Rightarrow \text{No Arbitrage de la classe }%
\mathcal{S}\Rightarrow \text{No Model Ind. Arbitrage}  \label{MIA}
\end{equation}%
and these notions depend only on the properties of the financial
market and are not necessarily related to any probabilistic
models.

\begin{remark}
The No Arbitrage concepts defined above, as well as the possible
generalization of No Free Lunch de la classe $\mathcal{S}$, can be
considered also in more general, continuous time, financial market
models. We choose to present our theory in the discrete time
framework, as the subsequent results in the next sections will
rely crucially on the discrete time setting.
\end{remark}

\begin{example} The flexibility of our approach relies on the arbitrary
choice of the class $\mathcal{S}$. Consider $\Omega=C^0([0,T];\R)$
which is a Polish space once endowed with the supremum norm
$\|\cdot\|_{\infty}$. We may consider two classes
\begin{equation*} \mathcal{S}^{\infty}  =  \{\text{ open balls in } \|\cdot\|_{\infty} \}\quad\text{and}
\quad \mathcal{S}^{1}  =  \{\text{ open balls in } \|\cdot\|_{1} \}
\end{equation*}
where $\|\omega\|_{1}=\int_0^T |\omega(t)|dt$. Notice that since
the integral operator $\int_0^T |\cdot|dt:C^0([0,T];\R)\rightarrow
\R$ is $\|\cdot\|_{\infty}$-continuous every open ball in
$\|\cdot\|_{1}$ is also open in $\|\cdot\|_{\infty}$. Hence every
Arbitrage de la classe $\mathcal{S}^1$ is also an Arbitrage de la
classe $\mathcal{S}^{\infty}$ but not the converse.
\\For instance consider a market described by an underlying process $S^1$ and a digital option $S^2$, where trading is allowed only in a set of finite times $\{0,1,...,T-1\}$. Define $S^1_0(\omega)=s_0$ for every $\omega\in\Omega$ and
$S^1_t(\omega)=\omega(t)$ for the underlying and
$S^2_t(\omega)=\mathbf{1}_B(\omega)\mathbf{1}_{T}(t)$ for the
option where $B:=\{\omega\mid
S^1_t(\omega)\in(s_0-\varepsilon,s_0+\varepsilon)\ \forall
t\in[0,T]\}$. A long position in the option at time $T-1$ is an
arbitrage de la classe $\mathcal{S}^{\infty}$ even though there
does not exist any arbitrage de la classe $\mathcal{S}^{1}$.
\end{example}

\subsection{Defragmentation}

When the reference probability $P\in \mathcal{P}$ is fixed, the
market admits a $P$-Classical Arbitrage if and only if there
exists a time $t\in
\{1,\ldots ,T\}$ and a random vector $\eta \in L^{0}(\Omega ,\mathcal{F}%
_{t-1},P;\mathbb{R}^{d})$ such that $\eta \cdot (S_{t}-S_{t-1})\geq 0$ $P$%
-a.s. and $P(\eta \cdot (S_{t}-S_{t-1})>0)>0$ (see \cite{DMW90} or \cite%
{FoSch}, Proposition 5.11). In our context the existence of an
Arbitrage de la classe $\mathcal{S}$, over a certain time interval
$[0,T]$, does not necessarily imply the existence of a single time
step where the arbitrage is realized on a set in $\mathcal{S}$. It
might happen, instead, that the agent needs to implement a
strategy over multiple time steps to achieve an arbitrage de la
classe $\mathcal{S}$. The following example shows exactly a simple
case in which this phenomenon occurs. Recall that
$\mathcal{L}(\Omega ,\mathcal{F};\R^{d})$ is the set of
$\R^{d}$-valued $\mathcal{F}$-measurable random variables on
$\Omega $.

\begin{example}
\label{multi} Consider a 2 periods market model composed by two
risky assets $S^{1},S^{2}$ on
$(\mathbb{R},\mathcal{B}(\mathbb{R}))$ which are described by the
following trajectories
\begin{equation*}
S^{1}\;:\;%
\begin{array}{llll}
& 3\rightarrow & 3 &\quad \omega \in A_{1} \\
\quad \nearrow &  & 5 &\quad \omega \in A_{2} \\
2\rightarrow & 2%
\begin{array}{l}
\nearrow \\
\searrow%
\end{array}
&  &  \\
\quad \searrow &  & 1 &\quad \omega \in A_{3} \\
& 1\rightarrow & 1 &\quad \omega \in A_{4} \\
&  &  &
\end{array}%
\qquad S^{2}\;:\;%
\begin{array}{llll}
& 7\rightarrow & 7 &\quad \omega \in A_{1} \\
\quad \nearrow &  & 3 &\quad \omega \in A_{2} \\
2\rightarrow & 2%
\begin{array}{l}
\nearrow \\
\searrow%
\end{array}
&  &  \\
\quad \searrow &  & 1 & \quad\omega \in A_{3} \\
& 1\rightarrow & 1 & \quad\omega \in A_{4} \\
&  &  &
\end{array}%
\end{equation*}%
Consider $H_{1}=(-1,+1)$ and $H_{2}=(\mathbf{1}_{A_{2}\cup A_{3}},-\mathbf{1}%
_{A_{2}\cup A_{3}})$. Then $H_{1}\cdot
(S_{1}-S_{0})=4\mathbf{1}_{A_{1}}$
and $H_{2}\cdot (S_{2}-S_{1})=2\mathbf{1}_{A_{2}}$. Choosing $A_{1}=\mathbb{Q%
}\cap (0,1),A_{2}=(\mathbb{R}\setminus \mathbb{Q})\cap (0,1)$ and $%
A_{3}=[1,+\infty ),A_{4}=(-\infty ,0]$ we observe that an Open
Arbitrage can
be obtained only by a two step strategy, while in each step we have \emph{%
only} $1p$-Arbitrages. \newline
In general the multi step strategy realizes the Arbitrage de la classe $%
\mathcal{S}$ at time $T$ even though it does not yield necessarily
a
positive gain at each time: i.e. there might exist a $t<T$ such that $%
\{V_{t}(H)<0\}\neq \varnothing $. This is the case of Example \ref%
{svuotamento}.
\end{example}

In the remaining of this section $\Delta
S_t=[S_{t}^1-S_{t-1}^1,\ldots,S_{t}^d-S_{t-1}^d]$.

\begin{lemma}
\label{easy}The strategy $H\in \mathcal{H}$ is a $1p$-Arbitrage if
and only
if there exists a time $t\in I_{1}$, an $\alpha \in \mathcal{L}(\Omega ,%
\mathcal{F}_{t-1};\R^{d})$ and a non empty $A\in \mathcal{F}_{t}$
such that
\begin{equation}
\begin{array}{cc}
\alpha (\omega )\cdot \Delta S_{t}(\omega )\geq 0 & \quad\forall
\,\omega \in
\Omega \\
\alpha (\omega )\cdot \Delta S_{t}(\omega )>0 &\quad \text{ on }A.%
\end{array}
\label{eqdef}
\end{equation}
\end{lemma}

\begin{proof}
($\Rightarrow $) Let $H\in \mathcal{H}$ be a $1p$-Arbitrage. Set
\begin{equation*}
\overline{t}=\min \{t\in \{1,\ldots ,T\}\mid V_{t}(H)\geq 0\text{ with }%
V_{t}(H)(\omega )>0\text{ for some }\omega \in \Omega \}.
\end{equation*}%
If $\overline{t}=1$, $\alpha =H_{1}$ satisfies the requirements. If $%
\overline{t}>1$, by definition, $\{V_{\overline{t}-1}(H)<0\}\neq
\varnothing $ or $\{V_{\overline{t}-1}(H)=0\}=\Omega $ . In the
first case, for $\alpha
=H_{\overline{t}}\mathbf{1}_{\{V_{\overline{t}-1}(H)<0\}}$ we have
$\alpha
\cdot \Delta S_{\overline{t}}\geq 0$ with strict inequality on $\{V_{%
\overline{t}-1}(H)<0\}$. In the latter case $\alpha
=H_{\overline{t}}$ satisfies the requirements. \newline
($\Leftarrow $) Take $\alpha \in \mathcal{L}(\Omega ,\mathcal{F}_{t-1};\R%
^{d})$ as by assumption and define $H\in \mathcal{H}$ by $H_{s}=0$
for every $s\neq t$ and $H_{t}=\alpha $. Hence $V_{T}(H)(\omega
)=V_{t}(H)(\omega )$ for every $\omega \in \Omega $ so that
$V_{T}(H)\geq 0$ and $\{\omega \in \Omega \mid V_{T}(H)(\omega
)>0\}=\{\omega \in \Omega \mid \alpha \cdot \Delta S_{t}(\omega
)>0\}$ which ends the proof.
\end{proof}

\begin{remark}
Notice that only the implication $(\Leftarrow )$ of the previous
Lemma holds true for Open Arbitrage. This means that there exists
an Open Arbitrage if
we can find a time $t\in I_{1}$, an $\alpha \in \mathcal{L}(\Omega ,\mathcal{%
F}_{t-1};\R^{d})$ and a set $A\in \mathcal{F}_{t}$ containing an
open set
such that (\ref{eqdef}) holds true. Similarly for Arbitrage de la classe $%
\mathcal{S}$. On the other hand the converse is false in general as
shown by Example \ref{multi}.
\end{remark}

The following Lemma provides a full characterization of Arbitrages
de la classe $\mathcal{S}$ by the mean of a multi-step
decomposition of the strategy.

\begin{lemma}[Defragmentation]
\label{defragmentation}The strategy $H\in \mathcal{H}$ is an
Arbitrage de la classe $\mathcal{S}$ if and only if there exists:

\begin{itemize}
\item[$\bullet $ ] a finite family $\{U_{t}\}_{t\in I}$ with $U_{t}\in \mathcal{F}_{t}$, $%
U_{t}\cap U_{s}=\varnothing $ for every $t\neq s$ and
$\bigcup_{t\in I}U_{t}$ contains a set in $\mathcal{S}$;

\item[$\bullet $ ] a strategy $\widehat{H}\in \mathcal{H}$ such that $V_{T}(\widehat{H}%
)\geq 0$ on $\Omega ,$ and $\widehat{H}_{t}\cdot \Delta S_{t}>0$
on $U_{t}$ for any $U_{t}\neq \varnothing .$
\end{itemize}
\end{lemma}

\begin{proof}
($\Rightarrow $) Let $H\in \mathcal{H}$ be an Arbitrage de la classe $%
\mathcal{S}$. Define $B_{t}=\{V_{t}(H)>0\}$ and
\begin{eqnarray*}
U_{1}=B_{1} &\Rightarrow &H_{1}\cdot \Delta S_{1}(\omega)>0\quad
\forall \,\omega
\in U_{1} \\
U_{2}=B_{1}^{c}\cap B_{2} &\Rightarrow &H_{2}\cdot \Delta
S_{2}(\omega)>0\quad
\forall \,\omega \in U_{2} \\
U_{T-1}=B_{1}^{c}\cap \ldots \cap B_{T-2}^{c}\cap B_{T-1}
&\Rightarrow
&H_{T-1}\cdot \Delta S_{T-1}(\omega)>0\quad \forall \,\omega \in U_{T-1} \\
U_{T}=B_{1}^{c}\cap \ldots \cap B_{T-2}^{c}\cap B_{T-1}^{c}\cap \mathcal{V}%
_{H}^{+} &\Rightarrow &H_{T}\cdot \Delta S_{T}(\omega)>0\quad
\forall \,\omega \in U_{T}
\end{eqnarray*}%
From the definition of $\{U_{1},U_{2},\ldots ,U_{T}\}$ we have that $%
\mathcal{V}_{H}^{+}\subseteq \bigcup_{i=1}^{T}U_{i}$. Set $\widehat{H}%
_{1}=H_{1}$ and consider the strategy for every $2\leq t\leq T$
given by
\begin{equation*}
\widehat{H}_{t}(\omega )=H_{t}(\omega )\mathbf{1}_{D_{t-1}}(\omega
)\quad \text{where }D_{t-1}=\left( \bigcup_{s=1}^{t-1}U_{s}\right)
^{c}.
\end{equation*}%
By construction $\widehat{H}\in \mathcal{H}$ and
$\widehat{H}_{t}\cdot \Delta S_{t}(\omega )>0$ for every $\omega
\in U_{t}$. \newline ($\Leftarrow $) The converse implication is
trivial.
\end{proof}

\section{Arbitrage de la classe $\mathcal{S}$ and Martingale Measures}

\label{secGeometric} Before addressing this topic in its full
generality we provide some insights into the
problem and we introduce some examples that will help to develop
the intuition on the approach that we adopt.
The required technical tools will then be stated in Sections \ref%
{mktFeas} and \ref{sectionNull}.\\

Consider the family of polar sets of $\mathcal{M}$
\begin{equation*}
\mathcal{N}:=\left\{ A\subseteq A^{\prime }\in \mathcal{F}\ \mid \
Q(A^{\prime })=0\ \forall \ Q\in \mathcal{M}\right\} .
\end{equation*}%
In Nutz and Bouchard \cite{Nutz2} the notion of
$NA(\mathcal{P}^{\prime })$ for any fixed family
$\mathcal{P}^{\prime }\subseteq \mathcal{P}$ is defined by:
\begin{equation*}
V_{T}(H)\geq 0\text{ }\mathcal{P}^{\prime }-q.s.\Rightarrow
V_{T}(H)=0\text{ }\mathcal{P}^{\prime }-q.s.
\end{equation*}%
where $H$ is a predictable process which is measurable with
respect to the universal completion of $\mathbb{F}$. One of the
main results in \cite{Nutz2} asserts that, under $NA(\mathcal{P}^{\prime })$%
, there exists a class $\mathcal{Q}^{\prime }$ of martingale
measures which
shares the same polar sets of $\mathcal{P}^{\prime }$. If we take $\mathcal{P%
}^{\prime }=\mathcal{P}$ then $NA(\mathcal{P})$ is equivalent to
No (universally measurable) $1p$-Arbitrage, since $\mathcal{P}$
contains all
Dirac measures. In addition, the class of polar sets of $\mathcal{P}$ is empty. In Section %
\ref{maxPolar} we will show that this same result is true also in
our setting as a consequence of Proposition \ref{LemNOpolar}. The
existence of a class of martingale measures with no polar sets
implies that $\forall \omega \in \Omega $ there exists $Q\in
\mathcal{M}$ such that $Q(\{\omega\})>0$ and since $\Omega $ is a
separable space we can
find a dense set $D:=\{\omega _{n}\}_{n=1}^{\infty }$, with associated $%
Q^{n}\in \mathcal{M}$, such that $\sum_{n=1}^{\infty
}\frac{1}{2^{n}}Q^{n}$ is a full support martingale measure (see
Lemma \ref{pasting}).
\begin{proposition}
We have the following implications \label{NoPA}

\begin{enumerate}
\item\label{item1conj}No $1p$-Arbitrage $\Longrightarrow \mathcal{M}_{+}\neq
\emptyset$.

\item \label{item2conj}$\mathcal{M}_{+}\neq \emptyset \Longrightarrow $ No
Open Arbitrage.
\end{enumerate}
\end{proposition}

\begin{proof}
The proof of \ref{item1conj}. is postponed to Section \ref{maxPolar}. \newline
We prove \ref{item2conj}. by observing that for any open set $O$ and $Q\in \mathcal{M}%
_{+}$ we have $Q(O)>0$. Since for any $H\in \mathcal{H}$ such that $%
V_{T}(H)\geq 0$ we have $Q(\mathcal{V}_{H}^{+})=0$, then
$\mathcal{V}_{H}^{+} $ does not contain any open set.
\end{proof}

\begin{example}
\label{CountNA}Note however that the existence of a full support
martingale measure is compatible with $1p$-Arbitrage so that the
converse implication of
\ref{item1conj}. in Proposition \ref{NoPA} does not hold. Let $(\Omega ,\mathcal{F})=(%
\mathbb{R}^{+},\mathcal{B}(\mathbb{R}^{+}))$. Consider the market
with one risky asset: $S_{0}=2$ and
\begin{equation}
S_{1}=\left\{
\begin{array}{ll}
3 &\quad \omega \in \mathbb{R}^{+}\setminus \mathbb{Q} \\
2 &\quad \omega \in \mathbb{Q}^{+}%
\end{array}%
\right.
\end{equation}%
Then obviously there exists a $1p$-Arbitrage even though there
exist full support martingale measures (those probabilities
assigning positive mass only to each rational).
\end{example}

As soon as we weaken No $1p$-Arbitrage, by adopting any other no
arbitrage conditions in Definition \ref{def}, there is no
guarantee of the existence of martingale measures, as shown in
Section \ref{examples}. In order to obtain the equivalence between
$\mathcal{M}\neq \varnothing $ and No Model Independent Arbitrage
(the weakest among the No Arbitrage conditions de la classe
$\mathcal{S}$) we will enlarge the filtration, as explained in
Section \ref{sectionFiltration}.

\subsection{Examples}

\label{examples}

This section provides a variety of counterexamples to many
possible conjectures on the formulation of the FTAP in the
model-free framework. A financially meaningful example is the one
of two call options with the same
spot price $p_1=p_2$ but with strike prices $K_1>K_2$, formulated in \cite%
{DH07}, which already highlights that the equivalence between
absence of model independent arbitrage and existence of martingale
measures is not possible.\newline

We consider a one period market (i.e. $T=1$) with $(\Omega ,\mathcal{F})=(%
\mathbb{R}^{+},\mathcal{B}(\mathbb{R}^{+}))$ and with $d=2$ risky assets $%
S=[S^{1},S^{2}]$, in addition to the riskless asset $S^{0}=1$.
Admissible
trading strategies are represented by vectors $H=(\alpha ,\beta )\in \mathbb{%
R}^{2}$ so that
\begin{equation*}
V_{T}(H)=\alpha \Delta S^{1}+\beta \Delta S^{2},
\end{equation*}%
where $\Delta S^{i}=S_{1}^{i}-S_{0}^{i}$ for $i=1,2$. Let $%
S_{0}=[S_{0}^{1},S_{0}^{2}]=[2,2]$,
\begin{equation}
S_{1}^{1}=\left\{
\begin{array}{ll}
3 & \quad\omega \in \mathbb{R}^{+}\setminus \mathbb{Q} \\
2 & \quad\omega \in \mathbb{Q}^{+}%
\end{array}%
\right. ;\quad S_{1}^{2}=\left\{
\begin{array}{ll}
1+\exp (\omega ) & \quad\omega \in \mathbb{R}^{+}\setminus \mathbb{Q} \\
1 & \quad\omega =0 \\
1+\exp (-\omega ) & \quad\omega \in \mathbb{Q}^{+}\setminus \{0\}%
\end{array}%
\right.  \label{ex1000}
\end{equation}%
and $\mathcal{F=F}^{S}.$ We notice the following simple facts.

\begin{enumerate}
\item There are no martingale measures:
\begin{equation*}
\mathcal{M}=\varnothing .
\end{equation*}%
Indeed, if we denote by $\mathcal{M}_{i}$ the set of martingale
measures for
the $i^{th}$ asset we have $\mathcal{M}_{1}=\{Q\in \mathcal{P}\ \mid \ Q(%
\mathbb{R}^{+}\setminus \mathbb{Q})=0\}$ and $\forall Q\in
\mathcal{M}_{2},\ Q(\mathbb{R}^{+}\setminus \mathbb{Q})>0$.

\item The final value of the strategy $H=(\alpha ,\beta )\in
\mathbb{R}^{2}$
is%
\begin{equation*}
V_{T}(H)=\left\{
\begin{array}{ll}
\alpha +\beta (\exp (\omega )-1) & \quad\omega \in
\mathbb{R}^{+}\setminus
\mathbb{Q} \\
-\beta & \quad\omega =0 \\
\beta (\exp (-\omega )-1) & \quad\omega \in \mathbb{Q}^{+}\setminus \{0\}%
\end{array}%
\right. .
\end{equation*}%
Only the strategies $H\in \mathbb{R}^{2}$ having $\beta =0$ and
$\alpha \geq
0$ satisfy $V_{T}(H)(\omega )\geq 0$ for all $\omega \in \Omega $. For $%
\beta =0$ and $\alpha >0,$
$\mathcal{V}_{H}^{+}=\mathbb{R}^{+}\setminus
\mathbb{Q}$ and therefore there are \textbf{No Open Arbitrage} and \textbf{%
No Model Independent Arbitrage (}but\textbf{\
}$\mathcal{M}=\varnothing )$. This fact persists even if we impose
boundedness restrictions on the process $S$ or on the admissible
strategies $H$, as the following modification of
the example shows: let $S_0=[2,2]$ and take $S_{1}^{1}= [2+\exp (-\omega )] \mathbf{1}_{\mathbb{R}^+\setminus \mathbb{Q}}+ 2 \mathbf{1}_{\mathbb{Q}^+}$ and $S_{1}^{2}=[1+\exp (\omega )\wedge 4] \mathbf{1}_{\mathbb{R}^+\setminus \mathbb{Q}}+ \mathbf{1}_{\{0\}}+ [1+\exp (-\omega )]\mathbf{1}_{\mathbb{Q}^+\setminus \{0\}}$.\\
\item Set $\mathcal{H}^{+}:=\left\{ H\in \mathcal{H}\mid \
V_{T}(H)\geq 0 \text{ and } V_{0}(H)=0\right\} $ so that we
have $\bigcup_{H\in \mathcal{H}^{+}}\mathcal{V}_{H}^{+}=\mathbb{R}%
^{+}\setminus \mathbb{Q}\subsetneqq \Omega $. This shows that the condition $%
\mathcal{M}=\varnothing $ is not equivalent to $\bigcup_{H\in \mathcal{H}%
^{+}}\mathcal{V}_{H}^{+}=\Omega $ i.e. it is not true that the set
of
martingale measures is empty iff for every $\omega$ there exists a strategy $%
H$ that gives positive wealth on $\omega$ and $V_0(H)=0$. In order
to
recover the equivalence between these two concepts (as in Proposition \ref%
{obloj}) we need to enlarge the filtration in the way explained in Section %
\ref{sectionFiltration}.

\item By fixing any probability $P$ there exists a $P$-Classical
Arbitrage, since the FTAP holds true and $\mathcal{M}=\varnothing
$. Indeed:

\begin{enumerate}
\item If $P(\mathbb{R}^{+}\setminus \mathbb{Q})=0,$ then $\beta =-1$ ($%
\alpha =0)$ yield a $P$-Classical arbitrage, since $\mathcal{V}_{H}^{+}=%
\mathbb{Q}^{+}$ and $P(\mathcal{V}_{H}^{+})=1$

\item If $P(\mathbb{R}^{+}\setminus \mathbb{Q})>0$ then $\beta =0$ and $%
\alpha =1$ yield a $P$-Classical arbitrage, since $\mathcal{V}_{H}^{+}=%
\mathbb{R}^{+}\setminus \mathbb{Q}$ and
$P(\mathcal{V}_{H}^{+})>0.$
\end{enumerate}

\item Instead, by adopting the definition of a $P$-a.s. Arbitrage ($%
V_{T}(H)(\omega )\geq 0$ for all $\omega \in \Omega $ and $P(\mathcal{V}%
_{H}^{+})>0$)$,$ there are two possibilities:

\begin{enumerate}
\item If $P(\mathbb{R}^{+}\setminus \mathbb{Q})=0,$ \textbf{No }$P$\textbf{%
-a.s. Arbitrage }holds, since only the strategies $H\in
\mathbb{R}^{2}$ having $\beta =0$ and $\alpha \geq 0$ satisfies
$V_{T}(H)(\omega )\geq 0$
for all $\omega \in \Omega $ and $\mathcal{V}_{H}^{+}=\mathbb{R}%
^{+}\setminus \mathbb{Q}$.

\item If $P(\mathbb{R}^{+}\setminus \mathbb{Q})>0,$ then $\beta =0$ and $%
\alpha =1$ \textbf{yield a }$P$\textbf{-a.s. arbitrage}, since $\mathcal{V}%
_{H}^{+}=\mathbb{R}^{+}\setminus \mathbb{Q}$ and
$P(\mathcal{V}_{H}^{+})>0.$
\end{enumerate}

\item Geometric approach: If we plot the vector $[\Delta
S^{1},\Delta S^{2}]$ on the real plane (see Figure \ref{PA}) we
see that there exists a unique separating hyperplane given by the
vertical axis. As a consequence 1p-Arbitrage can arise only by
investment in the first asset ($\beta =0 $). For a separating
hyperplane we mean an hyperplane in $\mathbb{R}^{d}$ passing by the
origin and such that one of the associated half-space contains
(not
necessarily strictly contains) all the image points of the random vector $%
[\Delta S^{1},\Delta S^{2}]$. Let us now consider this other example on $(%
\mathbb{R}^{+},\mathcal{B}(\mathbb{R}^{+})$. Let $S_{0}=[2,2]$, and%
\begin{equation}
S_{1}^{1}=\left\{
\begin{array}{ll}
3 & \quad\omega \in \mathbb{R}^{+}\setminus \mathbb{Q} \\
2 & \quad\omega =0 \\
1 & \quad\omega \in \mathbb{Q}^{+}\setminus \{0\}%
\end{array}%
\right. \qquad S_{1}^{2}=\left\{
\begin{array}{ll}
7 & \quad\omega \in \mathbb{R}^{+}\setminus \mathbb{Q} \\
2 & \quad\omega =0 \\
0 & \quad\omega \in \mathbb{Q}^{+}\setminus \{0\}%
\end{array}%
\right.  \label{ex1001}
\end{equation}

In both examples (\ref{ex1000}) and (\ref{ex1001}) there exist
separating hyperplanes i.e. a $1p$-Arbitrage can be obtained (see
Figure \ref{PA}). In example (\ref{ex1000}) $\mathcal{M}$ is empty
and we find a unique separating hyperplane: this hyperplane cannot
give a strict separation of
the set $[\Delta S^{1}(\omega ),\Delta S^{2}(\omega )]_{\omega \in \mathbb{Q}%
^{+}}$ even though $\mathbb{Q}^{+}$ does not support any
martingale measure. In example (\ref{ex1001})
$\mathcal{M}=\{\delta _{\omega =0}\}$, only the event $\{\omega
=0\}$ supports a martingale measure and there exists an infinite
number of hyperplanes which strictly separates the image of both
polar sets $\mathbb{R}^{+}\setminus \mathbb{Q}$ and
$\mathbb{Q}^{+}\setminus \{0\}$, namely, those separating the
convex grey region in Figure \ref{PA}.
\end{enumerate}

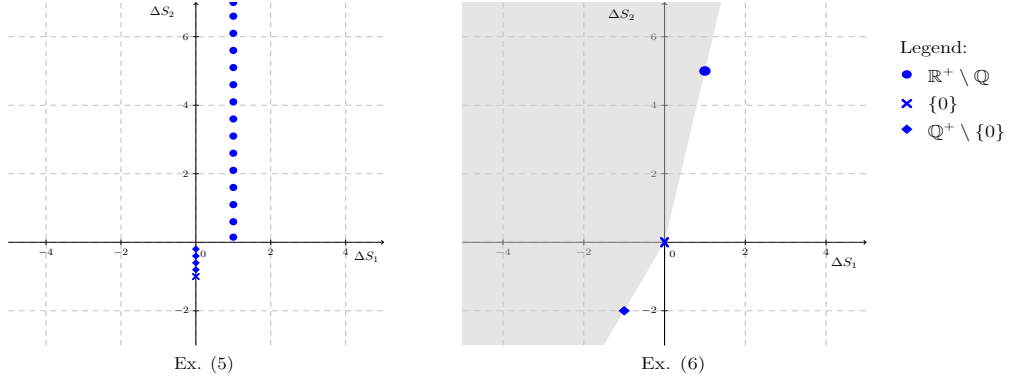
\begin{figure}[tbp]
\caption{{\protect\footnotesize {In examples (\protect\ref{ex1000}) and (%
\protect\ref{ex1001}), $0$ does not belong to the relative
interior of the convex set generated by the points $\{[\Delta
S^1(\protect\omega), \Delta
S^2(\protect\omega)]\}_{\protect\omega\in\Omega}$ and hence there
exists an hyperplane which separates the points.\newline }}}
\label{PA}\centering
\definecolor{cccccc}{rgb}{0.8,0.8,0.8}
\definecolor{qqqqff}{rgb}{0,0,1}
\definecolor{cqcqcq}{rgb}{0.75,0.75,0.75}
\begin{minipage}{.4\textwidth}
\resizebox{5cm}{5cm}{
\begin{tikzpicture}[line cap=round,line join=round,x=1cm,y=1cm]
\draw [color=cqcqcq,dash pattern=on 4pt off 4pt,
xstep=2.0cm,ystep=2.0cm] (-5,-3) grid (5,7); \draw[->,color=black]
(-5,0) -- (5.02,0); \foreach \x in {-4,-2,2,4}
\draw[shift={(\x,0)},color=black] (0pt,2pt) -- (0pt,-2pt)
node[below] {\footnotesize $\x$}; \draw[->,color=black] (0,-3) --
(0,7); \foreach \y in {-2,2,4,6} \draw[shift={(0,\y)},color=black]
(2pt,0pt) -- (-2pt,0pt) node[left] {\footnotesize $\y$};
\draw[color=black] (0pt,-10pt) node[right] {\footnotesize $0$};
\clip(-5,-4) rectangle (5,7); \draw (-1.35,7) node[anchor=north
west] {$ \Delta S_2 $}; \draw (4.1,-0.16) node[anchor=north west]
{$ \Delta S_1 $};
\begin{scriptsize}
\fill [color=qqqqff] (1,0.15) circle (3.0pt); \fill [color=qqqqff]
(1,0.6) circle (3.0pt); \fill [color=qqqqff] (1,1.1) circle
(3.0pt); \fill [color=qqqqff] (1,1.6) circle (3.0pt); \fill
[color=qqqqff] (1,2.1) circle (3.0pt); \fill [color=qqqqff]
(1,2.6) circle (3.0pt); \fill [color=qqqqff] (1,3.1) circle
(3.0pt); \fill [color=qqqqff] (1,3.6) circle (3.0pt); \fill
[color=qqqqff] (1,4.1) circle (3.0pt); \fill [color=qqqqff]
(1,4.6) circle (3.0pt); \fill [color=qqqqff] (1,5.1) circle
(3.0pt); \fill [color=qqqqff] (1,5.6) circle (3.0pt); \fill
[color=qqqqff] (1,6.1) circle (3.0pt); \fill [color=qqqqff]
(1,6.6) circle (3.0pt); \fill [color=qqqqff] (1,7) circle (3.0pt);
\fill [color=qqqqff] (0,-0.2) ++(-3.0pt,0 pt) --
++(3.0pt,3.0pt)--++(3.0pt,-3.0pt)--++(-3.0pt,-3.0pt)--++(-3.0pt,3.0pt);
\fill [color=qqqqff] (0,-0.4) ++(-3.0pt,0 pt) --
++(3.0pt,3.0pt)--++(3.0pt,-3.0pt)--++(-3.0pt,-3.0pt)--++(-3.0pt,3.0pt);
\fill [color=qqqqff] (0,-0.6) ++(-3.0pt,0 pt) --
++(3.0pt,3.0pt)--++(3.0pt,-3.0pt)--++(-3.0pt,-3.0pt)--++(-3.0pt,3.0pt);
\fill [color=qqqqff] (0,-0.8) ++(-3.0pt,0 pt) --
++(3.0pt,3.0pt)--++(3.0pt,-3.0pt)--++(-3.0pt,-3.0pt)--++(-3.0pt,3.0pt);
\draw [line width=1.5pt,domain=-0.08:0.08,color=qqqqff]
plot(\x,{(+1-1*\x)/-1}); \draw [line
width=1.5pt,domain=-0.08:0.08,color=qqqqff]
plot(\x,{(+1+1*\x)/-1});
\end{scriptsize}
\draw (-0.7,-3.2) node[anchor=north west] {\Large{Ex. (5)}};
\end{tikzpicture}}
\end{minipage}%
\begin{minipage}{.4\textwidth}
\resizebox{7.5cm}{5cm}{
\begin{tikzpicture}[line cap=round,line join=round,x=1.0cm,y=1.0cm]
\draw [color=cqcqcq,dash pattern=on 4pt off 4pt,
xstep=2.0cm,ystep=2.0cm] (-5,-3) grid (5,7); \draw[->,color=black]
(-5,0) -- (5,0); \foreach \x in {-4,-2,2,4}
\draw[shift={(\x,0)},color=black] (0pt,2pt) -- (0pt,-2pt)
node[below] {\footnotesize $\x$}; \draw[->,color=black] (0,-3) --
(0,7); \foreach \y in {-2,2,4,6} \draw[shift={(0,\y)},color=black]
(2pt,0pt) -- (-2pt,0pt) node[left] {\footnotesize $\y$};
\draw[color=black] (0pt,-10pt) node[right] {\footnotesize $0$};
\clip(-5,-4) rectangle (9,7); \fill[color=cccccc,fill=cccccc,fill
opacity=0.5] (2,10) -- (-8,10) -- (-8,-3)-- (-1.5,-3) -- (0,0) --
cycle;

\draw (-1.52,6.9) node[anchor=north west] {\parbox{3.18 cm}{$
\Delta S_2 \\  $}}; \draw (4,-0.34) node[anchor=north west] {$
\Delta S_1 $};
\begin{scriptsize}
\fill [color=qqqqff] (1,5) circle (4.0pt);
\fill [color=qqqqff] (-1,-2) ++(-4.0pt,0 pt) --
++(4.0pt,4.0pt)--++(4.0pt,-4.0pt)--++(-4.0pt,-4.0pt)--++(-4.0pt,4.0pt);

\draw [line width=2pt,domain=-0.08:0.08,color=qqqqff]
plot(\x,{(-0-1.5*\x)/-1}); \draw [line
width=2pt,domain=-0.08:0.08,color=qqqqff]
plot(\x,{(-0+1.5*\x)/-1});
\end{scriptsize}
\draw (-0.7,-3.2) node[anchor=north west] {\Large{Ex. (6)}};

\fill [color=qqqqff] (6,4.9) circle (3.5pt); \draw [line
width=2pt,domain=5.9:6.1,color=qqqqff] plot(\x,{(1.95-1*\x)/-1});
\draw [line width=2pt,domain=5.9:6.1,color=qqqqff]
plot(\x,{(-10.05+1*\x)/-1}); \fill [color=qqqqff] (6,3.3)
++(-4.0pt,0 pt) --
++(4.0pt,4.0pt)--++(4.0pt,-4.0pt)--++(-4.0pt,-4.0pt)--++(-4.0pt,4.0pt);

\draw (5.7,6) node[anchor=north west] {\Large{Legend:}}; \draw
(6.4,5.2) node[anchor=north west] {\Large
$\mathbb{R}^+\setminus\mathbb{Q}$}; \draw (6.4,4.4)
node[anchor=north west] {\Large $\{0\}$}; \draw (6.4,3.6)
node[anchor=north west] {\Large $\mathbb{Q}^+\setminus\{0\}$};
\end{tikzpicture}}
\end{minipage}%
\end{figure}

In conclusion the previous examples show that in a model-free
environment the existence of a martingale measure can not be
implied by arbitrage conditions - \textit{at least of the type
considered so far}. This is an important difference between the
model-free and quasi-sure analysis approach (see for example
\cite{Nutz2}):

\begin{itemize}
\item[$\bullet $ ] Model free approach: we deduce the `richness'
of the set $\mathcal{M}$ of martingale measures starting directly
from the underlying market structure $(\Omega ,\mathcal{F},S)$ and
we analyze the class of polar sets with respect to $\mathcal{M}$.

\item[$\bullet $ ] Quasi sure approach: the class of priors
$\mathcal{P}^{\prime }\subseteq \mathcal{P}$ and its polar sets
are given and one formulates a No-Arbitrage type condition to
guarantee the existence of a class of martingale measures which
has the same polar sets as the set of priors.
\end{itemize}

\paragraph{An alternative definition of Arbitrage.\label{SubAlt}}
The notion of No $P$-Classical Arbitrage ($%
P(V_{T}(H)<0)=0\Rightarrow P(V_{T}(H)>0)=0$) can be rephrased as: $%
V_{0}(H)=0$ and
\begin{equation}
\left\{ V_{T}(H)<0\right\} \text{ is negligible }\Rightarrow
\text{ }\left\{ V_{T}(H)>0\right\} \text{ is negligible}
\label{344}
\end{equation}%
or in our setting%
\begin{equation}
\mathcal{V}_{H}^{-}\text{ \ does not contain sets in
}\mathcal{S}\Rightarrow \mathcal{V}_{H}^{+}\text{ \ does not
contain sets in }\mathcal{S}. \label{334}
\end{equation}%
where $\mathcal{V}_{H}^{-}:=\left\{ \omega \in \Omega \mid
V_{T}(H)(\omega )<0\right\} .$ In the definition (\ref{334}) we
are giving up the requirement $V_{T}(H)\geq 0$, and so the
differences with respect to the existence of arbitrage
opportunities showed in Item 5 of the example in this section
disappear. However, this alternative definition of arbitrage does
not work well, as shown by the following example. Consider $(\Omega ,%
\mathcal{F})=(\mathbb{R}^{+},\mathcal{B}(\mathbb{R}^{+})$, a one
period market with one risky asset: $S_{0}=2$,
\begin{equation}
S_{1}=\left\{
\begin{array}{ll}
3 & \quad\omega \in \lbrack 1,\infty ) \\
2 & \quad\omega =[0,1)\setminus \mathbb{Q} \\
1 & \quad\omega \in \lbrack 0,1)\cap \mathbb{Q}%
\end{array}%
\right.  \label{ex12}
\end{equation}%
Consider the strategy of buying the risky asset: $H=1$. Then $\mathcal{V}%
_{H}^{-}=[0,1)\cap \mathbb{Q}$ does not contain an open set, $\mathcal{V}%
_{H}^{+}=[1,\infty )$ contains open sets. Therefore, there is an
Open Arbitrage (in the modified definition obtained from
(\ref{334})) but there
are full support martingale measures, for example $Q([0,1)\cap \mathbb{Q}%
)=Q([1,\infty ))=\frac{1}{2}.$ Notice also that by enlarging the
filtration the Open Arbitrage would persist.

A concept of no arbitrage similar to (\ref{344}) was introduced by
Cassese \cite{C08}, by adopting an ideal $\mathcal{N}$ of
\textquotedblleft negligible\textquotedblright\ sets - not
necessarily derived from probability measures. In a continuous
time setting, he proves that the absence of such an arbitrage is
equivalent to the existence of a finitely additive
\textquotedblleft martingale measure\textquotedblright . Our
results are not comparable with those by \cite{C08} since the
markets are clearly different, we do not require any structure on
the family $\mathcal{S} $ and \cite{C08} works with finitely
additive measures. In addition, the example (\ref{ex12}) just
discussed shows the limitation in our setting of the definition
(\ref{344}) for finding martingale probability measures with the
appropriate properties.

\subsection{Technical Lemmata}

\label{mktFeas}Recall that $S=(S_{t})_{t\in I}$ is an $\mathbb{R}^{d}$%
-valued stochastic process defined on a Polish space $\Omega $
endowed with
its Borel $\sigma $-algebra $\mathcal{F}=\mathcal{B}(\Omega )$ and $%
I_{1}:=\left\{ 1,...,T\right\} .$

\bigskip

\textit{Through the rest of the paper we will make use of the
natural filtration }$\mathcal{F}^{S}=\{\mathcal{F}_{t}^{S}\}_{t\in
I}$\textit{\ of
the process }$S$\textit{\ and for ease of notation we will not specify }$S$%
\textit{, but simply write }$\mathcal{F}_{t}$\textit{\ for }$\mathcal{F}%
_{t}^{S}$.

\bigskip

For the sake of simplicity we indicate by $\mathbf{Z}:=Mat(d\times
(T+1);\R)$ the space of $d\times (T+1)$ matrices with real entries
representing the space of all the possible trajectories of the
price process. Namely for every $\omega \in \Omega $ we have
$(S_{0}(\omega ),S_{1}(\omega ),...,S_{T}(\omega
))=(z_{0},z_{1},...,z_{T})=:z\in \mathbf{Z}$.
Fix $s\leq t$: for any $z\in \mathbf{Z}$ we indicate , the components from $%
s $ to $t$ by $z_{s:t}=(z_{s},...,z_{t})$ and $z_{t:t}=z_{t}$. Similarly $%
S_{s:t}=(S_{s},S_{s+1},...,S_{t})$ represents the process from time $s$ to $%
t $. \newline
We denote with $ri(K)$ the relative interior of a set $K\subseteq \mathbb{R}%
^{d}$. In this section we will make extensive use of the geometric
properties of the image in $\mathbb{R}^{d}$ of the increments of
the price process $\Delta S_{t}:=S_{t}-S_{t-1}$ relative to a set
$\Gamma \subseteq \Omega $. The typical sets that we will consider
are the level sets $\Gamma =\Sigma _{t-1}^{z}$, where:
\begin{equation}
\Sigma _{t-1}^{z}:=\{\omega \in \Omega \mid S_{0:t-1}(\omega
)=z_{0:t-1}\}\in \mathcal{F}_{t-1},\text{ }z\in \mathbf{Z}\text{,
}t\in I_{1} \label{Sigma}
\end{equation}%
and $\Gamma =A_{t-1}^{z},$ the intersection of the level set
$\Sigma _{t-1}^{z}$ with a set $A\in \mathcal{F}_{t-1}$:
\begin{equation}
A_{t-1}^{z}:=\{\omega \in A\mid S_{0:t-1}(\omega )=z_{0:t-1}\}\in \mathcal{F}%
_{t-1}.  \label{A}
\end{equation}%
%
%
%
%
%
%
%
%
%
%
%
%
%
%
%
%
%
%
%
%
%
For any $\Gamma \subseteq \Omega $ define the convex cone:
\begin{equation}
(\Delta S_{t}(\Gamma ))^{cc}:=co\left( conv\big(\Delta S_{t}(\Gamma )\big)%
\right) \cup \{0\}\subseteq \mathbb{R}^{d}.  \label{coconv}
\end{equation}%
If $0\in ri(\Delta S_{t}(\Gamma ))^{cc}$ we cannot apply the
hyperplane separating theorem to the convex sets $\{0\}$ and
$ri(\Delta S_{t}(\Gamma ))^{cc}$, namely, there is no $H\in
\mathbb{R}^{d}$ that satisfies $H\cdot \Delta S_{t}(\omega )\geq
0$ for all $\omega \in \Gamma $ with strict
inequality for some of them. As intuitively evident, and shown in Corollary %
\ref{NoPoint} below, $0\in ri(\Delta S_{t}(\Gamma ))^{cc}$ if and
only if No $1p$-Arbitrage are possible on the set $\Gamma $, since
a trading strategy on $\Gamma $ with a non-zero payoff always
yields both positive and negative outcomes.

In this situation, for $\Gamma =\Sigma _{t-1}^{z}$, the level set
is not suitable for the construction of a $1p$-Arbitrage
opportunity and sets with this property are naturally important
for the construction of a martingale measure. We wish then to
identify, for $\Gamma =\Sigma _{t-1}^{z}$ satisfying $0\notin
ri(\Delta S_{t}(\Gamma ))^{cc}$, those subset of $\Sigma
_{t-1}^{z}$ that retain this property. This result is contained in
the following key Lemma \ref{spezzamento}.

\bigskip

Observe first that for a convex cone $K\subseteq \mathbb{R}^{d}$ such that $%
0\notin ri(K)$ we can consider the family $V=\{v\in
\mathbb{R}^{d}\mid \Vert v\Vert =1\text{ and }v\cdot y\geq
0\;\forall \,y\in K\}$ so that
\begin{equation*}
\overline{K}=\bigcap_{v\in V}\{y\in \mathbb{R}^{d}\mid v\cdot
y\geq 0\}=\bigcap_{n\in \mathbb{N}}\{y\in \mathbb{R}^{d}\mid
v_{n}\cdot y\geq 0\},
\end{equation*}%
where $\{v_{n}\}=(\mathbb{Q}^{d}\cap V)\setminus \{0\}$.

\begin{definition}
Adopting the above notations, we will call $\sum_{n=1}^{\infty }\frac{1}{%
2^{n-1}}v_{n}\in V$ the standard separator.
\end{definition}

\begin{lemma}
\label{spezzamento}Fix $t\in I_{1}$ and $\Gamma \neq \varnothing $. If $%
0\notin ri(\Delta S_{t}(\Gamma ))^{cc}$ then there exist $\beta
\in \{1,\ldots ,d\}$, $H^{1},\ldots ,H^{\beta }$, $B^{1},\ldots
,B^{\beta
},B^{\ast }$ with $H^{i}\in \mathbb{R}^{d}$, $B^{i}\subseteq \Gamma $ and $%
B^{\ast }:=\Gamma \setminus \big(\cup _{j=1}^{\beta }B^{j}\big)$
such that:

\begin{enumerate}
\item $B^{i}\neq \varnothing $ for all $i=1,\ldots \beta $, and
$\{\omega \in \Gamma \mid \Delta S_{t}(\omega )=0\}\subseteq
B^{\ast }$ which may be empty;

\item $B^{i}\cap B^{j}=\varnothing $ if $i\neq j$;

\item \label{item_spezz_arb}$\forall i\leq \beta ,\ H^{i}\cdot
\Delta S_{t}(\omega )>0$ for all $\omega \in B^{i}$ and
$H^{i}\cdot \Delta S_{t}(\omega )\geq 0$ for all $\omega \in \cup
_{j=i}^{\beta }B^{j}\cup B^{\ast }$.

\item \label{item_spezz}$\forall H\in \mathbb{R}^{d}$ s.t. $H\cdot
\Delta S_{t}\geq 0$ on $B^{\ast }$ we have $H\cdot \Delta S_{t}=0$
on $B^{\ast }$.
\end{enumerate}

Moreover, for $z\in \mathbf{Z,}$ $A\in \mathcal{F}_{t-1}$ and
$\Gamma =A_{t-1}^{z}$ (or $\Gamma =\Sigma _{t-1}^{z}$) we have
$B^{i},B^{\ast }\in \mathcal{F}_{t}$ and
\begin{equation}
H(\omega ):=\sum_{i=1}^{\beta }H^{i}\mathbf{1}_{B^{i}}(\omega )
\label{levelArbitrage}
\end{equation}%
is an $\mathcal{F}_{t}$-measurable random variable that is
uniquely determined when we adopt for each $H^{i}$ the standard
separator.\newline Clearly in these cases, $\beta $, $H^{i}$, $H$,
$B^{i}$ and $B^{\ast }$ will depend on $t$ and $z$ and whenever
necessary they will be denoted by $\beta _{t,z}$, $H_{t,z}^{i}$,
$H_{t,z}$, $B_{t,z}^{i}$ and $B_{t,z}^{\ast }$.
\end{lemma}

\begin{proof}
Set $A^{0}:=\Gamma $ and $K^{0}=(\Delta S_{t}(\Gamma ))^{cc}
\subseteq \mathbb{R}^{d}$ and let $\Delta _{0}:=\{\omega \in
A^{0}\mid \Delta S_{t}(\omega )=0\},$ which may be empty.

\begin{description}
\item[Step 1:] The set $K^{0}\subseteq \mathbb{R}^{d}$ is
non-empty and
convex and so $ri\big(K^{0}\big)\neq \varnothing $. From the assumption $%
0\notin ri\big(K^{0}\big)$ there exists a standard separator
$H^{1}\in \mathbb{R}^{d}$ such that (i) $H^{1}\cdot \Delta
S_{t}(\omega )\geq 0$ for all $\omega \in A^{0}$; (ii)
$B^{1}:=\{\omega \in A^{0}\mid H^{1}\cdot \Delta S_{t}(\omega
)>0\}$ is non-empty. Set $A^{1}:=(A^{0}\setminus
B^{1})=\{\omega \in A^{0}\mid H^{1}\cdot \Delta S_{t}(\omega )=0\}$ and let $%
K^{1}:=(\Delta S_{t}(A^1 ))^{cc} $, which is a non-empty convex set with $%
dim(K^{1})\leq d-1$.\newline If $0\in ri\big(K^{1}\big)$ (this
includes the case $K^{1}=\left\{ 0\right\} $) the procedure is
complete: one cannot separate $\left\{ 0\right\} $ from
the relative interior of $K^{1}$. The conclusion is that $\beta =1$, $%
B^{\ast }=A^{1}=A^{0}\setminus B^{1}$ which might be empty, and
$\Delta
_{0}\subseteq B^{\ast }$. Notice that if $K^{1}=\left\{ 0\right\} $ then $%
B^{\ast }=\Delta _{0}$ which might be empty. Otherwise:

\item[Step 2:] If $0\notin ri\big(K^{1}\big)$ we find the standard
separator $H^{2}\in \mathbb{R}^{d}$ such that $H^{2}\cdot \Delta
S_{t}(\omega )\geq 0,$ for all $\omega \in A^{1},$ and the set
$B^{2}:=\{\omega \in A^{1}\mid
H^{2}\cdot \Delta S_{t}(\omega )>0\}$ is non-empty. Denote $%
A^{2}:=(A^{1}\setminus B^{2})$ and let $K^{2}=(\Delta S_{t}(A^2
))^{cc}$ with $dim(K^{2})\leq d-2$.\newline If $0\in
ri\big(K^{2}\big)$ (this includes $K^{2}=\left\{ 0\right\} $) the
procedure is complete and we have the conclusions with $\beta =2$ and $%
B^{\ast }=A^{1}\setminus B^{2}=A^{0}\setminus (B^{1}\cup B^{2})$, and $%
\Delta _{0}\subseteq B^{\ast }$. Notice that if $K^{1}=\left\{
0\right\} $ then $B^{\ast }=\Delta _{0}.$ Otherwise:

\item[...]

\item[Step d-1] If $0\notin ri\big(K^{d-2}\big)$ ...Set
$B^{d-1}\neq \varnothing ,$ $A^{d-1}=(A^{d-2}\setminus B^{d-1})$,
$K^{d-1}=(\Delta
S_{t}(A^{d-1} ))^{cc}$ with $dim(K^{d-1})\leq 1 $. If $0\in ri\big(K^{d-1}%
\big)$ the procedure is complete. Otherwise:

\item[Step d:] We necessarily have $0\notin ri\big(K^{d-1}\big)$, so that $%
dim(K^{d-1})=1$, and the convex cone $K^{d-1}$ necessarily
coincides with a half-line with origin in $0$. Then we find a
standard separator $H^{d}\in \mathbb{R}^{d}$ such that:
$B^{d}:=\{\omega \in A^{d-1}\mid H^{d}\cdot \Delta S_{t}(\omega
)>0\}\neq \varnothing $ and the set
\begin{equation*}
B^{\ast }:=\{\omega \in A^{d-1}\mid \Delta S_{t}(\omega
)=0\}=\{\omega \in A^{0}\mid \Delta S_{t}(\omega )=0\}=\Delta _{0}
\end{equation*}%
satisfies: $B^{\ast }=A^{d-1}\setminus B^{d}$. Set
$A^{d}:=A^{d-1}\setminus
B^{d}=B^{\ast }=\Delta _{0}$\ and $K^{d}:=(\Delta S_{t}(A^d ))^{cc} $. Then $%
K^{d}=\left\{ 0\right\} $.
\end{description}

\noindent Since $dim(\Delta S_{t}(\Gamma ))^{cc} \leq d$ we have at most $d$
steps. In case $\beta =d$ we have $\Gamma
=A^{0}=\bigcup_{i=1}^{d}B^{i}\cup \Delta _{0}$.
To prove the last assertion we note that for any fixed $t$ and
$z$, $B^{i}$ are $\mathcal{F}_{t}$-measurable since
$B^{i}=A_{t-1}^{z}\cap (f\circ S_{t})^{-1}((0,\infty ))$ where
$f:\mathbb{R}^{d}\mapsto \mathbb{R}$ is the
continuous function $f(x)=H^{i}\cdot (x-z_{t-1})$ with $H^{i}\in \mathbb{R}%
^{d}$ fixed.
\end{proof}

\begin{corollary}
\label{NoPoint}Let $t\in I_{1}$, $z\in \mathbf{Z}$, $A\in
\mathcal{F}_{t-1},$ $\Gamma =A_{t-1}^{z}$. Then $0\in ri(\Delta
S_{t}(\Gamma ))^{cc}$ if and only if there are No $1p$-Arbitrage
on $\Gamma $, i.e.:
\begin{equation}
\text{for all }H\in \mathbb{R}^{d}\text{ s.t.
}H(S_{t}-z_{t-1})\geq 0\text{ on }\Gamma \text{ we have
}H(S_{t}-z_{t-1})=0\text{ on }\Gamma .  \label{H1}
\end{equation}
\end{corollary}

\begin{proof}
Let $0\notin ri(\Delta S_{t}(\Gamma ))^{cc}.$ Then from Lemma \ref%
{spezzamento}-\ref{item_spezz_arb}) with $i=1$ we obtain a $1p$-Arbitrage $%
H^{1}$ on $\Gamma =\cup _{j=1}^{\beta }B^{j}\cup B^{\ast }$, since $%
B^{1}\neq \varnothing .$ Viceversa, if $0\in ri(\Delta
S_{t}(\Gamma ))^{cc}$ we obtain (\ref{H1}) from the argument
following equation (\ref{coconv}).
\end{proof}

\begin{definition}
\label{defBeta0}For $A\in \mathcal{F}_{t-1}$ and $\Gamma
=A_{t-1}^{z}$ we naturally extend the definition of $\beta _{t,z}$
in Lemma \ref{spezzamento} to the case of $0\in ri(\Delta
S_{t}(\Gamma ))^{cc}$ using
\begin{equation*}
\beta _{t,z}=0\overset{\cdot }{\Leftrightarrow }0\in ri(\Delta
S_{t}(\Gamma ))^{cc}
\end{equation*}%
with $B_{t,z}^{0}=\varnothing $ and $B_{t,z}^{\ast }=A_{t-1}^{z}\in \mathcal{%
F}_{t-1}$. In this case, we also extend the definition of the
random variable in \eqref{levelArbitrage} as $H_{t,z}(\omega
)\equiv 0$.
\end{definition}

\begin{corollary}
\label{corUnivParb}Let $t\in I_{1}$, $z\in \mathbf{Z,}$ $A\in \mathcal{F}%
_{t-1}$ and $\Gamma =A_{t-1}^{z}$ with $0\notin ri(\Delta
S_{t}(\Gamma ))^{cc}$. For any $P\in \mathcal{P}$ s.t. $P(\Gamma
)>0$ let $j:=\inf \{1\leq i\leq \beta \mid P(B_{t,z}^{i})>0\}$. If
$j<\infty $ the trading
strategy $H(s,\omega ):=H^{j}\mathbf{1}_{\Gamma }(\omega )\mathbf{1}%
_{\{t\}}(s)$ is a $P$-Classical Arbitrage.
\end{corollary}

\begin{proof}
From Lemma \ref{spezzamento}-\ref{item_spezz_arb}) we obtain that: $%
H^{j}\Delta S_{t}(\omega )>0$ on $B_{t,z}^{j}$ with $P(B_{t,z}^{j})>0;$ $%
H^{j}\Delta S_{t}(\omega )\geq 0$ on $\bigcup_{i=j}^{\beta
_{t,z}}B_{t,z}^{i}\cup B_{t,z}^{\ast }$ and $P(B_{t,z}^{k})=0$ for
$1\leq k<j $.
\end{proof}

\begin{remark}
\label{ImportantRemark}Let $D\subseteq \mathbb{R}^{d}$ and $%
C:=(D)^{cc}\subseteq \mathbb{R}^{d}$ be the convex cone generated
by $D.$ If
$0\in ri(C)$ then for any $x\in D$ there exist a finite number of elements $%
x_{j}\in D$ such that $0$ is a convex combination of $\left\{
x,x_{1},...,x_{m}\right\} $ with a strictly positive coefficient
of $x$.
Indeed, fix $x\in D$ and recall that for any convex set $C\subseteq \mathbb{R%
}^{d}$ we have
\begin{equation*}
ri(C):=\left\{ z\in C\mid \forall x\in C\text{ }\exists
\varepsilon >0\text{ s.t. }z-\varepsilon (x-z)\in C\right\} .
\end{equation*}%
As $0\in ri(C)$ and $x\in D\subseteq C$ we obtain $-\varepsilon
x\in C,$ for
some $\varepsilon >0,$ and therefore: $\frac{\varepsilon }{1+\varepsilon }x+%
\frac{1}{1+\varepsilon }(-\varepsilon x)=0.$ Since $-\varepsilon
x\in C$ then it is a linear combination with non negative
coefficients of elements
of $D$ and we obtain: $\frac{\varepsilon }{1+\varepsilon }x+\frac{1}{%
1+\varepsilon }\sum_{j=1}^{m}\alpha _{j}x_{j}=0$, which can be
rewritten - possibly normalizing the coefficients - as: $\lambda
x+\sum_{j=1}^{m}\lambda
_{j}x_{j}=0$, with $x_{j}\in D,$ $\lambda +\sum_{j=1}^{m}\lambda _{j}=1$, $%
\lambda >0$ and $\lambda _{j}\geq 0$. When the set $D\subseteq \mathbb{R}%
^{d} $ is the set of the image points of the increment of the price process $%
[\Delta S_{t}(\omega )]_{\omega \in \Gamma }$, for a fixed time
$t$, this observation shows that, however we choose $\omega \in
\Gamma $ we can
construct a conditional martingale measure, relatively to the period $%
[t-1,t],$ which assigns a strictly positive weight to $\omega $
and has finite support. The measure is determined by the
coefficients $\left\{ \lambda ,\lambda _{1},...,\lambda
_{m}\right\} $ in the equation: $0=\lambda \Delta S_{t}(\omega
)+\sum_{j=1}^{m}\lambda _{j}\Delta S_{t}(\omega _{j})$. This
heuristic argument is made precise in the following Corollary and
it will be used also in the proof of Proposition \ref{LemNOpolar}.
\end{remark}

\begin{corollary}
\label{corNoArbBeta}Let $z$, $t$, $\Gamma =A_{t-1}^{z}$ and
$B_{t,z}^{\ast }$ as in Lemma \ref{spezzamento}.
\begin{equation*}
\text{For all }U\subseteq B_{t,z}^{\ast },\;U\in \mathcal{F}\quad \text{%
there exists }Q\in \mathcal{M}(B_{t,z}^{\ast })\text{ s.t. }Q(U)>0
\end{equation*}%
where $\mathcal{M}(B)=\{Q\in \mathcal{P}\mid Q(B)=1\;\text{ and }%
E_{Q}[S_{t}\mid \mathcal{F}_{t-1}]=S_{t-1}$ $Q$-a.s.$\}$, for $B\in \mathcal{%
F}.$
\end{corollary}

\begin{proof}
From Lemma \ref{spezzamento}-\ref{item_spezz}) there are no
$1p$-Arbitrage restricted to $\Gamma =B_{t,z}^{\ast }$. Applying
Corollary \ref{NoPoint}
this implies that $0\in ri(\Delta S_{t}(B_{t,z}^{\ast }))^{cc}$. Take any $%
\omega \in U\subseteq B_{t,z}^{\ast }$ . Applying Remark \ref%
{ImportantRemark} to the set $D:=\Delta S_{t}(B_{t,z}^{\ast })$ and to $%
x:=\Delta S_{t}(\omega )\in D$, we deduce the existence of
$\{\omega _{1},\ldots ,\omega _{m}\}\subseteq B_{t,z}^{\ast }$ and
non negative coefficients $\left\{ \lambda _{t}(\omega
_{1}),...,\lambda _{t}(\omega _{m})\right\} $ and $\lambda
_{t}(\omega )>0$ such that: $\lambda
_{t}(\omega )+\sum_{j=1}^{m}\lambda _{t}(\omega _{j})=1$ and%
\begin{equation*}
0=\lambda _{t}(\omega )\Delta S_{t}(\omega )+\sum_{j=1}^{m}\lambda
_{t}(\omega _{j})\Delta S_{t}(\omega _{j}).
\end{equation*}%
Since $\{\omega _{1},\ldots ,\omega _{m}\}\subseteq B_{t,z}^{\ast }$ and $%
\omega \in B_{t,z}^{\ast }$ we have $S_{t-1}(\omega _{j})=z_{t-1}$ and $%
S_{t-1}(\omega )=z_{t-1}.$ Therefore:
\begin{equation}
0=\lambda _{t}(\omega )(S_{t}(\omega
)-z_{t-1})+\sum_{j=1}^{m}\lambda _{t}(\omega _{j})(S_{t}(\omega
_{j})-z_{t-1}),
\end{equation}%
so that $Q(\{\omega \})=\lambda _{t}(\omega )$ and $Q(\{\omega
_{j}\})=\lambda _{t}(\omega _{j}),$ for all $j$, give the desired
probability.
\end{proof}
\bigskip
\begin{example}
\label{ex3d} Let $(\Omega ,\mathcal{F})=(\mathbb{R}^{+},\mathcal{B}(\mathbb{R%
}^{+}))$ and consider a single period market with $d=3$ risky asset $%
S_{t}=[S_{t}^{1},S_{t}^{2},S_{t}^{3}]$ with $t=0,1$ and
$S_{0}=[2,2,2]$. Let
\begin{equation*}
S_{1}^{1}(\omega )=%
\begin{cases}
1 & \omega \in \mathbb{R}^{+}\setminus \mathbb{Q} \\
2 & \omega \in \mathbb{Q}\cap \lbrack 1/2,+\infty ) \\
3 & \omega \in \mathbb{Q}\cap [0,1/2)\qquad
\end{cases}%
S_{1}^{2}(\omega )=%
\begin{cases}
2 & \omega \in \mathbb{R}^{+}\setminus \mathbb{Q} \\
1+\omega ^{2} & \omega \in \mathbb{Q}\cap \lbrack 1/2,+\infty ) \\
1+\omega ^{2} & \omega \in \mathbb{Q}\cap [0,1/2)\qquad
\end{cases}%
\end{equation*}%
\begin{equation*}
S_{1}^{3}(\omega )=%
\begin{cases}
2+\omega ^{2} & \omega \in \mathbb{R}^{+}\setminus \mathbb{Q} \\
2 & \omega \in \mathbb{Q}\cap \lbrack 1/2,+\infty ) \\
2 & \omega \in \mathbb{Q}\cap [0,1/2)\qquad
\end{cases}%
\end{equation*}
\end{example}

Fix $t=1$ and $z\in \mathbf{Z}$ with $z_{0}=S_{0}$. It is easy to
check that in this case $\beta _{t,z}=2$ with
$B_{t,z}^{1}=\mathbb{R}^{+}\setminus
\mathbb{Q}$, $B_{t,z}^{2}=\mathbb{Q}\cap [0,1/2)$, $B_{t,z}^{\ast }=\mathbb{Q%
}\cap [1/2,+\infty )$. The corresponding strategies
$H=[h_{1},h_{2},h_{3}]$
(standard in the sense of Lemma \ref{spezzamento}) are given by $%
H_{t,z}^{1}=[0,0,1]$ and $H_{t,z}^{2}=[1,0,0]$. Note that
$H_{t,z}^{1}$ is a 1p-arbitrage with
$\mathcal{V}_{H_{t,z}^{1}}^{+}=B_{t,z}^{1}$. We have therefore
that $B_{t,z}^{1}$ is a null set with respect to any martingale
measure. The strategy $H_{t,z}^{2}$ satisfies $V_{T}(H_{t,z}^{2})\geq 0$ on $%
(B_{t,z}^{1})^{c}$ with $\mathcal{V}_{H_{t,z}^{2}}^{+}=B_{t,z}^{2}$ hence, $%
B_{t,z}^{2}$ is also an $\mathcal{M}$-polar set. This example
shows the need
of a multiple separation argument, as it is not possible to find a \emph{%
single} separating hyperplane $H\in \mathbb{R}^{d}$ such that the
image points of $B_{t,z}^{1}\cup B_{t,z}^{2}$ (which is
$\mathcal{M}$-polar), through the random vector $\Delta S,$ are
strictly contained in one of the
associated half-spaces. We have indeed that $B_{t,z}^{2}$ is a subset of $%
\{\omega \in \Omega \mid H_{t,z}^{1}(S_{1}-S_{0})=0\}$ where
$H_{t,z}^{1}$ is the only 1p-arbitrage in this market.

\definecolor{red}{rgb}{1,0,0}
\definecolor{blue}{rgb}{0,0,1}
\definecolor{qreen}{rgb}{0,1,0}


\tdplotsetmaincoords{60}{138}

\pgfmathsetmacro{\rvec}{.8} \pgfmathsetmacro{\thetavec}{30}
\pgfmathsetmacro{\phivec}{60}

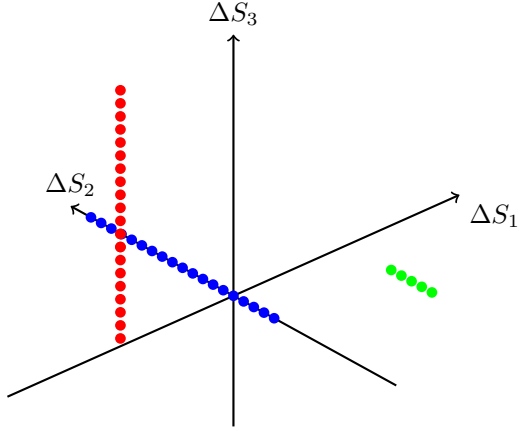
\begin{figure}[h]
\centering

\begin{tikzpicture}[scale=4,tdplot_main_coords]

\draw[thick,->] (1,0,0) -- (-1,0,0) node[anchor=north
west]{$\Delta S_1$}; \draw[thick,->] (0,0.8,0) -- (0,-0.8,0)
node[anchor=south]{$\Delta S_2$}; \draw[thick,->] (0,0,-0.5) --
(0,0,1) node[anchor=south]{$\Delta S_3$};

\fill [color=blue] (0,0.2,0) circle (0.5pt); \fill [color=blue]
(0,0.15,0) circle (0.5pt); \fill [color=blue] (0,0.1,0) circle
(0.5pt); \fill [color=blue] (0,0.05,0) circle (0.5pt); \fill
[color=blue] (0,0.0,0) circle (0.5pt); \fill [color=blue]
(0,-0.05,0) circle (0.5pt); \fill [color=blue] (0,-0.1,0) circle
(0.5pt); \fill [color=blue] (0,-0.15,0) circle (0.5pt); \fill
[color=blue] (0,-0.2,0) circle (0.5pt); \fill [color=blue]
(0,-0.25,0) circle (0.5pt); \fill [color=blue] (0,-0.3,0) circle
(0.5pt); \fill [color=blue] (0,-0.35,0) circle (0.5pt); \fill
[color=blue] (0,-0.4,0) circle (0.5pt); \fill [color=blue]
(0,-0.45,0) circle (0.5pt); \fill [color=blue] (0,-0.5,0) circle
(0.5pt); \fill [color=blue] (0,-0.55,0) circle (0.5pt); \fill
[color=blue] (0,-0.6,0) circle (0.5pt); \fill [color=blue]
(0,-0.65,0) circle (0.5pt); \fill [color=blue] (0,-0.7,0) circle
(0.5pt);

\fill [color=red] (0.5,0,0.03) circle (0.5pt); \fill [color=red]
(0.5,0,0.08) circle (0.5pt); \fill [color=red] (0.5,0,0.13) circle
(0.5pt); \fill [color=red] (0.5,0,0.18) circle (0.5pt); \fill
[color=red] (0.5,0,0.23) circle (0.5pt); \fill [color=red]
(0.5,0,0.28) circle (0.5pt); \fill [color=red] (0.5,0,0.33) circle
(0.5pt); \fill [color=red] (0.5,0,0.38) circle (0.5pt); \fill
[color=red] (0.5,0,0.43) circle (0.5pt); \fill [color=red]
(0.5,0,0.38) circle (0.5pt); \fill [color=red] (0.5,0,0.43) circle
(0.5pt); \fill [color=red] (0.5,0,0.48) circle (0.5pt); \fill
[color=red] (0.5,0,0.53) circle (0.5pt); \fill [color=red]
(0.5,0,0.58) circle (0.5pt); \fill [color=red] (0.5,0,0.63) circle
(0.5pt); \fill [color=red] (0.5,0,0.68) circle (0.5pt); \fill
[color=red] (0.5,0,0.73) circle (0.5pt); \fill [color=red]
(0.5,0,0.78) circle (0.5pt); \fill [color=red] (0.5,0,0.83) circle
(0.5pt); \fill [color=red] (0.5,0,0.88) circle (0.5pt); \fill
[color=red] (0.5,0,0.93) circle (0.5pt); \fill [color=red]
(0.5,0,0.98) circle (0.5pt);

\fill [color=green] (-0.5,0.22,0) circle (0.5pt); \fill
[color=green] (-0.5,0.27,0) circle (0.5pt); \fill [color=green]
(-0.5,0.32,0) circle (0.5pt); \fill [color=green] (-0.5,0.37,0)
circle (0.5pt); \fill [color=green] (-0.5,0.42,0) circle (0.5pt);
\end{tikzpicture}

\caption{{\protect\footnotesize {Decomposition of $\Omega$ in
Example \protect\ref{ex3d}: $B^{1}_{t,z}$ in red $B^{2}_{t,z}$ in
green and $B^{\ast}_{t,z}$ in blue}}}
\end{figure}

The corollaries \ref{corUnivParb} and \ref{corNoArbBeta} show the
difference
between the sets $B^{i}$ and $B^{\ast }$. Restricted to the time interval $%
[t-1,t]$, a probability measure whose mass is concentrated on
$B^{\ast }$ admits an equivalent martingale measure while for
those probabilities that assign positive mass to at least one
$B^{i}$ an arbitrage opportunity can be constructed. We can
summarize the possible situations as follows.

\begin{corollary}
\label{resume}For $\Gamma =A_{t-1}^{z},$ with $A\in
\mathcal{F}_{t-1}$, \ and $\mathcal{M}(B)$ defined in Corollary
\ref{corNoArbBeta} we have:

\begin{enumerate}
\item \label{beta0case}$B_{t,z}^{\ast }=A_{t-1}^{z}\Leftrightarrow $ No $1p$%
-Arbitrage on $A_{t-1}^{z}\overset{\cdot }{\Leftrightarrow }0\in
ri(\Delta S_t(A_{t-1}^{z}))^{cc}.$

\item \label{betacase}$B_{t,z}^{\ast }=\varnothing \Leftrightarrow
0\notin conv\big(\Delta S_{t}(A_{t-1}^{z})\big)$

\item \label{beta1case}$\beta _{t,z}=1$ and $B_{t,z}^{\ast }\neq
\varnothing
\Longrightarrow \exists H\in \mathbb{R}^{d}\setminus \{0\}$ s.t. $%
B_{t,z}^{\ast }=\{\omega \in A_{t-1}^{z}\mid H(S_{t}(\omega
)-z_{t-1})=0\}$ is \textquotedblleft
martingalizable\textquotedblright\ i.e. $\forall
U\subset B$, $B\in \mathcal{F}$ there exists $Q\in \mathcal{M}(A_{t-1}^{z})%
\text{ s.t. }Q(U)>0$.
\end{enumerate}
\end{corollary}

\begin{proof}
Equivalence \ref{beta0case}. immediately follows from Corollary
\ref{NoPoint} and Definition \ref{defBeta0}.\newline To show
\ref{betacase}. we use the sets $K^{i}$ for $i=1,\ldots ,\beta
_{t,z} $ and the other notations from the proof of Lemma
\ref{spezzamento}. Suppose first that $0\notin conv\big(\Delta
S_{t}(\Gamma )\big)$ which implies $0\notin ri(\Delta S_{t}(\Gamma
))^{cc}$ and $\Delta _{0}=\varnothing $. From the assumption we
have $0\notin conv\big(\Delta
S_{t}(C)\big)$ for any subset $C\subseteq \Gamma $ so, in particular, $%
0\notin ri(K^{i})$ unless $K^{i}=\left\{ 0\right\} $. This implies $%
B_{t,z}^{\ast }=\Delta _{0}=\varnothing $.\newline Suppose now
$0\in conv\big(\Delta S_{t}(\Gamma )\big)$. If $0\in ri(\Delta
S_{t}(\Gamma ))^{cc}$, by Definition \ref{defBeta0} we have
$B_{t,z}^{\ast }=\Gamma \neq \varnothing $. Suppose then $0\notin
ri(\Delta S_{t}(\Gamma ))^{cc}$. As $0\in conv\big(\Delta
S_{t}(\Gamma )\big)$ there exists $n\geq 1 $ such that:
$0=\sum_{j=1}^{n}\lambda _{j}(S_{t}(\omega _{j})-z_{t-1})$, with
$\sum_{j=1}^{n}\lambda _{j}=1$, $\lambda _{j}>0$ and $\omega
_{j}\in \Gamma $ for all $j$. If $0$ is extremal then $n=1,$ $\
S_{t}(\omega _{1})-z_{t-1}=0$ and $\{\omega _{1}\}\in \Delta
_{0}\subseteq B_{t,z}^{\ast } $. If $n\geq 2$ we have $0\in
conv\left( \Delta S_{t}\{\omega _{1},\ldots ,\omega _{n}\}\right)
$ so that for any $H\in \mathbb{R}^{d}$ that satisfies
$H\cdot \Delta S_{t}(\omega _{i})\geq 0$ for any $i=1,\ldots ,n$ we have $%
H\cdot \Delta S_{t}(\omega _{i})=0$. Hence $\{\omega _{1},\ldots
,\omega _{n}\}\subseteq B_{t,z}^{\ast }$ by definition of
$B_{t,z}^{\ast }$. \newline
We conclude by showing \ref{beta1case}. From Lemma \ref{spezzamento} items %
\ref{item_spezz_arb} and \ref{item_spezz}, if we select $H=H^{1}$ then $%
\{\omega \in \Gamma \mid H^{1}(S_{t}(\omega )-z_{t-1})=0\}=\Gamma
\setminus B_{t,z}^{1}=B_{t,z}^{\ast }\neq \varnothing $ and on
$B_{t,z}^{\ast }$ we may apply Corollary \ref{corNoArbBeta}.
\end{proof}

\subsection{On $\mathcal{M}$-polar sets}

\label{sectionNull}

We consider for any $t\in I$ the sigma-algebra $\mathfrak{F}%
_{t}:=\bigcap_{Q\in \mathcal{M}}\mathcal{F}_{t}^{Q}$, where $\mathcal{F}%
_{t}^{Q}$ is the $Q$-completion of $\mathcal{F}_{t}$.
$\mathfrak{F}_{t}$ is
the universal completion of $\mathcal{F}_{t}$ with respect to $\mathcal{M}=%
\mathcal{M}(\mathbb{F})$. Notice that the introduction of this
enlarged
filtration needs the knowledge \emph{a priori} of the whole class $\mathcal{M%
}$ of martingale measures. Recall that any measure $Q\in
\mathcal{M}$ can be
uniquely extended to a measure $\overline{Q}$ on the enlarged sigma algebra $%
\mathfrak{F}_{T}$ so that we can write with slight abuse of notation $%
\mathcal{M}(\mathbb{F})=\mathcal{M}(\mathfrak{F})$ where $\mathfrak{F}:=\{%
\mathfrak{F}_{t}\}_{t\in I}$.\newline

We wish to show now that under any martingale measure the sets
$B_{t,z}^{i}$ (and their arbitrary unions) introduced in Lemma
\ref{spezzamento} must be null-sets. To this purpose we need to
recall some properties of a proper regular conditional probability
(see Theorems 1.1.6, 1.1.7 and 1.1.8 in Stroock-Varadhan
\cite{Strook}).

\begin{theorem}
\label{SV}Let $(\Omega ,\mathcal{F},Q)$ be a probability space, where $%
\Omega $ is a Polish space, $\mathcal{F}$ is the Borel $\sigma $-algebra, $%
Q\in \mathcal{P}.$ Let $\mathcal{A}\subseteq \mathcal{F}$ be a
countably generated sub sigma algebra of $\mathcal{F}$. Then there
exists a proper regular conditional probability, i.e. a function
$Q_{\mathcal{A}}(\cdot ,\cdot ):(\Omega ,\mathcal{F})\mapsto
\lbrack 0,1]$ such that:
\end{theorem}

\begin{itemize}
\item[a)] for all $\omega \in \Omega $, $Q_{\mathcal{A}}(\omega
,\cdot )$ is a probability measure on $\mathcal{F}$;

\item[b)] for each $B\in \mathcal{F}$, the function
$Q_{\mathcal{A}}(\cdot ,B)$ is a version of $Q(B\mid
\mathcal{A})(\cdot )$;

\item[c)] $\exists N\in \mathcal{A}$ with $Q(N)=0$ such that$\ Q_{\mathcal{A}%
}(\omega ,B)=1_{B}(\omega )$ for $\omega \in \Omega \setminus N$
and $B\in \mathcal{A}$

\item[d)] In addition, if $X\in L^{1}(\Omega ,\mathcal{F},Q)$ then
$\
E_{Q}[X\mid \mathcal{A}](\omega )=\int_{\Omega }X(\tilde{\omega})Q_{\mathcal{%
A}}(\omega ,d\tilde{\omega})$ $Q-a.s.$
\end{itemize}

Recall that $\mathcal{F}_{t}=\mathcal{F}_{t}^{S}$, $t\in I$, is
countably generated.

\begin{lemma}
\label{conditional}Fix $t\in I_{1}=\left\{ 1,...,T\right\} $, $A\in \mathcal{%
F}_{t-1}$, $Q\in \mathcal{M}$ and for $z\in \mathbf{Z}$ consider $%
A_{t-1}^{z}:=\{\omega \in A\mid S_{0:t-1}(\omega )=z_{0:t-1}\}$.
Then
\begin{equation*}
\bigcup_{z\in \mathbf{Z}}\{\omega \in A_{t-1}^{z}\text{ s.t. }Q_{\mathcal{F}%
_{t-1}}(\omega ,\cup _{i=1}^{\beta _{t,z}}B_{t,z}^{i})>0\}
\end{equation*}%
is a subset of an $\mathcal{F}_{t-1}$-measurable $Q$-null set.
\end{lemma}

\begin{proof}
If $Q(A)=0$ there is nothing to show. Suppose now $Q(A)>0.$ In
this proof we
set for the sake of simplicity $X:=S_{t}$, $Y:=E_{Q}[X\mid \mathcal{F}%
_{t-1}]=S_{t-1}$ $Q$-a.s. $\beta :=\beta _{t,z}$ and $\mathcal{A}:=\mathcal{F%
}_{t-1}=\mathcal{F}_{t-1}^{S}$. Set
\begin{equation*}
D_{t}^{z}:=\left\{ \omega \in A_{t-1}^{z}\text{ such that }Q_{\mathcal{A}%
}(\omega ,\cup _{i=1}^{\beta }B_{t,z}^{i})>0\right\} .
\end{equation*}%
If $z\in \mathbf{Z}$ is such that $0\in ri(\Delta
S_{t}(A_{t-1}^{z}))^{cc}$
then $\cup _{i=1}^{\beta }B_{t,z}^{i}=\varnothing $ and $D_{t}^{z}=%
\varnothing $. So we can consider only those $z\in \mathbf{Z}$ such that $%
0\notin ri(\Delta S_{t}(A_{t-1}^{z}))^{cc}.$ Fix such $z.$\newline
Since $\mathcal{A}=\mathcal{F}_{t-1}^{S}$ is countably generated,
$Q$ admits a proper regular conditional probability
$Q_{\mathcal{A}}$.\newline From Theorem \ref{SV} d) we obtain:
\begin{equation*}
Y(\omega )=\int_{\Omega }X(\tilde{\omega})Q_{\mathcal{A}}(\omega ,d\tilde{%
\omega})\quad Q-a.s.
\end{equation*}%
As $A_{t-1}^{z}\in \mathcal{A}$, by Theorem \ref{SV} c) there exists a set $%
N\in \mathcal{A}$ with $Q(N)=0$ so that $Q_{\mathcal{A}}(\omega
,A_{t-1}^{z})=1$ on $A_{t-1}^{z}\setminus N$ and therefore we have
\begin{equation}
\int_{\Omega }X(\tilde{\omega})Q_{\mathcal{A}}(\omega ,d\tilde{\omega}%
)=\int_{A_{t-1}^{z}}X(\tilde{\omega})Q_{\mathcal{A}}(\omega ,d\tilde{\omega}%
)\quad \forall \omega \in A_{t-1}^{z}\setminus N.  \label{XY}
\end{equation}%
Since $0\notin ri(\Delta S_{t}(\Gamma ))^{cc}$ we may apply Lemma \ref%
{spezzamento}: for any $i=1,\ldots ,\beta $, there exists $H^{i}\in \mathbb{R%
}^{d}$ such that $H^{i}\cdot (X(\tilde{\omega})-z_{t-1})\geq 0$ for all $%
\tilde{\omega}\in \cup _{l=i}^{\beta }B_{t,z}^{l}\cup B_{t,z}^{\ast }$ and $%
H^{i}\cdot (X(\tilde{\omega})-z_{t-1})>0$ for every
$\tilde{\omega}\in B_{t,z}^{i}$.

Now we fix $\omega \in D_{t}^{z}\setminus N\subseteq
A_{t-1}^{z}\setminus N.$ Then the index $j:=\min \{1\leq i\leq
\beta \mid Q_{\mathcal{A}}(\omega
,B_{t,z}^{i})>0\}$ is well defined and: i) $H^{j}\cdot (X(\tilde{\omega}%
)-z_{t-1})>0$ on $B_{t,z}^{j}$, (ii) $Q_{\mathcal{A}}(\omega
,B_{t,z}^{j})>0$ iii) $H^{j}\cdot (X(\tilde{\omega})-z_{t-1})\geq
0$ on $\cup _{l=j}^{\beta }B_{t,z}^{l}\cup B_{t,z}^{\ast }$; iv)
$Q_{\mathcal{A}}(\omega ,B_{t,z}^{i})=0$ for $i<j.$ From i) and
ii) we obtain
\begin{equation*}
Q_{\mathcal{A}}(\omega ,A_{t-1}^{z}\cap \{H^{j}\cdot (X-z_{t-1})>0\})\geq Q_{%
\mathcal{A}}(\omega ,B_{t,z}^{j})>0.
\end{equation*}%
From iii) and iv) we obtain:%
\begin{eqnarray*}
Q_{\mathcal{A}}(\omega ,\{H^{j}\cdot (X-z_{t-1})\geq 0\}) &\geq &Q_{\mathcal{%
A}}(\omega ,\cup _{l=j}^{\beta }B_{t,z}^{l}\cup B_{t,z}^{\ast }) \\
&\geq &Q_{\mathcal{A}}(\omega ,A_{t-1}^{z})-Q_{\mathcal{A}}(\omega
,\cup _{i<j}B_{t,z}^{i})=1.
\end{eqnarray*}%
Hence
\begin{equation*}
H^{j}\cdot \left(
\int_{A_{t-1}^{z}}X(\tilde{\omega})Q_{\mathcal{A}}(\omega
,d\tilde{\omega})-z_{t-1}\right) =\int_{A_{t-1}^{z}}H^{j}\cdot (X(\tilde{%
\omega})-z_{t-1})Q_{\mathcal{A}}(\omega ,d\tilde{\omega})>0
\end{equation*}%
and therefore, from equation (\ref{XY}) and from $z_{t-1}=Y(\omega
)$, we
have:%
\begin{equation*}
H^{j}\cdot \left( \int_{\Omega }X(\tilde{\omega})Q_{\mathcal{A}}(\omega ,d%
\tilde{\omega})-Y(\omega )\right) >0.
\end{equation*}%
As this holds for any $\omega \in D_{t}^{z}\setminus N$ we obtain:
\begin{equation*}
D_{t}^{z}\setminus N\subseteq \{\omega \in \Omega \mid Y(\omega
)\neq
\int_{\Omega }X(\tilde{\omega})Q_{\mathcal{A}}(\omega ,d\tilde{\omega}%
)\}=:N^{\ast }\in \mathcal{F}_{t-1}
\end{equation*}%
with $Q(N^{\ast })=0.$ Hence, $D_{t}^{z}\subseteq N\cup N^{\ast
}:=N_{0}$
with $Q(N_{0})=0$ and $N_{0}$ not dependent on $z$. As this holds for every $%
z\in \mathbf{Z}$ we conclude that $\bigcup_{z\in \mathbf{Z}%
}D_{t}^{z}\subseteq N_{0}.$
\end{proof}

\begin{corollary}
\label{B-meas}Fix $t\in I_{1}$ and $Q\in \mathcal{M}$. 
If
\begin{equation*}
\mathfrak{B}_{t}:=\bigcup_{z\in \mathbf{Z}}\left\{
\bigcup_{i=1}^{\beta _{t,z}}B_{t,z}^{i}\right\}
\end{equation*}%
for $B_{t,z}^{i}$ given in Lemma \ref{spezzamento} with $\Gamma
=\Sigma
_{t-1}^{z}$ or $\Gamma =A_{t-1}^{z}$ (defined in equations \eqref{Sigma} and %
\eqref{A}), then $\mathfrak{B}_{t}$ is a subset of an $\mathcal{F}_{t}$%
-measurable $Q$ null set.
\end{corollary}

\begin{proof}
First we consider the case $\Gamma =\Sigma _{t-1}^{z}$ and
$B_{t,z}^{i}$ given in Lemma \ref{spezzamento} with $\Gamma
=\Sigma _{t-1}^{z}$. As in the
previous proof, we denote the sigma-algebra $\mathcal{F}_{t-1}$ with $%
\mathcal{A}:=\mathcal{F}_{t-1}$. Notice that if $z\in \mathbf{Z}$
is such that $0\in ri(\Delta S_{t}(\Gamma ))^{cc}$ then $\cup
_{i=1}^{\beta _{t,z}}B_{t,z}^{i}=\varnothing ,$ hence we may
assume that $0\notin ri(\Delta S_{t}(\Gamma ))^{cc}$ . From the
proof of Lemma \ref{conditional}
\begin{equation*}
\bigcup_{z\in \mathbf{Z}}D_{t}^{z}\subseteq N_{0}=N\cup N^{\ast }
\end{equation*}%
with $Q(N_{0})=0.$ Notice that if $\omega \in \Omega \setminus
N_{0}$ then,
for all $z\in \mathbf{Z},$ either $\omega \notin \Sigma _{t-1}^{z}$ or $Q_{%
\mathcal{A}}(\omega ,\cup _{i=1}^{\beta _{t,z}}B_{t,z}^{i})=0$. Hence $%
\omega \in \Sigma _{t-1}^{z}\setminus N_{0}$ implies
$Q_{\mathcal{A}}(\omega
,\cup _{i=1}^{\beta _{t,z}}B_{t,z}^{i})=0.$ By Theorem \ref{SV} c) we have $%
Q_{\mathcal{A}}(\omega ,(\Sigma _{t-1}^{z})^{c})=0$ for all
$\omega \in \Sigma _{t-1}^{z}\setminus N_{0}$.

Fix now $\omega \in \Sigma _{t-1}^{z}\setminus N_{0}$ and consider
the
completion $\mathcal{F}_{t}^{Q_{\mathcal{A}}(\omega ,\cdot )}$ of $\mathcal{F%
}_{t}$ and the unique extension on
$\mathcal{F}_{t}^{Q_{\mathcal{A}}(\omega
,\cdot )}$ of $Q_{\mathcal{A}}(\omega ,\cdot )$, which we denote with $%
\widehat{Q}_{\mathcal{A}}(\omega ,\cdot ):\mathcal{F}_{t}^{Q_{\mathcal{A}%
}(\omega ,\cdot )}\rightarrow \lbrack 0,1]$.

From $\mathfrak{B}_{t}\cap (\Sigma _{t-1}^{z})^{c}\subseteq
(\Sigma _{t-1}^{z})^{c}$ and $Q_{\mathcal{A}}(\omega ,(\Sigma
_{t-1}^{z})^{c})=0$ we
deduce: $\mathfrak{B}_{t}\cap (\Sigma _{t-1}^{z})^{c}\in \mathcal{F}_{t}^{Q_{%
\mathcal{A}}(\omega ,\cdot )}$ and $\widehat{Q}_{\mathcal{A}}(\omega ,%
\mathfrak{B}_{t}\cap (\Sigma _{t-1}^{z})^{c})=0$. From
$\mathfrak{B}_{t}\cap
\Sigma _{t-1}^{z}=\cup _{i=1}^{\beta _{t,z}}B_{t,z}^{i}$ and $Q_{\mathcal{A}%
}(\omega ,\cup _{i=1}^{\beta _{t,z}}B_{t,z}^{i})=0$ we deduce: $\mathfrak{B}%
_{t}\cap \Sigma _{t-1}^{z}\in
\mathcal{F}_{t}^{Q_{\mathcal{A}}(\omega ,\cdot )}$ and
$\widehat{Q}_{\mathcal{A}}(\omega ,\mathfrak{B}_{t}\cap \Sigma
_{t-1}^{z})=0$. Then $\mathfrak{B}_{t}=(\mathfrak{B}_{t}\cap
\Sigma
_{t-1}^{z})\cup (\mathfrak{B}_{t}\cap (\Sigma _{t-1}^{z})^{c})\in \mathcal{F}%
_{t}^{Q_{\mathcal{A}}(\omega ,\cdot )}$ and
$\widehat{Q}_{\mathcal{A}}\left( \omega ,\mathfrak{B}_{t}\right)
=0$. Since $\omega \in \Sigma
_{t-1}^{z}\setminus N_{0}$ was arbitrary, we showed that $\widehat{Q}_{%
\mathcal{A}}(\omega ,\mathfrak{B}_{t})=0$ for all $\omega \in
\Sigma _{t-1}^{z}\setminus N_{0}$ and all $z\in \mathbf{Z}$. Since
$\bigcup_{z\in \mathbf{Z}}(\Sigma _{t-1}^{z}\setminus
N_{0})=\Omega \setminus N_{0}$ we
have:%
\begin{equation}
\mathfrak{B}_{t}\in \mathcal{F}_{t}^{Q_{\mathcal{A}}(\omega ,\cdot
)}\text{
and }\widehat{Q}_{\mathcal{A}}(\omega ,\mathfrak{B}_{t})=0\text{ for all }%
\omega \in \Omega \setminus N_{0}\text{ with }\ Q(N_{0})=0.
\label{111}
\end{equation}

Now consider the sigma algebra
\begin{equation*}
\widehat{\mathcal{F}}_{t}=\bigcap_{\omega \in \Omega \setminus N_{0}}%
\mathcal{F}_{t}^{Q_{\mathcal{A}}(\omega ,\cdot )}
\end{equation*}%
and observe that $\mathfrak{B}_{t}\in \widehat{\mathcal{F}}_{t}$.
Notice
that if a subset $B\subseteq \Omega $ satisfies: $B\subseteq C$ for some $%
C\in \mathcal{F}_{t}$ with $Q_{\mathcal{A}}(\omega ,C)=0$ for all
$\omega
\in \Omega \setminus N_{0},$ then%
\begin{equation*}
Q(C)=\int_{\Omega }Q_{\mathcal{A}}(\omega ,C)Q(d\omega
)=\int_{\Omega \setminus N_{0}}Q_{\mathcal{A}}(\omega ,C)Q(d\omega
)=0,
\end{equation*}%
so that $B\in \mathcal{F}_{t}^{Q}$. This shows that $\mathcal{F}%
_{t}\subseteq \widehat{\mathcal{F}}_{t}\subseteq
\mathcal{F}_{t}^{Q}$. Hence
$\mathfrak{B}_{t}\in \mathcal{F}_{t}^{Q}$. Let $\widehat{Q}:\widehat{%
\mathcal{F}}_{t}\rightarrow \lbrack 0,1]$ be defined by
$\widehat{Q}(\cdot ):=\int_{\Omega
}\widehat{Q}_{\mathcal{A}}(\omega ,\cdot )Q(d\omega ).$ Then
$\widehat{Q}$ is a probability which satisfies
$\widehat{Q}(B)=Q(B)$ for
every $B\in \mathcal{F}_{t}$ and therefore is an extension on $\widehat{%
\mathcal{F}}_{t}$ of $Q$. Since
$\overline{Q}:\mathcal{F}_{t}^{Q}\rightarrow
\lbrack 0,1]$ is the unique extension on $\mathcal{F}_{t}^{Q}$ of $Q$ and $%
\mathcal{F}_{t}\subseteq \widehat{\mathcal{F}}_{t}\subseteq \mathcal{F}%
_{t}^{Q}$ then $\widehat{Q}$ is the restriction of $\overline{Q}$ on $%
\widehat{\mathcal{F}}_{t}$ and%
\begin{equation*}
\overline{Q}(\mathfrak{B}_{t})=\widehat{Q}(\mathfrak{B}_{t})=\int_{\Omega }%
\widehat{Q}_{\mathcal{A}}(\omega ,\mathfrak{B}_{t})Q(d\omega
)=\int_{\Omega \setminus N_{0}}\widehat{Q}_{\mathcal{A}}(\omega
,\mathfrak{B}_{t})Q(d\omega )=0.
\end{equation*}%
Suppose now $A\in \mathcal{F}_{t-1},$ $\Gamma =A_{t-1}^{z}$ and set $%
\mathfrak{C}_{t}:=\bigcup_{z\in \mathbf{Z}}\left\{ \cup
_{i=1}^{\beta
_{t,z}}B_{t,z}^{i}\right\} $ where $B_{t,z}^{i}$ is given in Lemma \ref%
{spezzamento} with $\Gamma =A_{t-1}^{z}$. Fix any $\omega \in A$. Then $%
\Sigma _{t}^{S_{0:T}(\omega )}\subseteq A$ since $A\in
\mathcal{F}_{t-1}$. As a consequence $\mathfrak{C}_{t}\subseteq
\mathfrak{B}_{t}$.
\end{proof}

\begin{corollary}
\label{conditionalPolar}Fix $t\in I_{1}=\left\{ 1,...,T\right\} $
and for any $A\in \mathcal{F}_{t-1}$ consider
$A_{t-1}^{z}=\{\omega \in A\mid S_{0:t-1}(\omega )=z_{0:t-1}\}\neq
\varnothing $. Then for any $Q\in \mathcal{M}$ the set $\bigcup
\{A_{t-1}^{z}\mid 0\notin conv\big(\Delta
S_{t}(A_{t-1}^{z})\big)\}$ is a subset of an $\mathcal{F}_{t-1}$-measurable $%
Q$-null set and as a consequence is an $\mathcal{M}$-polar set.
\end{corollary}

\begin{proof}
From Corollary \ref{resume} 2), the condition $0\notin
conv\big(\Delta S_{t}(A_{t-1}^{z})\big)$ implies that $\cup
_{i=1}^{\beta _{t,z}}B_{t,z}^{i}=A_{t-1}^{z}$, thus:
$Q_{\mathcal{A}}(\omega ,A_{t-1}^{z})=1$ on $A_{t-1}^{z}\setminus
N$, $D_{t}^{z}=\{\omega \in A_{t-1}^{z}$ s.t.
$Q_{\mathcal{F}_{t-1}}(\omega ,A_{t-1}^{z})>0\}\supseteq
A_{t-1}^{z}\setminus N$ and
\begin{equation*}
\left( \bigcup \{A_{t-1}^{z}\mid 0\notin conv\big(\Delta S_{t}(A_{t-1}^{z})%
\big)\}\setminus N\right) \subseteq \bigcup_{z\in \mathbf{Z}%
}D_{t}^{z}\subseteq N_{0}\in \mathcal{F}_{t-1}.
\end{equation*}
\end{proof}


\subsubsection{Backward effect in the multiperiod case}

\label{SecSvuotamento}

The following example shows that additional care is required in
the multi-period setting:

\begin{example}
\label{svuotamento} Let $\Omega =\{\omega _{1},\omega _{2},\omega
_{3},\omega _{4}\}$ and consider a single risky asset $S_{t}$ with $t=0,1,2$%
.
\begin{equation*}
S_{0}=7\qquad S_{1}(\omega )=%
\begin{cases}
8\quad \omega \in \{\omega _{1},\omega _{2}\} \\
3\quad \omega \in \{\omega _{3},\omega _{4}\}\qquad%
\end{cases}%
S_{2}(\omega )=%
\begin{cases}
9\quad \omega =\omega _{1} \\
6\quad \omega =\omega _{2} \\
5\quad \omega =\omega _{3} \\
4\quad \omega =\omega _{4}%
\end{cases}%
\end{equation*}%
Fix $z\in \mathbf{Z}$ with the first two components
$(z_{0},z_{1})$ equal to $(7,3).$

\noindent \textbf{First period:} $\Sigma _{0}^{z}=\Omega $ and
$0\in ri(conv(\Delta S_{1}(\Sigma _{0}^{z})))=(-4,1)$ and there
exists $Q_{1}$
such that $Q_{1}(\omega _{i})>0$ for $i=1,2,3,4$ and $S_{0}=E_{Q_{1}}[S_{1}]$%
. If we restrict the problem to the first period only, there
exists a full support martingale measure for $(S_{0},S_{1})$ and
there are no $\mathcal{M}$-polar sets.
\newline
\textbf{Second period:} $\Sigma _{1}^{z}=\{\omega _{3},\omega _{4}\}$, $%
0\notin conv(\Delta S_{2}(\Sigma _{1}^{z}))=[1,2]$ and hence
$\Sigma _{1}^{z} $ is not supported by any martingale measure for
$S$, i.e. if $Q\in \mathcal{M}$ then $Q(\{\omega _{3},\omega
_{4}\})=0.$\newline \textbf{Backward}: As $\{\omega _{3},\omega
_{4}\}$ is a $Q$ null set for any martingale measure $Q\in
\mathcal{M}$, then $Q(\{\omega _{1},\omega _{2}\})=1$. This
reflects into the first period as $0\notin conv(\Delta
S_{1}(\{\omega _{1},\omega _{2}\}))=\left\{ 1\right\} $ and we
deduce that also $\{\omega _{1},\omega _{2}\}$ is not supported by
any martingale measure, implying $\mathcal{M}=\varnothing $.

This example thus shows that new $\mathcal{M}$-polar sets (as $\{\omega _{3},\omega _{4}\}$%
) can arise at later times creating a backward effect on the
existence martingale measures. In order to detect these situations
at time $t$, we shall need to anticipate certain polar sets at
posterior times.
\end{example}

More formally we need to consider the following iterative
procedure. Let
\begin{eqnarray*}
\Omega _{T}:= &&\Omega \\
\Omega _{t-1}:= &&\Omega _{t}\setminus \bigcup_{z\in
\mathbf{Z}}\{\Sigma
_{t-1}^{z}\mid 0\notin conv\big(\Delta S_{t}(\widetilde{\Sigma }_{t-1}^{z})%
\big)\},\text{\quad }t\in I_{1},
\end{eqnarray*}%
where%
\begin{equation*}
\widetilde{\Sigma }_{t-1}^{z}:=\{\omega \in \Omega _{t}\mid
S_{0:t-1}=z_{0:t-1}\},\text{\quad }t\in I_{1}.
\end{equation*}

We show that the set $B_{t,z}^{i}$ obtained from Lemma
\ref{spezzamento} with $\Gamma =\widetilde{\Sigma }_{t}^{z}$
belong to the family of polar set of $\mathcal{M}(\mathbb{F})$:
\begin{equation*}
\mathcal{N}:=\left\{ A\subseteq A^{\prime }\in \mathcal{F}\ \mid \
Q(A^{\prime })=0\ \forall \ Q\in \mathcal{M}(\mathbb{F})\right\}
\end{equation*}%
More precisely,

\begin{lemma}
\label{propPolari}For all $t\in I_{1}$ and $z\in \mathbf{Z}$
consider the
sets $B_{t,z}^{i}$ from Lemma \ref{spezzamento} with $\Gamma =\widetilde{%
\Sigma }_{t-1}^{z}$. Let
\begin{equation*}
\widetilde{\mathfrak{B}}_{t}:=\bigcup_{z\in \mathbf{Z}}\left\{
\cup
_{i=1}^{\beta _{t,z}}B_{t,z}^{i}\right\} \quad \mathfrak{D}%
_{t-1}:=\bigcup_{z\in \mathbf{Z}}\left\{ \Sigma _{t-1}^{z}\mid
0\notin conv(\Delta S_{t}(\widetilde{\Sigma }_{t-1}^{z}))\right\}
\end{equation*}%
For any $Q\in \mathcal{M}$, $\widetilde{\mathfrak{B}}_{t}$ is a subset of a $%
\mathcal{F}_{t}$-measurable $Q$-null set and
$\mathfrak{D}_{t-1}$ is a subset of an $\mathcal{F}_{t-1}$-measurable $Q$%
-null set.
\end{lemma}

\begin{proof}
We prove this by backward induction. For $t=T$ the assertion is
true from Corollary \ref{B-meas} and Corollary
\ref{conditionalPolar}. Suppose now the claim holds true
for any $k+1\leq t\leq T$. From the
inductive hypothesis there exists $N_{k}^{Q}\in \mathcal{F}_{k}$ such that $%
\mathfrak{D}_{k}\subseteq N_{k}^{Q}$ with $Q(N_{k}^{Q})=0$.
Introduce the
auxiliary $\mathcal{F}_{k}$-measurable random variable
\begin{equation}\label{Q_modification2}
X_{k}^{Q}:=S_{k-1}%
\mathbf{1}_{N_{k}^{Q}}+S_{k}\mathbf{1}_{(N_{k}^{Q})^{c}}
\end{equation} and notice that $%
E_{Q}[X_{k}^{Q}\mid \mathcal{F}_{k-1}]=S_{k-1}$ $Q$-a.s. Moreover from $%
\Delta X_{k}^{Q}:=X_{k}^{Q}-S_{k-1}=0$ on $N_{k}^{Q}$ and $\Omega
\setminus N_{k}^{Q}\subseteq \Omega \setminus \mathfrak{D}_{k}$,
we can deduce that
\begin{equation}
0\notin ri(\Delta S_{k}(\widetilde{\Sigma
}_{k-1}^{z}))^{cc}\Longrightarrow 0\notin ri(\Delta
X_{k}^{Q}(\Sigma _{k-1}^{z}))^{cc}  \label{ri}
\end{equation}%
which implies $\widetilde{\mathfrak{B}}_{k}\subseteq \mathfrak{B}%
_{k}(X_{k}^{Q})\cup N_{k}^{Q}$ where we denote
$\mathfrak{B}_{k}(X_{k}^{Q})$ the set obtained from Corollary
\ref{B-meas} with $\Gamma =\Sigma _{k-1}^{z}$ and $X_{k}^{Q}$
which replaces $S_{k}$. According to Corollary \ref{B-meas}
we find $M_{k}^{Q}\in \mathcal{F}_{k}$ with $Q(M_{k}^{Q})=0$ such that $%
\widetilde{\mathfrak{B}}_{k}\subseteq
\mathfrak{B}_{k}(X_{k}^{Q})\cup N_{k}^{Q}\subseteq M_{k}^{Q}\cup
N_{k}^{Q}$. Since $Q$ is arbitrary we have the thesis.\\
We now show the second assertion.\\ For every $Q\in
\mathcal{M}$ and $\underline{\varepsilon }=(\varepsilon
,...,\varepsilon )\in \mathbb{R}^{d}$ with $\varepsilon >0$ we can
define
\begin{equation}  \label{Q_modification}
S_{k}^{Q}=(S_{k-1}+\underline{\varepsilon
})\mathbf{1}_{N_{k}^{Q}\cup
M_{k}^{Q}}+S_{k}\mathbf{1}_{(N_{k}^{Q}\cup M_{k}^{Q})^{c}}
\end{equation}%
and $E_{Q}[S_{k}^{Q}\mid \mathcal{F}_{k-1}]=S_{k-1}$. With $\Delta
S_{k}^{Q}:=S_{k}^{Q}-S_{k-1}$ we claim
\begin{equation}
\mathfrak{D}_{k-1}\subseteq \bigcup_{z\in \mathbf{Z}}\{\Sigma
_{k-1}^{z}\mid 0\notin conv(\Delta S_{k}^{Q}(\Sigma
_{k-1}^{z}))\}.  \label{secondak}
\end{equation}%
Indeed let $z\in \mathbf{Z}$ such that $\Sigma _{k-1}^{z}\subseteq \mathfrak{%
D}_{k-1}$ and observe that
\begin{equation}
0\notin conv(\Delta S_{k}(\widetilde{\Sigma
}_{k-1}^{z}))\Leftrightarrow 0\notin conv(\Delta S_{k}(\Sigma
_{k-1}^{z}\setminus \mathfrak{D}_{k})). \label{convSvuotato}
\end{equation}

\noindent Since $\Sigma _{k-1}^{z}\setminus N_{k}^{Q}\subseteq
\Sigma
_{k-1}^{z}\setminus \mathfrak{D}_{k}\subseteq \widetilde{\mathfrak{B}}%
_{k}\subseteq N_{k}^{Q}\cup M_{k}^{Q}$, then
\begin{eqnarray*}
\Sigma _{k-1}^{z} &=&(\Sigma _{k-1}^{z}\cap N_{k}^{Q})\cup (\Sigma
_{k-1}^{z}\setminus N_{k}^{Q})\subseteq N_{k}^{Q}\cup M_{k}^{Q} \\
&\subseteq &\bigcup_{z\in \mathbf{Z}}\{\Sigma _{k-1}^{z}\mid
0\notin conv(\Delta S_{k}^{Q}(\Sigma _{k-1}^{z}))\}
\end{eqnarray*}%
for any $\Sigma _{k-1}^{z}\subseteq \mathfrak{D}_{k-1}$. Hence the
claim since $\bigcup_{z}\{\Sigma _{k-1}^{z}\mid 0\notin
conv(\Delta
S_{k}^{Q}(\Sigma _{k-1}^{z}))\}$ is a subset of an $\mathcal{F}_{k-1}$%
-measurable $Q$-null set.
\end{proof}

\subsection{On the maximal $\mathcal{M}$-polar set and the support of
martingale measures}

\label{maxPolar}

\noindent The sets introduced in Sections \ref{mktFeas} and \ref%
{SecSvuotamento} provide a geometric decomposition of $\Omega $ in
two parts, $\Omega =\Omega _{\ast }\cup (\Omega _{\ast })^{c}$
specified in Proposition \ref{LemNOpolar} below. The set $\Omega
_{\ast }$ contains those events $\omega $ supported by martingale
measures, namely, for any of those events it is possible to
construct a martingale measure (even with finite support) that
assign positive probability to $\omega $. Observe that such a
decompostion is induced by $S$ and it is determined prior to
arbitrage considerations.

\begin{proposition}
\label{LemNOpolar} Let $\{\Omega _{t}\}_{t\in I}$ as defined in Section \ref%
{sectionFiltration} and, for any $z\in \mathbf{Z,}$ let $\beta _{t,z}$ and $%
B_{t,z}^{\ast }$ be the index $\beta $ and the set $B^{\ast }$ from Lemma %
\ref{spezzamento} with $\Gamma =\widetilde{\Sigma }_{t-1}^{z}$.
Define
\begin{equation*}
\Omega _{\ast }:=\bigcap_{t=1}^{T}\left( \bigcup_{z\in \mathbf{Z}%
}B_{t,z}^{\ast }\right) .
\end{equation*}%
We have the following
\begin{equation*}
\mathcal{M}\neq \varnothing \Longleftrightarrow \Omega _{\ast
}\neq \varnothing \Longleftrightarrow \mathcal{M\cap P}_{f}\neq
\varnothing ,
\end{equation*}%
where
\begin{equation*}
\mathcal{P}_{f}:=\left\{ P\in \mathcal{P}\mid \text{supp(}P\text{) is finite}%
\right\}
\end{equation*}%
is the set of probability measures whose support is a finite number of $%
\omega \in \Omega $.

If $\mathcal{M}\neq \varnothing $ then for any $\omega _{\ast }\in
\Omega _{\ast }$ there exists $Q\in \mathcal{M}$ such that
$Q(\{\omega _{\ast }\})>0 $, so that $(\Omega _{\ast })^{c}$ is
the maximal $\mathcal{M}$-polar set, i.e. $(\Omega _{\ast })^{c}$
is an $\mathcal{M}$-polar set and
\begin{equation}
\forall N\in \mathcal{N}\text{ we have }N\subseteq (\Omega _{\ast
})^{c}. \label{decomp2}
\end{equation}
\end{proposition}

\begin{proof}
Observe first that:
\begin{equation*}
(\Omega _{\ast
})^{c}=\bigcup_{t=1}^{T}\widetilde{\mathfrak{B}}_{t}.
\end{equation*}%
From Lemma \ref{propPolari}, $\widetilde{\mathfrak{B}}_{t}$ is an
$\mathcal{M} $-polar set for any $t\in I_{1}$, which implies
$(\Omega _{\ast })^{c}$ is an $\mathcal{M}$-polar set. Suppose now
that $\Omega _{\ast }=\varnothing $ so that $\Omega
=\bigcup_{t=1}^{T}\widetilde{\mathfrak{B}}_{t}$ is a polar set. We
can conclude that $\mathcal{M}=\varnothing $.\newline
Suppose now that $\Omega _{\ast }\neq \varnothing $. We show that for every $%
\omega _{\ast }\in \Omega _{\ast }$ there exists a $Q\in
\mathcal{M}$ such that $Q(\{\omega _{\ast }\})>0$.
%
Observe now that for any $t\in I_{1}$ and for any $\omega \in
\Omega _{\ast } $, $0\in ri(\Delta S_{t}(B_{t,z}^{\ast }))^{cc}$
with $z=S_{0:T}(\omega )$.
As we did in Corollary \ref{corNoArbBeta}, we apply Remark \ref%
{ImportantRemark} and conclude that there exists a finite number
of elements of $B_{t,z}^{\ast },$ named $C_{t}(\omega ):=\{\omega
,\omega _{1},\ldots ,\omega _{m}\}\subseteq B_{t,z}^{\ast }$, such
that
\begin{equation}
S_{t-1}(\omega )=\lambda _{t}(\omega )S_{t}(\omega
)+\sum_{j=1}^{m}\lambda _{t}(\omega _{j})S_{t}(\omega _{j})
\label{first}
\end{equation}%
where $\lambda _{t}(\omega )>0$ and $\lambda _{t}(\omega
)+\sum_{j=1}^{m}\lambda _{t}(\omega _{j})=1$.\newline
Fix now $\omega _{\ast }\in \Omega _{\ast }$. We iteratively build a set $%
\Omega _{f}^{T}$ which is suitable for being the finite support of
a discrete martingale measure (and contains $\omega _{\ast
}$).\newline
Start with $\Omega _{f}^{1}=C_{1}(\omega _{\ast })$ which satisfies (\ref%
{first}) for $t=1$. For any $t>1$, given $\Omega _{f}^{t-1}$,
define $\Omega _{f}^{t}:=\left\{ C_{t}(\omega )\mid \omega \in
\Omega _{f}^{t-1}\right\} $. Once $\Omega _{f}^{T}$ is settled, it
is easy to construct a martingale measure via \eqref{first}:
\begin{equation*}
Q(\{\omega \})=\prod_{t=1}^{T}\lambda _{t}(\omega )\quad \forall
\omega \in \Omega _{f}^{T}
\end{equation*}%
Since, by construction, $\lambda _{t}(\omega _{\ast })>0$ for any
$t\in I_{1} $, we have $Q(\{\omega _{\ast }\})>0$ and $Q\in
\mathcal{M\cap P}_{f}.$

\bigskip

To show \eqref{decomp2} just observe from the previous line that
$\Omega _{\ast }$ is not $\mathcal{M}$-polar, while $(\Omega
_{\ast })^{c}=\bigcup_{t=1}^{T}\widetilde{\mathfrak{B}}_{t}$ is
$\mathcal{M}$-polar thanks to Lemma \ref{propPolari}.
\end{proof}

\begin{proof}[Proof of Proposition \protect\ref{NoPA}]
The absence of $1p$-Arbitrages readly implies that $\Omega _{\ast
}=\Omega $ (see Corollary \ref{resume}). Take a dense subset
$\{\omega _{n}\}_{n\in \mathbb{N}}$ of $\Omega $: from Proposition
\ref{LemNOpolar} for any $\omega _{n}$ there exists a martingale
measure $Q^{n}\in \mathcal{M}$ such that $Q^{n}(\{\omega
_{n}\})>0$. From
Lemma \ref{pasting} in the Appendix $Q:=\sum_{i=1}^{\infty }\frac{1}{2^{i}}%
Q^{i}\in \mathcal{M}$, moreover $Q(\{\omega _{n}\})>0\ \forall n\in \mathbb{N%
}$. Since $\{\omega _{n}\}_{n\in \mathbb{N}}$ is dense, $Q$ is a
full support martingale measure.
\end{proof}


\subsection{Enlarged Filtration and Universal Arbitrage Aggregator}

\label{sectionFiltration}

In Sections \ref{mktFeas} and \ref{sectionNull} we solve the problem of characterizing the
$\mathcal{M}$-polar sets of a certain market model on a fixed time
interval $[t-1,t]$ for $t\in I_{1}=\left\{ 1,...,T\right\} $. In
particular, if we look at the level sets $\Sigma _{t-1}^{z}$ of
the price process $(S_{t})_{t\in I}$, we can identify the
component of these sets that must be polar (Corollary
\ref{B-meas}) which coincides with the whole level set when
$0\notin conv(\Delta S_{t}(\Sigma _{t-1}^{z}))$ (Corollary
\ref{conditionalPolar}). Further care
is required in the multiperiod case due to the backward effects (see Section %
\ref{SecSvuotamento}), but nevertheless a full characterization of $\mathcal{%
M}$-polar sets is obtained in Section \ref{maxPolar} .\newline

In this section we build a predictable strategy that embrace all
the inefficiencies of the market. Unfortunately, even on a single
time-step, the
polar set given by Corollary \ref{B-meas} belongs, in general, to $\mathfrak{%
F}_{t}$ (the universal $\mathcal{M}$-completion), hence the
trading
strategies suggested by equation \eqref{levelArbitrage} in Lemma \ref%
{spezzamento} fail to be predictable. This reflects into the
necessity of enlargement of the original filtration by
anticipating some one step-head information. Under this filtration
enlargement, which depends only on the underlying structure of the
market, the set of martingale measures will not change (see Lemma
\ref{invariant}).

\begin{definition}
We call Universal Arbitrage Aggregator the strategy%
\begin{equation}
H_{t}^{\bullet }(\omega )\mathbf{1}_{\Sigma _{t-1}^{z}}:=\left\{
\begin{array}{ccl}
H_{t,z}(\omega ) & \quad\text{on } & \bigcup_{i=1}^{\beta _{t,z}}B_{t,z}^{i} \\
0 & \quad\text{on } & \Sigma _{t-1}^{z}\setminus \bigcup_{i=1}^{\beta
_{t,z}}B_{t,z}^{i}%
\end{array}%
\right. ,\quad t\in I_{1}=\left\{ 1,...,T\right\} ,
\label{universal}
\end{equation}%
where $z\in \mathbf{Z}$ satisfies $z_{0:t-1}=S_{0:t-1}(\omega ),$
$H_{t,z},$
$B_{t,z}^{i},B_{t,z}^{\ast }$ comes from \eqref{levelArbitrage} and Lemma %
\ref{spezzamento} with $\Gamma =\widetilde{\Sigma }_{t-1}^{z}$.
\newline
This strategy is predictable with respect to the enlarged filtration $%
\widetilde{\mathbb{F}}=\{\widetilde{\mathcal{F}}_{t}\}_{t\in I}$
given by
\begin{eqnarray}
\widetilde{\mathcal{F}}_{t} &:&=\mathcal{F}_{t}\vee \sigma
(H_{1}^{\bullet },...,H_{t+1}^{\bullet }),\text{ }t\in \left\{
0,...,T-1\right\}
\label{enlargment} \\
\widetilde{\mathcal{F}}_{T} &:&=\mathcal{F}_{T}\vee \sigma
(H_{1}^{\bullet },...,H_{T}^{\bullet }).
\end{eqnarray}
\end{definition}

\begin{remark}
The strategy $H^{\bullet }$ in equation \eqref{universal} satisfies $%
V_{T}(H^{\bullet })\geq 0$\ and
\begin{equation}
\mathcal{V}_{H^{\bullet
}}^{+}=\bigcup_{t=1}^{T}\widetilde{\mathfrak{B}}_{t}. \label{VH}
\end{equation}%
Indeed, from Lemma \ref{spezzamento}\ $H_{t,z}\cdot \Delta S_{t}>0$\ on $%
\bigcup_{i=1}^{\beta _{t,z}}B_{t,z}^{i},$ so that $\bigcup_{t=1}^{T}%
\widetilde{\mathfrak{B}}_{t}\subseteq \mathcal{V}_{H^{\bullet
}}^{+}$. On the other hand, $\mathcal{V}_{H^{\bullet
}}^{+}\subseteq \{H_{t}^{\bullet
}\neq 0$ for some $t\}\subseteq \bigcup_{t=1}^{T}\widetilde{\mathfrak{B}}%
_{t} $. \newline For $t<T$ we therefore conclude that
$\widetilde{\mathcal{F}}_{t}\subseteq
\mathcal{F}_{t}\vee \bigcup_{s=1}^{t+1}\mathcal{N}_{s}\subseteq \mathfrak{F}%
_{t}$, where
\begin{equation*}
\mathcal{N}_{t}:=\left\{ A=\bigcup_{z\in \mathbf{V}}\bigcup_{i\in
J(z)}B_{t,z}^{i}\mid \text{for some }%
\begin{array}{c}
\mathbf{V}\subseteq \mathbf{Z} \\
J(z)\subseteq \{1,...,\beta _{t,z}\}%
\end{array}%
\right\} \cup \mathfrak{D}_{t} ,
\end{equation*}%
while for $t=T$, $\widetilde{\mathcal{F}}_{T}\subseteq
\mathcal{F}_{T}\vee \bigcup_{s=1}^{T}\mathcal{N}_{s}\subseteq
\mathfrak{F}_{T}$.

For any $Q\in \mathcal{M}$ and $t\in I$, any element of $\mathcal{N}_{t}$ is a subset of a $%
\mathcal{F}_{t}$-measurable $Q$-null set .
\end{remark}

From now on we will assume that the class of admissible trading strategies $%
\widetilde{\mathcal{H}}$ is given by all $\widetilde{\mathbb{F}}$
predictable processes. We can rewrite the definition of Arbitrage
de la classe $\mathcal{S}$ using strategies adapted to
$\widetilde{\mathbb{F}}$.
Namely, an Arbitrage de la classe $\mathcal{S}$ with respect to $\widetilde{%
\mathcal{H}}$ is an $\widetilde{\mathbb{F}}$-predictable processes $%
H=[H^{1},\ldots ,H^{d}]$ such that $V_{T}(H)\geq 0$ and
$\{V_{T}(H)>0\}$ contains a set in $\mathcal{S}$.

\begin{remark}
No Arbitrage de la classe $\mathcal{S}$ with respect to $\widetilde{\mathcal{%
H}}$ implies No Arbitrage de la classe $\mathcal{S}$ with respect to $%
\mathcal{H}$.
\end{remark}

\begin{remark}
\textbf{\emph{(Financial interpretation of the filtration
enlargement)}} Fix $t\in I_{1}$, $z\in \mathbf{Z}$, the event
$\Sigma _{t-1}^{z}=\{S_{0:t-1}=z_{0:t-1}\}$ and suppose the market
presents the opportunity given by $0\notin ri(\Delta S_{t}(\Sigma
_{t-1}^{z}))^{cc}$.
Consider two probabilities $P_{k}\in \mathcal{P}$, $k=1,2,$ for which $%
P_{k}(\Sigma _{t-1}^{z})>0$. Following Lemma \ref{spezzamento}, if $%
j_{k}:=\inf \{i=1,\ldots ,\beta \mid P_{k}(B_{t,z}^{i})>0\}<\infty
$, then
the rational choice for the strategy is $H^{j_{k}}$, as shown in Corollary %
\ref{corUnivParb}.  Thus it is possible that $j_{k}<\infty $ holds
for both probabilities, so that the two agents represented by
$P_{1}$ and $P_{2}$ agree that $\Sigma _{t-1}^{z}$ is
a non-efficient level set of the market, although it is possible that $%
j_{1}\neq j_{2}$ so that they might not agree on the trading strategy $%
H^{j_{k}}$ that establish the $P_{k}$-Classical Arbitrage on
$\Sigma _{t-1}^{z}$. In such case, these two arbitrages are
realized on different subsets of $\Sigma _{t-1}^{z}$ and generate
different payoffs. Nevertheless note that any of these agents is
able to find an arbitrage opportunity among the finite number of
trading strategies $\{H_{t,z}^{i}\}_{i=1}^{\beta _{t,z}}$ given by
Lemma \ref{spezzamento} (recall $\beta _{t,z}\leq d$). The
filtration enlargement allows to embrace them all. This can be
referred to the analogous discussion in \cite{DH07}:
\textquotedblleft A weak arbitrage opportunity is a situation
where we know there must be an arbitrage but we cannot tell,
without further information, what strategy will realize
it\textquotedblright .
\end{remark}

We expand on this argument more formally. Recall that Lemma
\ref{spezzamento} provides a partition of any level set
$\widetilde{\Sigma }_{t-1}^{z}$ with the following property: for
any $\omega \in (\Omega _{\ast })^{c}$ there
exists a unique set $B_{t,z}^{i}$, identified by $i=i(\omega )$, such that $%
\omega \in B_{t,z}^{i}$ with $z=S_{0:T}(\omega )$. Define
therefore, for any $t\in I_{1}$ the multifunction
\begin{equation}
\mathbb{H}_{t}(\omega ):=\left\{ H\in \mathbb{R}^{d}\text{ s.t.
}H\cdot \Delta S_{t}(\widehat{\omega })\geq 0\text{ for any
}\widehat{\omega }\in \cup _{j=i(\omega )}^{\beta _{t,z}}B_{t,z}^{j}\cup B_{t,z}^{\ast }\right\} \quad \text{if
}\omega \in (\Omega _{\ast })^{c}  \label{multifunctionH}
\end{equation}%
and $\mathbb{H}_{t}(\omega )=\{0\}$ otherwise.\newline Observe
that for any $t\in I_1$, if $\omega_1,\omega_2$ satisfy
$S_{0:t-1}(\omega_1)=S_{0:t-1}(\omega_2)$ and
$i(\omega_1)=i(\omega_2)$ they belong to the same $B^i_{t,z}$ and
$\mathbb{H}_t(\omega_1)=\mathbb{H}_t(\omega_2)$. In other words $\mathbb{H}_t$ is constant on any
$B^i_{t,z}$ and therefore for any open set
$V\subseteq\mathbb{R}^d$ we have $$\{\omega\in\Omega\mid
\mathbb{H}_t(\omega)\cap V\neq
\varnothing\}=\bigcup_{z\in\mathbf{Z}}\bigcup_{i=1}^{\beta_{t,z}}\{B^i_{t,z}\mid
\mathbb{H}_t(B^i_{t,z})\cap V\neq \varnothing\}$$ from which
$\mathbb{H}_t$ is measurable with respect to
$\mathcal{F}_{t-1}\vee \bigcup_{s=1}^{t}\mathcal{N}_{s}$. Note
that since $H_{t}^{\bullet }(\omega )\in \mathbb{H}_{t}(\omega )$
for any $\omega
\in \Omega $, we have that $H_{t}^{\bullet }$ is a selection of $\mathbb{H}%
_{t}$ with the same measurability. We now show how the process $\mathbb{H}:=(%
\mathbb{H}_{t})_{t\in I_1}$ provides $P$-Classical Arbitrage as
soon as we choose a probabilistic model $P\in\mathcal{P}$ which is not absolutely continuous with respect to $\nu:=\sup_{Q\in\mathcal{M}}Q$. The case of $P\ll \nu$ is discussed in Remark \ref{split}.
\begin{theorem}
\label{theoremExtraction}Let $\mathbb{H}$ be defined in %
\eqref{multifunctionH}. If $P\in \mathcal{P}$ is not absolutely continuous with respect to $\nu$ then there exists an $\mathbb{F}^P$-predictable trading strategy $H^{P}$ which is a $P$%
-Classical Arbitrage and
\begin{equation*}
H^{P}(\omega )\in \mathbb{H}(\omega )\qquad P\text{-a.s.}
\end{equation*}
where
$\mathcal{F}_{t}^{P}$
denote the $P$-completion of $\mathcal{F}_{t}$ and $\mathbb{F}^P:=\{\mathcal{F}_{t}^{P}\}_{t\in I}$.
\end{theorem}

\begin{proof}
See Appendix \ref{appendixProof}.
\end{proof}
From Lemma \ref{dec} if $P\in\mathcal{P}$ fulfills the hypothesis of Theorem \ref{theoremExtraction} there exists an  $\mathcal{F}$-measurable set $F\subseteq (\Omega_*)^c$ with $P(F)>0$. Note that from Remark
\ref{remarkPsing} this family is non-empty unless $%
(\Omega _{\ast })^{c}=\varnothing $. Theorem \ref{theoremExtraction} asserts therefore that for any probabilistic models which supports $(\Omega _{\ast })^{c}$ an $\mathbb{F}^P$-\emph{predictable} arbitrage
opportunity can be found among the values of the set-valued process $\mathbb{%
H}$. This property suggested us to baptize $\mathbb{H}$ as the
Universal Arbitrage Aggregator and thus $H^{\bullet }$ as a
(standard) selection of the Universal Arbitrage Aggregator. Note
that we could have considered a
different selection of $\mathbb{H}$ satisfying the essential requirement %
\eqref{VH}. Since this choice does not affect any of our results
we simply take $H^{\bullet }$.
\begin{remark}
\label{split} Recall from (\ref{NoA}) that any $P\in
(\mathcal{P}_{0})^{c}$ admits a $P$-Classical Arbitrage
opportunity. We can distinguish
between two different classes in $(\mathcal{P}_{0})^{c}$. The first one is: $%
\mathcal{P}_{\mathcal{M}}:=\left\{ P\in (\mathcal{P}_{0})^{c}\mid P<<%
\nu\right\} $ or, in other words, an element $P\in (\mathcal{P}%
_{0})^{c}$ belong to $\mathcal{P}_{\mathcal{M}}$ iff any subset of
$(\Omega _{\ast })^{c}$ is $P$-null. Then for each probability $P$
in this class, there exists a probability $P^{\prime }$ with
larger support that
annihilates any $P$-Classical Arbitrage opportunity. Recall Example \ref%
{ex3d} where $\Omega _{\ast }=\mathbb{Q}\cap \lbrack 1/2,+\infty
)$. By choosing $P=\delta _{\{\frac{1}{2}\}}\in
\mathcal{P}_{\mathcal{M}}$ we
clearly have $P$-Classical Arbitrages. Nevertheless by simply taking $%
P^{\prime }=\lambda \delta _{\{\frac{1}{2}\}}+(1-\lambda )\delta
_{\{2\}}$ for some $0<\lambda <1$ this market is arbitrage free.
From a model-independent point of view these situations must not
considered as market inefficiencies since they vanish as soon as
more trajectories are considered. This feature is captured by the
Universal Arbitrage Aggregator by means of the property:
$H^{\bullet }=0$ on $\Omega _{\ast }$.\newline
On the other hand when $P\in (\mathcal{P}_{0})^{c}\setminus \mathcal{P}_{%
\mathcal{M}}$ then $P$ assigns a positive measure to some $\mathcal{M}$%
-polar $\mathcal{F}$-measurable set $F\in \mathcal{N}$. Therefore,
any other $P^{\prime }\in \mathcal{P}$ with larger support will
satisfy $P^{\prime
}(F)>0$ and the probabilistic model $(\Omega ,\mathcal{F},\mathbb{F}%
,S,P^{\prime })$ will also exhibit $P^{\prime }$-Classical
Arbitrages. In the case of Example \ref{ex3d} $(\Omega _{\ast
})^{c}=B^{1}\cup B^{2}$ where
$B^{1}=\mathbb{R}^{+}\setminus \mathbb{Q}$ and $B^{2}=\mathbb{Q}\cap [0,1/2)$%
. If $P((\Omega _{\ast })^{c})>0$ the market exhibits a
$P$-Classical
Arbitrage, but this is still valid for any probabilistic model given by $%
P^{\prime }$ with $P<<P^{\prime }$. In particular if $P^{\prime
}(B^{1})>0$
then $H^{1}:=[0,0,1]$ is a $P^{\prime }$-Classical Arbitrage, while if $%
P^{\prime }(B^{1})=0$ and $P^{\prime }(B^{2})>0$ then
$H^{2}:=[1,0,0]$ is the desired strategy. In this example,
$H_{1}^{\bullet }=H^{1}1_{\left\{ \mathbb{R}^{+}\setminus
\mathbb{Q}\right\} }+H^{2}1_{\left\{ \mathbb{Q}\cap
[0,1/2)\right\} }.$
\end{remark}

\begin{lemma}
\label{invariant} $\mathcal{M}(\mathbb{F})\leftrightarrows \mathcal{M}(%
\widetilde{\mathbb{F}})$ with the following meaning
\begin{itemize}
\item[$\bullet $ ] the restriction of any $\widetilde{Q}\in \mathcal{M}(\widetilde{%
\mathbb{F}})$ to $\mathcal{F}_{T}$ belongs to
$\mathcal{M}(\mathbb{F})$;

\item[$\bullet $ ] any $Q\in \mathcal{M}(\mathbb{F})$ can be
uniquely extended to an element of
$\mathcal{M}(\widetilde{\mathbb{F}})$
\end{itemize}
\end{lemma}

\begin{proof}
Let $\widetilde{Q}\in \mathcal{M}(\widetilde{\mathbb{F}})$ and
$Q\in \mathcal{P}(\Omega )$ be the restriction to
$\mathcal{F}_{T}$. For any $t\in
I_{1}$ and $A\in \mathcal{F}_{t-1}$ we have $E_{Q}[(S_{t}-S_{t-1})\mathbf{1}%
_{A}]=E_{\widetilde{Q}}[(S_{t}-S_{t-1})\mathbf{1}_{A}]=0$. Let now
$Q\in
\mathcal{M}(\mathbb{F})$. There exists a unique extension to $\widetilde{%
\mathcal{F}}_{t}$ of $Q$ that we call $\widetilde{Q}$. For any $\widetilde{A}%
\in \widetilde{\mathcal{F}}_{t-1}$ with $t\in I_{1}$ there exists
$A\in
\mathcal{F}_{t-1}$ such that $\widetilde{Q}(\widetilde{A})=\widetilde{Q}%
(A)=Q(A)$. Hence $E_{\widetilde{Q}}[(S_{t}-S_{t-1})\mathbf{1}_{\widetilde{A}%
}]=E_{Q}[(S_{t}-S_{t-1})\mathbf{1}_{A}]=0$. We conclude that $E_{\widetilde{Q%
}}[S_{t}\mid \widetilde{\mathcal{F}}_{t-1}]=S_{t-1}$, hence $\widetilde{Q}%
\in \mathcal{M}(\widetilde{\mathbb{F}})$.
\end{proof}

\begin{remark}
The filtration enlargement $\widetilde{\mathbb{F}}$ has been
introduced to guarantee the aggregation of $1p$-Arbitrages on the
sets $B_{t,z}^{i}$
obtained from Lemma \ref{spezzamento} with $\Gamma =\widetilde{\Sigma }%
_{t-1}^{z}$. If indeed we follow \cite{Cohen} we can consider any
collection
of probability measures $\Theta _{t}:=\{P_{t,z}^{i}\}$ on $(\Omega ,\mathcal{%
F})$ such that $P_{t,z}^{i}(B_{t,z}^{i})=1$. Observe first that
\begin{equation*}
\mathcal{F}_{t}^{\Theta _{t}}\supseteq \sigma \left( \bigcup
\{B_{t,z}^{i}\mid z\in \mathbf{V},i\in J(z)\}\right)
\end{equation*}%
with $\mathbf{V}$ and $J(z)$ arbitrary. For any $P_{t,z}^{i}$ we
have indeed
that $\mathcal{F}_{t}^{P_{t,z}^{i}}$ contains any subset of $%
(B_{t,z}^{i})^{c}$. Therefore if $A=\bigcup \{B_{t,z}^{i}\mid z\in
\mathbf{V} $, $i\in J(z)\}$ we have

\begin{itemize}
\item[$\bullet $ ] if $z\notin \mathbf{V}$ or $i\notin J(z)$ then $A\in \mathcal{F}%
_{t}^{P_{t,z}^{i}}$ trivially because $A\subset (B_{t,z}^{i})^{c}$

\item[$\bullet $ ] if $z\in \mathbf{V}$ and $i\in J(z)$ then $A\in \mathcal{F}%
_{t}^{P_{t,z}^{i}}$ because $A=B_{t,z}^{i}\cup \bar{A}$ with $\bar{A}%
\subseteq (B_{t,z}^{i})^{c}$
\end{itemize}

It is easy to check that $\Theta _{t}$ has the Hahn property on $\mathcal{F}%
_{t}$ as defined in Definition 3.2, \cite{Cohen}, with $\Phi
_{t}:=\Theta
_{t}\mid _{\mathcal{F}_{t}}$. We can therefore apply Theorem 3.16 in \cite%
{Cohen} to find an $\mathcal{F}_{t}^{\Theta _{t}}$- measurable function $%
H_{t}$ such that $H_{t}=H_{t,z}^{i}$ $P_{t,z}^{i}$-a.s. which means that $%
H_{t}(\omega )=H_{t,z}^{i}$ for every $\omega \in B_{t,z}^{i}$.
\end{remark}


\subsection{Main Results}

\label{main}


Our aim now is to show how the absence of arbitrage de la classe $%
\mathcal{S}$ provides a pricing functional \emph{via} the
existence of a martingale measure with nice properties. \newline
Clearly the \textquotedblleft No $1p$-Arbitrage\textquotedblright
condition is the strongest that one can assume in this model
independent framework and we have shown in Proposition \ref{NoPA}
that it automatically implies the existence of a full support
martingale measure. On the other
hand we are interested in characterizing those markets which can exhibit $1p$%
-Arbitrages but nevertheless admits a rational system of pricing
rules.\newline

The set $\Omega _{\ast }$ introduced in Section \ref{maxPolar} has a clear financial interpretation as it represents the set of events for which No 1p-Arbitrage can be found. This is the content of the following Proposition.\\
Let $(\Omega ,\widetilde{\mathcal{F}}_{T},\widetilde{\mathbb{F}})$, $%
\widetilde{\mathcal{H}}$ as in Section \ref{sectionFiltration} and
define
\begin{equation*}
\widetilde{\mathcal{H}}^{+}:=\left\{ H\in
\widetilde{\mathcal{H}}\mid \
V_{T}(H)(\omega )\geq 0\,\forall \,\omega \in \Omega \mbox{ and }%
V_{0}(H)=0\right\} .
\end{equation*}

\begin{proposition}
\label{obloj}

\begin{itemize}
\item[(1)] $\mathcal{V}_{H^{\bullet }}^{+}=\bigcup_{H\in \widetilde{\mathcal{%
H}}^{+}}\mathcal{V}_{H}^{+}=(\Omega _{\ast })^{c}$

\item[(2)] $\mathcal{M}\neq \varnothing $ if and only if
$\bigcup_{H\in
\widetilde{\mathcal{H}}^{+}}\mathcal{V}_{H}^{+}$ is strictly contained in $%
\Omega $.
\end{itemize}
\end{proposition}

\begin{proof}
(2) follows from (1) and Proposition \ref{LemNOpolar}. Indeed: $\mathcal{M}%
\neq \varnothing $ iff $\Omega _{\ast }\neq \varnothing $ iff
$(\Omega
_{\ast })^{c}\varsubsetneqq \Omega $ iff $\bigcup_{H\in \widetilde{\mathcal{H%
}}^{+}}\mathcal{V}_{H}^{+}\subsetneq \Omega $. Now we prove (1). Given %
\eqref{VH}, we only need to show the inclusion $\bigcup_{H\in \widetilde{%
\mathcal{H}}^{+}}\mathcal{V}_{H}^{+}\subseteq (\Omega _{\ast })^{c}$. Let $%
\overline{\omega }\in \bigcup_{H\in \widetilde{\mathcal{H}}^{+}}\mathcal{V}%
_{H}^{+}$, then there exists $\overline{H}\in
\widetilde{\mathcal{H}}^{+}$ and $t\in I_{1}$ such that
$\overline{H}_{t}(\omega )\cdot \Delta
S_{t}(\omega )\geq 0$ $\forall \omega \in \Omega $ and $\overline{H}_{t}(%
\overline{\omega })\cdot \Delta S_{t}(\overline{\omega })>0$. Let $z=S_{0:T}(%
\overline{\omega })$. From Lemma \ref{spezzamento} there exists
$i\in \{1,\ldots ,\beta _{t,z}\}$ such that $\overline{\omega }\in
B_{t,z}^{i}$ hence we conclude that $\overline{\omega }\in
\widetilde{\mathfrak{B}}_{t}$ and therefore $\overline{\omega }\in
(\Omega _{\ast })^{c}$.
\end{proof}

\begin{proof}[Proof of Theorem \protect\ref{C-polarEquiv}]
We prove that
\begin{equation*}
\exists \text{ an Arbitrage de la classe }\mathcal{S}\text{ in }\widetilde{%
\mathcal{H}}\Longleftrightarrow \mathcal{M}=\varnothing \text{ or }\mathcal{N%
}\text{ contains sets of }\mathcal{S}.
\end{equation*}%
Notice that if $H\in \mathcal{\widetilde{H}}$ satisfies
$V_{T}(H)(\omega)\geq 0\
\forall \omega \in \Omega $ then, if $\mathcal{M}\neq \varnothing $, $%
\mathcal{V}_{H}^{+}\in \mathcal{N}$, otherwise $0<\mathbb{E}%
_{Q}[V_{T}(H)]=V_{0}(H)=0$ for $Q\in \mathcal{M}$. If there exists an $%
\widetilde{\mathcal{H}}$-Arbitrage de la classe $\mathcal{S}$ then $\mathcal{%
V}_{H}^{+}$ contains a set in $\mathcal{S}$ and therefore
$\mathcal{N}$ contains a set in $\mathcal{S}$. If instead
$\mathcal{M}=\varnothing $ we already have the thesis. For the
opposite implication, we exploit the Universal Arbitrage
$H^{\bullet }\in \widetilde{\mathcal{H}}$ as defined in equation
\eqref{universal} satisfying $V_{T}(H^{\bullet })(\omega)\geq 0\
\forall
\omega \in \Omega $ and $\mathcal{V}_{H^{\bullet }}^{+}=\bigcup_{t=1}^{T}%
\widetilde{\mathfrak{B}}_{t}=(\Omega _{\ast })^{c}.$ If $\mathcal{M}%
=\varnothing $ then, by Proposition \ref{LemNOpolar}, $(\Omega
_{\ast })^{c}=\Omega $ and $H^{\bullet }$ is an
$\widetilde{\mathcal{H}}$-Model Independent Arbitrage and hence
(from \eqref{MIA}) $H^{\bullet }$ is also an
Arbitrage de la classe $\mathcal{S}$. If $\mathcal{M}\neq \varnothing $ and $%
\mathcal{N}$ contains a set $C$ in $\mathcal{S}$ then $C\subseteq
(\Omega _{\ast })^{c}=\mathcal{V}_{H^{\bullet }}^{+}$ from
\eqref{decomp2} and
Proposition \ref{obloj}, item 1. Therefore $H^{\bullet }$ is an $\widetilde{%
\mathcal{H}}$-Arbitrage de la classe $\mathcal{S}$.
\end{proof}

\begin{definition}
\label{defRS}Define the following convex subset of $\mathcal{P}$:
\begin{equation}\label{eqRS}
\mathcal{R}_{\mathcal{S}}:=\left\{ Q\in \mathcal{\mathcal{P}}\mid Q(C)>0%
\text{ for all }C\in \mathcal{S}\right\} .
\end{equation}
\end{definition}

The martingale measures having the property of the class $\mathcal{R}_{%
\mathcal{S}}$ will be associated to the Arbitrage de la classe
$\mathcal{S}.$

\begin{example}
\label{ExA}We consider the examples introduced in Definition
\ref{def}.
Suppose there are no Model Independent Arbitrage in $\widetilde{\mathcal{H}}$%
. From Theorem \ref{C-polarEquiv} we obtain:

\begin{enumerate}
\item $1p$-Arbitrage: $\mathcal{S}=\left\{ C\in \mathcal{F}\mid
C\neq \varnothing \right\} .$

\begin{itemize}
\item[$\bullet $ ] No $1p$-Arbitrage in $\widetilde{\mathcal{H}}$ iff $\mathcal{N}%
=\varnothing ;$

\item[$\bullet $ ] $\mathcal{R}_{\mathcal{S}}=\mathcal{P}_{+}$, if
$\Omega $ finite or countable; otherwise
$\mathcal{R}_{\mathcal{S}}=\varnothing .$

\item In the case of $np$-Arbitrage we have:
\begin{equation*}
\mathcal{R}_{\mathcal{S}}=\left\{ Q\in \mathcal{P}\mid Q(A)>0\text{ for all }%
A\subseteq \Omega \text{ having at least }n\text{
elements}\right\}
\end{equation*}%
No $np$-Arbitrage in $\widetilde{\mathcal{H}}$ iff $\mathcal{N}$
does not contain elements having more than $n-1$ elements.
\end{itemize}

\item Open Arbitrage: $\mathcal{S}=\left\{ C\in \mathcal{B}(\Omega )\mid C%
\text{ open non-empty}\right\} $.

\begin{itemize}
\item[$\bullet $ ] No Open Arbitrage in $\widetilde{\mathcal{H}}$
iff $\mathcal{N}$ does not contain open sets;

\item[$\bullet $ ] $\mathcal{R}_{\mathcal{S}}=\mathcal{P}_{+}$.
\end{itemize}

\item $\mathcal{P}^{\prime }$-q.s. Arbitrage: $\mathcal{S}=\left\{
C\in \mathcal{F}\mid P(C)>0\text{ for some }P\in
\mathcal{P}^{\prime }\right\} ,$ $\mathcal{P}^{\prime }\subseteq
\mathcal{P}$.

\begin{itemize}
\item[$\bullet $ ] No $\mathcal{P}^{\prime }$-q.s. Arbitrage in
$\widetilde{\mathcal{H}}$ iff $\mathcal{N}$ may contain only
$\mathcal{P}^{\prime }$-polar sets;

\item[$\bullet $ ] $\mathcal{R}_{\mathcal{S}}=\left\{ Q\in
\mathcal{P}\mid P^{\prime }\ll Q\text{ for all }P^{\prime }\in
\mathcal{P}^{\prime }\right\} .$
\end{itemize}

\item $P$-a.s. Arbitrage: $\mathcal{S}=\left\{ C\in
\mathcal{F}\mid P(C)>0\right\} ,$ $P\in \mathcal{P}$.

\begin{itemize}
\item[$\bullet $ ] No $P$-a.s. Arbitrage in
$\widetilde{\mathcal{H}}$ iff $\mathcal{N}$ may contain only
$P$-null sets;

\item[$\bullet $ ] $\mathcal{R}_{\mathcal{S}}=\left\{ Q\in
\mathcal{P}\mid P\ll Q\right\} .$
\end{itemize}

\item Model Independent Arbitrage: $\mathcal{S}=\left\{ \Omega
\right\} $.

\begin{itemize}
\item[$\bullet $ ] $\mathcal{R}_{\mathcal{S}}=\mathcal{P}$.
\end{itemize}

\item \label{item6}$\varepsilon $-Arbitrage: $\mathcal{S}=\left\{
C_{\varepsilon }(\omega )\mid \omega \in \Omega \right\} ,$ where $%
\varepsilon >0$ is fixed and $C_{\varepsilon }(\omega )$ is the
closed ball in $(\Omega ,d)$ of radius $\varepsilon $ and centered
in $\omega .$

\begin{itemize}
\item[$\bullet $ ] No $\varepsilon $-Arbitrage in $\widetilde{\mathcal{H}}$ iff $\mathcal{%
N}$ does not contain closed balls of radius $\varepsilon $;

\item[$\bullet $ ] $\mathcal{R}_{\mathcal{S}}=\left\{ Q\in
\mathcal{P}\mid Q(C_{\varepsilon }(\omega ))>0\text{ for all
}\omega \in \Omega \right\} .$
\end{itemize}
\end{enumerate}
\end{example}

\begin{corollary}
\label{CorS}Suppose that the class $\mathcal{S}$ has the property:%
\begin{equation}
\exists \left\{ C_{n}\right\} _{n\in \mathbb{N}}\subseteq
\mathcal{S}\text{
s.t. }\forall C\in \mathcal{S}\text{ }\exists \overline{n}\text{ s.t. }C_{%
\overline{n}}\subseteq C\text{.}  \label{c}
\end{equation}%
Then:
\begin{equation}
\text{No Arb. de la classe }\mathcal{S}\text{ in }\widetilde{\mathcal{H}}%
\Longleftrightarrow \mathcal{M}\cap \mathcal{R}_{\mathcal{S}}\neq
\varnothing .  \label{E1}
\end{equation}
\end{corollary}

\begin{proof}
Suppose $Q\in \mathcal{M}\cap \mathcal{R}_{\mathcal{S}}\neq
\varnothing $. Then any polar set $N\in \mathcal{N}$ does not
contain sets in $\mathcal{S}$
(otherwise, if $C\in \mathcal{S}$ and $C\subseteq N$ then $Q(C)>0$ and $%
Q(N)=0$, a contradiction). Then, from Theorem \ref{C-polarEquiv},
No Arbitrage de la classe $\mathcal{S}$ holds true. Conversely,
suppose that No Arbitrage de la classe $\mathcal{S}$ holds true so
that $\mathcal{M}\neq \varnothing $ and let $\left\{ C_{n}\right\}
_{n\in C}\subseteq \mathcal{S}$
be the collection of sets in the assumption. From Theorem \ref{C-polarEquiv}%
, we obtain that $N\in \mathcal{N}$ does not contain any set in $\mathcal{S}$%
, and so each set $C_{n}$ is not a polar set, hence for each $n$
there
exists $Q_{n}\in \mathcal{M}$ such that $Q_{n}(C_{n})>0.$ Set $%
Q:=\sum_{n=1}^{\infty }\frac{1}{2^{n}}Q_{n}\in \mathcal{M}$ (see Lemma \ref%
{pasting}). Take any $C\in \mathcal{S}$ \ and let
$C_{\overline{n}}\subseteq C.$ Then
\begin{equation*}
Q(C)\geq \frac{1}{2^{\overline{n}}}Q_{\overline{n}}(C)\geq \frac{1}{2^{%
\overline{n}}}Q_{\overline{n}}(C_{\overline{n}})>0
\end{equation*}%
and $Q\in \mathcal{M}\cap \mathcal{R}_{\mathcal{S}}$.
\end{proof}

\begin{corollary}
\label{CorNo}Let $\mathcal{S}$ be the class of non empty open
sets. Then the
condition (\ref{c}) is satisfied and therefore%
\begin{equation}
\text{No Open Arbitrage in
}\widetilde{\mathcal{H}}\Longleftrightarrow \mathcal{M}_{+}\neq
\varnothing .  \label{E2}
\end{equation}
\end{corollary}

\begin{proof}
Consider a dense countable subset $\{\omega _{n}\}_{n\in \mathbb{N}}$ of $%
\Omega $, as $\Omega $ is Polish. Consider the open balls:
\begin{equation*}
B^{m}(\omega _{n}):=\left\{ \omega \in \Omega \mid d(\omega ,\omega _{n})<%
\frac{1}{m}\right\} ,\text{ }m\in \mathbb{N}\text{,}
\end{equation*}%
The density of $\{\omega _{n}\}_{n\in \mathbb{N}}$ implies that
$\Omega =\bigcup_{n\in \mathbb{N}}B^{m}(\omega _{n})$ for any
$m\in \mathbb{N}$. Take any open set $C\subseteq \Omega $. Then
there exists some $\overline{n}$ such that $\omega
_{\overline{n}}\in C$. Take $\overline{m}\in \mathbb{N}$
sufficiently big so that $B^{\overline{m}}(\omega
_{\overline{n}})\subseteq C $.
\end{proof}

\begin{corollary}
Suppose that $\Omega $ is finite or countable. Then the condition
(\ref{c}) is fulfilled and therefore:
\begin{equation}
\text{No Arb. de la classe }\mathcal{S}\text{ in }\widetilde{\mathcal{H}}%
\Longleftrightarrow \mathcal{M}\cap \mathcal{R}_{\mathcal{S}}\neq
\varnothing .  \label{E3}
\end{equation}%
In particular:
\begin{eqnarray}
\text{No }1p\text{-Arbitrage in }\widetilde{\mathcal{H}} &%
\Longleftrightarrow &\mathcal{M}_{+}\neq \emptyset .  \label{a1} \\
\text{No }P\text{-a.s. Arbitrage in }\widetilde{\mathcal{H}}
&\Longleftrightarrow &\exists \text{ }Q\in \mathcal{M}\text{ s.t.
}P\ll Q.
\label{a2} \\
\text{No }\mathcal{P}^{\prime }\text{-q.s. Arbitrage in }\widetilde{\mathcal{%
H}}\text{ } &\Longleftrightarrow &\exists \text{ }Q\in
\mathcal{M}\text{ s.t. }P^{\prime }\ll Q\text{ }\forall P^{\prime
}\in \mathcal{P}^{\prime }. \label{a3}
\end{eqnarray}
\end{corollary}

\begin{proof}
Define $\mathcal{S}_{0}:=\left\{ \left\{ \omega \right\} \mid
\omega \in \Omega \text{ such that there exists }C\in
\mathcal{S}\text{ with }\omega \in C\right\} $. Then
$\mathcal{S}_{0}$ is at most a countable set and satisfies
condition (\ref{c}).
\end{proof}

\begin{remark}
While (\ref{a1}) holds also for $1p$-Arbitrage in $\mathcal{H}$
(see Proposition \ref{corPoint})), (\ref{a2}) and (\ref{a3}) can
not be improved.
Indeed, by replacing in the example (\ref{ex1000}) $\mathbb{R}^{+}$ with $%
\mathbb{Q}^{+}$ and $\mathbb{Q}^{+}$ with $\mathbb{N}$, $\Omega $
is
countable, we still have $\mathcal{M}=\varnothing $ but there are No $P$%
-a.s. Arbitrage in $\mathcal{H}$ if $P(\mathbb{Q}^{+}\diagdown
\mathbb{N})=0$ (see Section \ref{examples}, item 5 (a)).
\end{remark}

\begin{remark}
There are other families of sets satisfying condition (\ref{c}).
For example, in a topological setting, nowhere dense subset of
$\Omega $ (those having closure with empty interior)\ are often
considered \textquotedblleft negligible\textquotedblright\ sets.
Then the class of sets which are the complement of nowhere dense
sets, satisfies condition (\ref{c}).
\end{remark}
\begin{remark}
Condition (\ref{c}) is not necessary to obtain the desired equivalence (\ref%
{E1}). Consider for example the class $\mathcal{S}$ defining $\varepsilon $%
-Arbitrage in Example \ref{ExA} item \ref{item6}. In such a case condition (%
\ref{c}) fails, as soon as $\Omega $ is uncountable. However, we
now prove that (\ref{E1}) holds true, when $\Omega =\R$. We
already know by the previous proof that $\mathcal{M}\cap
\mathcal{R}_{\mathcal{S}}\neq
\varnothing $ implies No Arbitrage de la classe $\mathcal{S}$ in $\widetilde{%
\mathcal{H}}$. For the converse, from No Arbitrage de la classe
$\mathcal{S}$
in $\widetilde{\mathcal{H}}$ we know that each element in $\mathcal{S}%
:=\{[r-\varepsilon ,r+\varepsilon ]\mid r\in \R\}$ is not a polar
set. Consider the countable class
\begin{equation*}
G:=\{[q-\varepsilon ,q+\varepsilon ]\mid q\in \mathbb{Q}\}\subseteq \mathcal{%
S}.
\end{equation*}%
Each set $G_{n}\in G$ is not a polar set, hence for each $n$ there exists $%
Q_{n}\in \mathcal{M}$ such that $Q_{n}(G_{n})>0.$ Set $\overline{Q}%
:=\sum_{n=1}^{\infty }\frac{1}{2^{n}}Q_{n}\in \mathcal{M}$ (see Lemma \ref%
{pasting}). The set
\begin{equation*}
D:=\{r\in \R\mid \overline{Q}([r-\varepsilon ,r+\varepsilon ])=0\}
\end{equation*}%
is at most countable. Indeed, any two distinct intervals
$J:=[r-\varepsilon ,r+\varepsilon ]$ and $J^{\prime }:=[r^{\prime
}-\varepsilon ,r^{\prime }+\varepsilon ],$ with $r,r^{\prime }\in
D$, must be disjoint, otherwise for a rational $q$ between $r$ and
$r^{\prime }$ we would have: $[q-\varepsilon
,q+\varepsilon ]\subseteq J\cup J^{\prime }$ and thus $\overline{Q}%
([q-\varepsilon ,q+\varepsilon ])=0,$ which is impossible by
construction of $\overline{Q}$. For each $r_{n}\in D$ the set
$[r_{n}-\varepsilon ,r_{n}+\varepsilon ]\in \mathcal{S}$ is not a
polar set, hence for each $n$
there exists $\widehat{Q}_{n}\in \mathcal{M}$ such that $\widehat{Q}%
_{n}([r_{n}-\varepsilon ,r_{n}+\varepsilon ])>0.$ Set $\widehat{Q}%
:=\sum_{n=1}^{\infty }\frac{1}{2^{n}}\widehat{Q}_{n}\in \mathcal{M}$. Thus $%
Q:=\frac{1}{2}\overline{Q}+\frac{1}{2}\widehat{Q}\in
\mathcal{M}\cap \mathcal{R}_{\mathcal{S}}$ is the desired measure.
\end{remark}

\section{Feasible Markets}

\label{marketfeasible} We extend the classical notion of arbitrage
with respect to a single
probability measure $P\in \mathcal{P}$ to a class of probabilities $\mathcal{%
R}\subseteq \mathcal{P}$ as follows:

\begin{definition}
The market admits \textbf{$\mathcal{R}$-Arbitrage} if

\begin{itemize}
\item[$\bullet $ ] for all $P\in \mathcal{R}$ there exists a
$P$-Classical Arbitrage.
\end{itemize}

We denote with No $\mathcal{R}$-Arbitrage the property: for some
$P\in \mathcal{R}$, $NA(P)$ holds true.
\end{definition}

\begin{remark}[Financial interpretation of $\mathcal{R}$-Arbitrage.]
If a model admits an $\mathcal{R}$-Arbitrage then the agent will
not be able to find a fair pricing rule, whatever model $P\in
\mathcal{R}$ he will choose. However, the presence of an
$\mathcal{R}$-Arbitrage only implies
that for each $P$ there exists a trading strategy $H^{P}$ which is a $P$%
-Classical Arbitrage and this is a different concept respect to
the existence of one single trading strategy $H$ that realizes an
arbitrage for all $P\in \mathcal{R}$. In the particular case of
$\mathcal{R}=\mathcal{P}$ this notion was firstly introduced in
\cite{DH07} as ``Weak Arbitrage opportunity'' and further studied
in \cite{CO11,DOR14} and the reference therein. The No
$\mathcal{R}$-Arbitrage property above should not be confused with
the condition $NA(\mathcal{R})$ introduced by Bouchard and Nutz
\cite{Nutz2} and recalled in Section \ref{secGeometric} as well as
in Definition \ref{def}, item \ref{4defS}.
\end{remark}

We set:
\begin{equation*}
\mathcal{P}_{e}(P)=\left\{ P^{\prime }\in \mathcal{P}\mid
P^{\prime }\sim
P\right\} , \qquad
\mathcal{M}_{e}(P)=\left\{ Q\in \mathcal{M}\mid Q\sim P\right\}
\end{equation*}%
In discrete time financial markets the Dalang-Morton-Willinger
Theorem applies, so that $NA(P)$ iff $\mathcal{M}_{e}(P)\neq
\varnothing .$

\begin{proposition}
\label{FullSuppEquiv} Suppose that $\mathcal{R}\subseteq
\mathcal{P}$ has the property: $P\in \mathcal{R}$ implies
$\mathcal{P}_{e}(P)\subseteq \mathcal{R}.$ Then
\begin{equation*}
\text{No }\mathcal{R}\text{-Arbitrage}\text{ iff }\mathcal{M}\cap \mathcal{R}%
\neq \varnothing .
\end{equation*}%
In particular%
\begin{eqnarray*}
\text{No }\mathcal{R}_{\mathcal{S}}\text{-Arbitrage
iff}\;\mathcal{M}\cap
\mathcal{\mathcal{R}_{\mathcal{S}}} &\neq &\varnothing , \\
\text{No }\mathcal{P}_{+}\text{-Arbitrage}\text{ iff
}\mathcal{M}_{+} &\neq
&\varnothing , \\
\text{No }\mathcal{P}\text{-Arbitrage}\text{ iff }\mathcal{M}
&\neq &\varnothing .
\end{eqnarray*}%
where $\mathcal{R}_{\mathcal{S}}$ is defined in \eqref{eqRS} and all arbitrage conditions here
are\ with respect to $\mathcal{H}$.
\end{proposition}

\begin{proof}
Suppose $Q\in \mathcal{M}\cap \mathcal{R}\neq \varnothing $. Since
$Q\in \mathcal{R}$ and $NA(Q)$ holds true we have No
$\mathcal{R}$-Arbitrage.
Viceversa, suppose No $\mathcal{R}$-Arbitrage holds true. Then there exists $%
P\in \mathcal{R}$ for which $NA(P)$ holds true and therefore there exists $%
Q\in \mathcal{M}_{e}(P)$. The assumption
$\mathcal{P}_{e}(P)\subseteq
\mathcal{R}$ implies that $Q\in \mathcal{M}_{e}(P):=\mathcal{M}\cap \mathcal{%
P}_{e}(P)\subseteq \mathcal{M}\cap \mathcal{R}$. The particular
cases
follows from the fact that $\mathcal{R}_{\mathcal{S}}$ has the property: $%
P\in \mathcal{R}_{\mathcal{S}}$ implies
$\mathcal{P}_{e}(P)\subseteq \mathcal{R}_{\mathcal{S}}$.
\end{proof}

\begin{remark}
As a result of the previous proposition, whenever (\ref{E1}), (\ref{E2}), (%
\ref{E3}) hold true each (equivalent) condition in (\ref{E1}), (\ref{E2}), (%
\ref{E3}) is also equivalent to: No
$\mathcal{R}_{\mathcal{S}}$-Arbitrage in
$\mathcal{H}$ (with $\mathcal{S}:=\left\{ \text{open sets}\right\} $ for (%
\ref{E2}))$.$
\end{remark}

Given the measurable space $(\Omega ,\mathcal{F})$ and the price
process $S$ defined on it, in this section we investigate the
properties of the set of arbitrage free (for $S$) probabilities on
$(\Omega ,\mathcal{F})$. A minimal reasonable requirement on the
financial market is the existence of at least one probability
$P\in \mathcal{P}$ that does not allow any $P$-Classical
Arbitrage. Recall from the Introduction the definition of the set
\begin{equation*}
\mathcal{P}_{0}=\left\{ P\in \mathcal{P}\mid
\mathcal{M}_{e}(P)\neq \varnothing \right\} .
\end{equation*}%
By Proposition \ref{FullSuppEquiv} and the definition of
$\mathcal{P}_{0}$
it is clear that:%
\begin{equation*}
\text{No }\mathcal{P}\text{-Arbitrage}\Leftrightarrow
\mathcal{M}\neq \varnothing \Leftrightarrow \mathcal{P}_{0}\neq
\varnothing ,
\end{equation*}%
and each one of these conditions is equivalent to No Model
Independent
Arbitrage with respect to $\widetilde{\mathcal{H}}$ (Theorem \ref%
{corollaryMI}). When $\mathcal{P}_{0}\neq \varnothing $, it is possible that only
very few models (i.e. a \textquotedblleft small\textquotedblright\
set of probability measures - the extreme case being
$|\mathcal{P}_{0}|=1$) are arbitrage free. On the other hand, the
financial market could be very \textquotedblleft well
posed\textquotedblright , so that for \textquotedblleft
most\textquotedblright\ models no arbitrage is assured - the
extreme case being $\mathcal{P}_{0}=\mathcal{P}$.

To distinguish these two possible occurrences we analyze the
conditions under which the set $\mathcal{P}_{0}$ is dense in
$\mathcal{P}$: in this case even if there could be some particular
models allowing arbitrage opportunities, the financial market is
well posed for most models.

\begin{definition}
The market is feasible if $\overline{\mathcal{P}_{0}}=\mathcal{P}$
\end{definition}
Recall that we are here considering the $\sigma
(\mathcal{P},C_{b})$- closure.\\
In Proposition \ref{propNAfullSupp} we characterize feasibility
with the existence of a full support martingale measure, a
condition independent of any a priori fixed probability.

\begin{lemma}
\label{LemWeakClosure}For all $P\in \mathcal{P_{+}}$
\begin{equation*}
\overline{\mathcal{P}_{e}(P)}=\mathcal{P}\text{ and
}\mathcal{P}_{+}\text{ is }\sigma (\mathcal{P},C_{b})\text{-dense
in }\mathcal{P}\text{.}
\end{equation*}
\end{lemma}

\begin{proof}
It is well know that under the assumption that $(\Omega ,d)$ is separable, $%
\mathcal{P_{+}}\neq \varnothing $. Let us first show that $\forall
a\in
\Omega $ we have that $\delta _{a}\in \overline{\mathcal{P}_{e}(P)}$ where $%
P\in \mathcal{P}_{+}$ and $\delta _{a}$ is the point mass
probability
measure in $a$. Let%
\begin{equation*}
A_{n}:=\left\{ \omega \in \Omega \ :\quad d(a,\omega
)<\frac{1}{n}\right\} .
\end{equation*}%
This set is open in the topology induced by $d$ and, since $P$ has
full
support, $0<P(A_{n})<1$. Define the conditional probability measure $%
P_{n}:=P(\cdot \mid A_{n})$. For all $\ 0<\lambda <1$, $P_{\lambda
}:=\lambda P(\cdot \mid A_{n}^{c})+(1-\lambda )P(\cdot \mid
A_{n})$ is a full support measure equivalent to $P$ and
$P_{\lambda }$ converges weakly to $P(\cdot \mid A_{n})$ as
$\lambda \downarrow 0.$ Hence: $P_{n}\in
\overline{\mathcal{P}_{e}(P)}$. In order to show that $P_{n}\overset{w}{%
\rightarrow }\delta _{a}$ we prove that $\forall \ G$ open
$\liminf
P_{n}(G)\geq \delta _{a}(G)$. If $a\in G$ then $\delta _{a}(G)=1$ and $%
P(G\cap A_{n})=P(A_{n})$ eventually so we have that $\liminf
P_{n}(G)=1=\delta _{a}(G)$. Otherwise if $a\not\in G$ then $\delta
_{a}(G)=0$ and the inequality is obvious.\newline
Since $\forall a\in \Omega $ we have that $\delta _{a}\in \overline{\mathcal{%
P}_{e}(P)}$ then $co(\{\delta _{a}\ :\ a\in \Omega \})\subseteq \overline{%
\mathcal{P}_{e}(P)}$ and from the density of the probability
measures with finite support in $\mathcal{P}$ (respect to the weak
topology) it follows that
$\overline{\mathcal{P}_{e}(P)}=\mathcal{P}$. The last assertion is
obvious since $\mathcal{P}_{e}(P)\subseteq \mathcal{P}_{+}$ for
each $P\in \mathcal{P_{+}}$.
\end{proof}

\begin{proposition}
\label{propNAfullSupp}The following assertions are equivalent:

\begin{enumerate}
\item \label{item_NA_full1} $\mathcal{M}_{+}\neq \emptyset$;

\item \label{item_NA_full2} No $\mathcal{P_{+}}$-Arbitrage;

\item \label{item_NA_full3} $\mathcal{P}_{0}\cap
\mathcal{P_{+}}\neq \emptyset$;

\item \label{item_NA_full4} $\overline{\mathcal{P}_{0}\cap \mathcal{P_{+}}}=%
\mathcal{P}$;

\item \label{item_NA_full5}
$\overline{\mathcal{P}_{0}}=\mathcal{P}$.
\end{enumerate}
\end{proposition}

\begin{proof}
Since $\mathcal{M}_{+}\neq \emptyset \Leftrightarrow $ No $\mathcal{P_{+}}$%
-Arbitrage by Proposition \ref{FullSuppEquiv} and No $\mathcal{P_{+}}$%
-Arbitrage $\Leftrightarrow \mathcal{P}_{0}\cap
\mathcal{P_{+}}\neq
\emptyset $ by definition, $\ref{item_NA_full1}),\ref{item_NA_full2}),\ref%
{item_NA_full3})$ are clearly equivalent.\newline Let us show that
$\ref{item_NA_full3})\Rightarrow \ref{item_NA_full4})$: Assume
that $\mathcal{P}_{0}\cap \mathcal{P_{+}}\neq \emptyset $ and
observe
that if $P\in \mathcal{P}_{0}\cap \mathcal{P_{+}}$ then $\mathcal{P}%
_{e}(P)\subseteq \mathcal{P}_{0}\cap \mathcal{P_{+}}$, which implies that $%
\overline{\mathcal{P}_{e}(P)}\subseteq
\overline{\mathcal{P}_{0}\cap \mathcal{P_{+}}}\subseteq
\mathcal{P}$. From Lemma \ref{LemWeakClosure} we conclude that
$\ref{item_NA_full4})$ holds.\newline
Observe now that the implication $\ref{item_NA_full4})\Rightarrow \ref%
{item_NA_full5})$ holds trivially, so we just need to show that $\ref%
{item_NA_full5})\Rightarrow \ref{item_NA_full3})$. Let $P\in \mathcal{P}_{+}$%
. If $P$ satisfies NA($P$) there is nothing to show, otherwise by $\ref%
{item_NA_full5})$ there exist a sequence of probabilities $P_{n}\in \mathcal{%
P}_{0}$ such that $P_{n}\overset{w}{\rightarrow }P$ and the condition NA($%
P_{n}$) holds $\forall n\in \mathbb{N}$. Define $P^{\ast
}:=\sum_{n=1}^{+\infty }\frac{1}{2^{n}}P_{n}$ and note that for
this probability the condition NA($P^{\ast }$) holds true, so we
just need to
show that $P^{\ast }$ has full support. Assume by contradiction that $%
supp(P^{\ast })\subset \Omega .$ Then there exist an open set $O$ such that $%
P^{\ast }(O)=0$ and $P(O)>0$ since $P$ has full support. From
$P_{n}(O)=0$ for all $n$, and $P_{n}\overset{w}{\rightarrow }P$ we
obtain $0=\lim \inf P_{n}(O)\geq P(O)>0$, a contradiction.
\end{proof}

\begin{remark}
From the previous proof we observe that if the market is feasible then $%
\overline{\bigcup_{P\in \mathcal{P}_{0}}supp(P)}=\Omega $ and no
\textquotedblleft significantly large parts\textquotedblright\ of
$\Omega $ are neglected by no arbitrage models $P\in
\mathcal{P}_{0}$.
\end{remark}

\begin{proof}[Proof of Theorem \protect\ref{mainEquivalence}]
Proposition \ref{propNAfullSupp} guarantees: 1. $\Leftrightarrow $ 2. $%
\Leftrightarrow $ 3. and Corollary \ref{CorNo} assures: 3.
$\Leftrightarrow $ 4.
\end{proof}

\paragraph{The case of a countable space $\Omega $.}

When $\Omega =\{\omega _{n}\mid n\in \mathbb{N}\}$ is countable it
is possible to provide another characterization of feasibility
using the norm topology instead of the weak topology on
$\mathcal{P}$. More precisely, we consider the topology induced by
the total variation norm. A sequence of probabilities $P_{n}$
\emph{converges in variation} to $P$ if $\Vert P_{n}-P\Vert
\rightarrow 0$, where the variation norm of a signed measure $R$
is defined by:%
\begin{equation}
\Vert R\Vert =\sup_{(A_{i},\ldots ,A_{n})\in \mathcal{F}%
}\sum_{i=1}^{n}|R(A_{i})|,  \label{norm}
\end{equation}%
and $(A_{i},\ldots ,A_{n})$ is a finite partition of $\Omega $.

\begin{lemma}
\label{LemStrongClosure} Let $\Omega $ be a countable space. Then
$\forall P\in \mathcal{P_{+}}$
\begin{equation*}
\overline{\mathcal{P}_{e}(P)}^{\Vert \cdot \Vert }=\overline{\mathcal{P}_{+}}%
^{\Vert \cdot \Vert }=\mathcal{P}\text{.}
\end{equation*}
\end{lemma}

\begin{proof}
Since $\Omega $ is countable we have that
\begin{equation*}
\mathcal{P}=\{P:=\{p_{n}\}_{1}^{\infty }\in \ell ^{1}\ \mid \
p_{n}\geq 0\ \forall n\in \mathbb{N},\ \Vert P\Vert _{1}=1\},
\end{equation*}%
\begin{equation*}
\mathcal{P}_{+}=\{P\in \mathcal{P}\ \mid \ p_{n}>0\ \forall n\in \mathbb{N}%
\},
\end{equation*}%
with $\Vert \cdot \Vert _{1}$ the $\ell ^{1}$ norm. Observe that
in the
countable case $\mathcal{P}_{e}(P)=\mathcal{P}_{+}$ for every $P\in \mathcal{%
P_{+}}$. So we only need to show that for any $P\in \mathcal{P}$ and any $%
\varepsilon >0$ there exists $P^{\prime }\in \mathcal{P}_{+}$ s.t.
$\Vert P-P^{\prime }\Vert _{1}\leq \varepsilon .$

Let $P\in \mathcal{P}\setminus \mathcal{P}_{+}$. Then $P=\{p_{n}\}_{1}^{%
\infty }\in \ell ^{1}$ and there exists at least one index $n$ for which $%
p_{n}=0.$ Let $\alpha >0$ be the constant satisfying
\begin{equation*}
\sum_{n\in \mathbb{N}\text{ s.t. }p_{n}=0}\frac{\alpha }{2^{n}}=1.
\end{equation*}%
There also exists one index $n$, say $n_{1},$ for which $1\geq
p_{n_{1}}>0$. Let $p:=p_{n_{1}}>0.$

For any positive $\varepsilon <p$, define $P^{\prime
}=\{p_{n}^{\prime }\}$ by: $p_{n_{1}}^{\prime
}=p-\frac{\varepsilon }{2},$ $p_{n}^{\prime }=p_{n}$
for all $n\neq n_{1}$ s.t. $p_{n}>0$, $p_{n}^{\prime }=\frac{\alpha }{2^{n}}%
\frac{\varepsilon }{2}$ for all $n$ s.t. $p_{n}=0$. Then
$p_{n}^{\prime }>0$
for all $n$ and $\sum_{n=1}^{\infty }p_{n}^{\prime }=\sum_{n\text{ s.t.}%
p_{n}>0\text{ }}p_{n}=1$ , so that $P^{\prime }\in \mathcal{P}_{+}$ and $%
\Vert P-P^{\prime }\Vert _{1}=\varepsilon .$
\end{proof}

\begin{remark}
In the general case, when $\Omega $ is uncountable, while it is
still true that $\overline{\mathcal{P}_{+}}^{\Vert \cdot \Vert
}=\mathcal{P}$, it is no
longer true that $\overline{\mathcal{P}_{e}(P)}^{\Vert \cdot \Vert }=%
\mathcal{P}$ for any $P\in \mathcal{P}_{+}$.\newline Take $\Omega
=[0,1]$ and $\mathcal{P}_{e}(\lambda )$ the set of probability
measures equivalent to Lebesgue. It is easy to see that $\delta
_{0}\notin \overline{\mathcal{P}_{e}(\lambda )}^{\Vert \cdot \Vert
}$ since $\Vert P-\delta _{0}\Vert \geq P((0,1])=1$ .
\end{remark}

\begin{proposition}
\label{propCount}If $\Omega $ is countable, the following
conditions are equivalent:

\begin{enumerate}
\item $\mathcal{M}_{+}\neq \emptyset$;

\item No $\mathcal{P_{+}}$-Arbitrage;

\item $\mathcal{P}_{0}\cap \mathcal{P_{+}}\neq \emptyset$;

\item $\overline{\mathcal{P}_{0}}^{\Vert \cdot \Vert
}=\mathcal{P}$,
\end{enumerate}

where $\Vert \cdot \Vert $ is the total variation norm on
$\mathcal{P}$
\end{proposition}

\begin{proof}
Using Lemma \ref{LemStrongClosure} the proof is straightforward
using the same techniques as in Proposition \ref{propNAfullSupp}.
\end{proof}

\section{On Open Arbitrage}

\label{secRob}

In the introduction we already illustrated the interpretation and
robust features of the dual formulation of Open Arbitrage. In
order to prove the
equivalence between Open Arbitrage and (\ref{open}) consider the following definition and recall that $\mathcal{V}%
_{H}^{+}:=\left\{ \omega \in \Omega \mid \ V_{T}(H)(\omega
)>0\right\} $.

\begin{definition}
\label{RobArb copy(1)}Let $\tau $ be a topology on $\mathcal{P}$ and $%
\mathfrak{H}$ be a class of trading strategies. Set%
\begin{equation*}
W(\tau ,\mathfrak{H})=\left\{ H\in \mathfrak{H}\mid
\begin{array}{c}
\text{there exists a non empty }\tau -\text{open set
}\mathcal{U}\subseteq
\mathcal{P}\text{ such that } \\
\forall P\in \mathcal{U}\quad V_{T}(H)\geq 0\text{ }P\text{-a.s.}\quad \text{%
and}\quad P(\mathcal{V}_{H}^{+})>0%
\end{array}%
\right\}
\end{equation*}
\end{definition}

Clearly, $W(\tau ,\mathfrak{H})$ consists of the trading
strategies satisfying condition (\ref{open}) with respect to the
appropriate topology and the measurability requirement. The first
item in the next proposition is the announced equivalence. The
second item shows that the analogue
equivalence is true also with respect to the class $\widetilde{\mathcal{H}}$%
. Therefore, in Theorem \ref{mainEquivalence} we could add to the
four equivalent conditions also the dual formulation of Open
Arbitrage with respect to $\widetilde{\mathcal{H}}$.

\begin{proposition}
\label{WS-characterization}(1) Let $\sigma :=\sigma
(\mathcal{P},C_{b})$ and $\parallel \cdot \parallel $ the
variation norm defined in (\ref{norm}). Then:
\begin{eqnarray*}
H &\in &W(\parallel \cdot \parallel
,\mathcal{H})\Longleftrightarrow \quad
H\in \mathcal{H}\text{ is a $1p$-Arbitrage} \\
&\Uparrow & \\
H &\in &W(\sigma ,\mathcal{H})\text{ }\Longleftrightarrow \quad
H\in \mathcal{H}\text{ is an Open Arbitrage}
\end{eqnarray*}%
In addition, if $H\in W(\sigma ,\mathcal{H})$ then
$V_{T}(H)(\omega )\geq 0$ for all $\omega \in \Omega $.

(2) Let $\mathcal{F}=\mathcal{B}(\Omega )$ be the Borel sigma
algebra and
let $\widetilde{\mathcal{F}}$ be a sigma algebra such that $\mathcal{%
F\subseteq }\widetilde{\mathcal{F}}$. Define the set
\begin{equation*}
\widetilde{\mathcal{P}}:=\{\widetilde{P}:\widetilde{\mathcal{F}}\rightarrow
\lbrack 0,1]\mid \widetilde{P}\text{ is a probability}\},
\end{equation*}%
and endow $\widetilde{\mathcal{P}}$ with the topology $\widetilde{\sigma }%
:=\sigma (\widetilde{\mathcal{P}},C_{b})$. The class of admissible
trading
strategies $\widetilde{\mathcal{H}}$ is given by all $\widetilde{\mathbb{F}}$%
- predictable processes. Then
\begin{equation*}
H\in W(\widetilde{\sigma },\widetilde{\mathcal{H}})\text{ }%
\Longleftrightarrow \quad H\in \widetilde{\mathcal{H}}\text{ is an
Open Arbitrage in }\widetilde{\mathcal{H}}
\end{equation*}%
In addition, if $H\in W(\widetilde{\sigma },\widetilde{\mathcal{H}})$ then $%
V_{T}(H)(\omega )\geq 0$ for all $\omega \in \Omega $.
\end{proposition}

\begin{proof}
We prove (1) and we postpone the proof of (2) to the Appendix.

(a) $H$ is a $1p$-Arbitrage $\Rightarrow $ $H\in W(\parallel \cdot
\parallel
,\mathcal{H})$. Let $H\in \mathcal{H}$ be a $1p$-Arbitrage. Then $%
V_{T}(H)(\omega )\geq 0\ \forall \omega \in \Omega $ and there
exists a probability $P$ such that
$P(\mathcal{V}_{H}^{+})>\varepsilon >0$. From the implication
$\Vert P-Q\Vert <\varepsilon \Rightarrow |P(C)-Q(C)|<\varepsilon
$ for every $C\in \mathcal{F}$, we obtain: $\bar{P}(\mathcal{V}_{H}^{+})>0$ $%
\forall \bar{P}\in B_{\varepsilon }(P)$, where $B_{\varepsilon
}(P)$ is the ball of radius $\varepsilon $ centered in $P$. Hence
$H\in W(\parallel \cdot
\parallel ,\mathcal{H})$.

(b) $H\in W(\parallel \cdot \parallel ,\mathcal{H})$ $\Rightarrow
$ $H$ is a
$1p$-Arbitrage. If $H\in W(\parallel \cdot \parallel ,\mathcal{H})$ then $%
V_{T}(H)\geq 0$ $P$-a.s. for all $P$ in the open set
$\mathcal{U}.$ We need only to show that $B:=\left\{ \omega \in
\Omega \ \mid \ V_{T}(H)(\omega )<0\right\} $ is empty. By
contradiction, let $\omega \in B$, take any $P\in \mathcal{U}$ and
define the probability $P_{\lambda }:=\lambda \delta
_{\omega }+(1-\lambda )P $. Since $V_{T}(H)\geq 0$ $P$-a.s. we must have $%
P(\omega )=0$, otherwise $P(B)>0$. However, $P_{\lambda }(B)\geq
P_{\lambda }(\omega )=\lambda >0$ for all positive $\lambda $ and
$P_{\lambda }$ will eventually belongs to $\mathcal{U}$, as
$\lambda \downarrow 0$, which contradicts $V_{T}(H)\geq 0$
$P$-a.s. for any $P\in \mathcal{U}$.

(c) $H\in W(\sigma ,\mathcal{H})$ $\Rightarrow $ $H\in W(\parallel
\cdot
\parallel ,\mathcal{H})$. This claim is trivial because every weakly open
set is also open in the norm topology.

(d) If $H\in W(\sigma ,\mathcal{H})$ then $V_{T}(H)(\omega )\geq 0$ for all $%
\omega \in \Omega $. This follows from (c) and (b).

(e) $H\in W(\sigma ,\mathcal{H})$ $\Rightarrow $ $H$ is an Open
Arbitrage. Suppose $H\in W(\sigma ,\mathcal{H})$, so that
$V_{T}(H)(\omega )\geq 0$ for all $\omega \in \Omega $. We claim
that $(\mathcal{V}_{H}^{+})^{c}=\left\{ \omega \in \Omega \mid \
V_{T}(H)=0\right\} $ is not dense in $\Omega $. This will imply
the thesis as $int(\mathcal{V}_{H}^{+})$ will then be a non
empty open set on which $V_{T}(H)>0$. Suppose by contradiction that $%
\overline{(\mathcal{V}_{H}^{+})^{c}}=\Omega $. We know by Lemma \ref%
{denseDirac} in the Appendix that the set $\mathcal{Q}$ of
embedded
probabilities $co(\{\delta _{\omega }\}\mid \omega \in (\mathcal{V}%
_{H}^{+})^{c})$ is weakly dense in $\mathcal{P}$ and hence it
intersects, in particular, the weakly open set $\mathcal{U}$ in
the definition of $W(\sigma
,\mathcal{H})$. However, for every $P\in \mathcal{Q}$ we have $V_{T}(H)=0\ P$%
-a.s. and so $H$ is not in $W(\sigma ,\mathcal{H})$.

(f) $H$ is an Open Arbitrage $\Rightarrow $ $H\in W(\sigma
,\mathcal{H})$.
Note first that if $F$ is a closed subset of $\Omega $, then $\mathcal{P}%
(F):=\{P\in \mathcal{P}\mid \ supp(P)\subset F\}$ is a $\sigma (\mathcal{P}%
,C_{b})$-closed face of $\mathcal{P}$ from Th. 15.19 in
\cite{Aliprantis}. If $H$ is an Open Arbitrage then
$\mathcal{V}_{H}^{+}$ contains an open set and in particular
$G:=\overline{(\mathcal{V}_{H}^{+})^{c}}$ is a closed set strictly
contained in $\Omega $. Observe then that $\mathcal{U}:=\left(
\mathcal{P}(G)\right) ^{c}$ is a non empty open set of
probabilities that fulfills the properties in the definition of
$W(\sigma ,\mathcal{H})$.
\end{proof}

The following proposition is an improvement of (\ref{a1}), as the $1p$%
-Arbitrage is defined with respect to $\mathcal{H}$.

\begin{proposition}
\label{corPoint}For $\Omega $ countable: No $1p$-Arbitrage in $\mathcal{H}$ $%
\Longleftrightarrow \mathcal{M}_{+}\neq \emptyset .$
\end{proposition}

\begin{proof}
From Propositions \ref{NoPA} and \ref{WS-characterization} we only
need to prove $\mathcal{M}_{+}\neq \emptyset \Longrightarrow $
$W(\parallel \cdot
\parallel ,\mathcal{H})=\varnothing $. From Proposition \ref{propCount} item
4) we have $\mathcal{M}_{+}\neq \emptyset \Longrightarrow \overline{\mathcal{%
P}_{0}}^{\Vert \cdot \Vert }=\mathcal{P}$ and so for every (norm) open set $%
\mathcal{U}\subseteq \mathcal{P}$ there exists $P\in
\mathcal{P}_{0}\cap \mathcal{U}$ for which $NA(P)$ holds, which
implies $W(\parallel \cdot
\parallel ,\mathcal{H})=\varnothing $.
\end{proof}


\subsection{On the continuity of S with respect to $\protect\omega $}

\label{remarkCont} Consider first a one period market $I=\{0,1\}$
with $S_0=s_0\in\mathbb{R}^d$ and $S_{1}$ a random outcome
\textit{continuous} in $\omega $. Then every $1p$-Arbitrage generates an
Open Arbitrage (this was shown by \cite{Riedel} and is
intuitively clear). From Proposition \ref{NoPA}, No $1p$-Arbitrage implies $%
\mathcal{M}_{+}\neq \varnothing $ and therefore No Open Arbitrage.
We then
conclude that, in this particular case, the three conditions are all equivalent and Theorem \ref%
{mainEquivalence} holds without the enlargement of the natural
filtration so that we recover in particular the result stated in
\cite{Riedel}.

\bigskip

Differently from the one period case, in the multi-period setting
it is no longer true that No Open Arbitrage and No $1p$-Arbitrage
(with respect to admissible strategies $\mathcal{H}$) are
equivalent, as shown by the following examples. Moreover, even
with $S$ continuous in $\omega ,$ No Open Arbitrage is not
equivalent to $\mathcal{M}_{+}\neq \varnothing $ as long as we do
not enlarge the filtration as in Section \ref{sectionFiltration}.

\begin{example}
\label{multi1} Consider $\Omega =[0,1]\times \lbrack 0,1]$, $\mathcal{F}=%
\mathcal{B}_{[0,1]}\otimes \mathcal{B}_{[0,1]}$ and the canonical
process given by $S_{1}(\omega )=\omega _{1}$ and $S_{2}(\omega
)=\omega _{2}$. Clearly for any $\omega =(\omega _{1},\omega
_{2})$ such that $\omega _{1}\in (0,1)$ we have that $0\in
ri(\Delta S_{2}(\Sigma _{1}^{\omega
}))^{cc}$. On the other hands for $\overline{\omega }=(1,\omega _{2})$ or $%
\widehat{\omega }=(0,\omega _{2})$ we have $1p$-Arbitrages since $S_{2}(%
\overline{\omega })\leq S_{1}(\overline{\omega })$ with $<$ for
any $\omega _{2}\neq 1$ and $S_{2}(\widehat{\omega })\geq
S_{1}(\widehat{\omega })$ with
$>$ for any $\omega _{2}\neq 0$. Denote by $\Sigma ^{1}=\{S_{1}=1\}$ and $%
\Sigma ^{0}=\{S_{1}=0\}$ then $\alpha (\omega )=-\mathbf{1}_{\Sigma ^{1}}+%
\mathbf{1}_{\Sigma ^{0}}$ is a $1p$-Arbitrage which does not admit
any open arbitrage since neither $\Sigma ^{1}$ nor $\Sigma ^{0}$
are open sets, and any strategy which is not zero on $(\Sigma
^{1}\cup \Sigma ^{0})^{c}$ gives both positive and negative
payoffs.
\end{example}

\begin{example}
\label{multi2} We show an example of a market with $S$ continuous
in $\omega $, with no Open Arbitrage in $\mathcal{H}$ and
$\mathcal{M}_+=\varnothing$.
Let us first introduce the following continuous functions on $%
\Omega=[0,+\infty)$

\begin{equation*}
\varphi_{a,b}^m(\omega):=\left\lbrace
\begin{array}{ll}
m(\omega-a) & \omega\in [a,\frac{a+b}{2}] \\
-m(\omega-b) & \omega\in [\frac{a+b}{2},b] \\
0 & \text{otherwise}%
\end{array}%
\right. \phi_{a,b}^m(\omega):=\left\lbrace
\begin{array}{ll}
m(\omega-a) & \omega\in [a,a+1] \\
m & \omega\in [a+1,b-1] \\
-m(\omega-b) & \omega\in [b-1,b] \\
0 & \text{otherwise}%
\end{array}%
\right.
\end{equation*}
with $a,b,m\in\mathbb{R}$. Define the continuous (in $\omega$)
stochastic process $(S_t)_{t=0,1,2,3}$
\begin{eqnarray*}
S_0(\omega)&:=&\frac{1}{2} \\
S_1(\omega)&:=&\phi_{[0,3]}^1(\omega)+\phi_{[3,6]}^1(\omega)+\sum_{k=3}^%
\infty\varphi_{[2k,2k+2]}^1(\omega) \\
S_2(\omega)&:=&\phi_{[0,3]}^{\frac{1}{2}}(\omega)+\phi_{[3,6]}^{\frac{1}{2}%
}(\omega)+\sum_{k=3}^\infty\varphi_{[2k,2k+2]}^2(\omega) \\
S_3(\omega)&:=&\varphi_{[0,3]}^2(\omega)+\varphi_{[3,6]}^\frac{1}{4}%
(\omega)+\varphi_{[6,8]}^4(\omega)+\sum_{k=4}^\infty\varphi_{[2k+1-\frac{1}{k%
},2k+1+\frac{1}{k}]}^{4k}(\omega)
\end{eqnarray*}
It is easy to check that given $z\in\mathbf{Z}$ such that $z_{0:2}=[\frac{1}{%
2},1,2]$, we have $\Sigma_{2}^z=\{2k+1\}_{k\geq 3}$ and $H:=\mathbf{1}%
_{\Sigma_{2}^z}$ is the only $1p$-Arbitrage opportunity in the
market. One can also check that $\mathcal{V}_H^+=\Sigma_{2}^z$, as
a consequence, $H$ is
not an Open Arbitrage and%
\begin{equation}  \label{null}
Q(\{2k+1\}_{k\geq 3})=0 \text{ for any }Q\in\mathcal{M}
\end{equation}
Consider now $\hat{z}\in\mathbf{Z}$ with $\hat{z}_{0:2}=[\frac{1}{2},1,\frac{%
1}{2}]$ and the corresponding level set $\Sigma_{2}^{\hat{z}}$. It
is easy to check that
\begin{eqnarray}  \label{neg}
\Sigma_{2}^{\hat{z}}=[1,2]\cup[4,5]&\text{and}&\Delta S_2<0 \text{
on } \Sigma_{2}^{\hat{z}}
\end{eqnarray}
Observe now that $z_{0:1}=\hat{z}_{0:1}$ and that $\Sigma_{1}^z=[1,2]\cup[4,5%
]\cup\{2k+1\}_{k\geq 3}$. We therefore have
\begin{equation*}
S_2(\omega)=\left\lbrace%
\begin{array}{ll}
2 & \omega\in\{2k+1\}_{k\geq 3} \\
\frac{1}{2} & \omega\in [1,2]\cup[4,5]%
\end{array}%
\right.\text{ for } \omega\in\Sigma_{1}^z
\end{equation*}
From $S_1(\omega)=1$ on $\Sigma_{1}^z$, \eqref{null} and
\eqref{neg} we have that any martingale measure must satisfy
$Q([1,2]\cup[4,5])=0$. In other
words there exist polar sets with non-empty interior which implies $\mathcal{%
M}_+=\varnothing$.
\end{example}


\section{Appendix}

\subsection{proof of Theorem \protect\ref{theoremExtraction}}

\label{appendixProof}

\begin{lemma}[Lebesgue decomposition of $P$]
\label{lemmaDec}Let $\nu :=\sup_{Q\in \mathcal{M}}Q$. For any $P\in \mathcal{%
P}$ there exists a set $F\in \mathcal{F}$ such that $F\subseteq
(\Omega
_{\ast })^{c}$, and the measures $P_{c}(\cdot ):=P(\cdot \setminus F)$ and $%
P_{s}(\cdot ):=P(\cdot \cap F)$ satisfy
\begin{equation}
P_{c}\ll \nu ,\ P_{s}\perp \nu \qquad \text{and}\qquad
P=P_{c}+P_{s} \label{dec}
\end{equation}
\end{lemma}

\begin{proof}
We wish to apply Theorem 4.1 in \cite{LYL07} to $\mu =P\in \mathcal{P}$ and $%
\nu =\sup_{Q\in \mathcal{M}}Q$. It is easy to check that: 1) $\mu
$ and $\nu $ are monotone $[0,1]$-valued set functions on
$\mathcal{F}$ satisfying $\mu (\varnothing )=0$ and $\nu
(\varnothing )=0$; 2) $P$ is \textit{exhaustive},
i.e. if $\{A_{n}\}_{n\in \mathbb{N}}$ is a disjoint sequence then $%
P(A_{n})\rightarrow 0$ (indeed, $1\geq P\left( \cup
_{n}A_{n}\right)
=\sum_{n}P(A_{n})\geq 0\Rightarrow P(A_{n})\rightarrow 0$; 3) $\nu $ \textit{%
is weakly null additive}: if $A,B\in \mathcal{F}$ with $\nu
(A)=\nu (B)=0$ then $\nu (A\cup B)=0$ (indeed, if $\nu (A)=\nu
(B)=0$ then for any $Q\in \mathcal{M},$ $Q(A)=Q(B)=0$ which
implies $Q(A\cup B)=0$ and $\nu (A\cup B)=0 $); 4) $\nu $ is
\textit{continuous from below}. Indeed if $A_{n}\nearrow A$ then
$Q(A_{n})\uparrow Q(A),$ $Q(A)=\sup_{n}Q(A_{n})$ and
\begin{equation*}
\lim_{n\rightarrow \infty }\nu (A_{n})=\sup_{n}\nu
(A_{n})=\sup_{n}\sup_{Q\in \mathcal{M}}Q(A_{n})=\sup_{Q\in \mathcal{M}%
}\sup_{n}Q(A_{n})=\nu (A).
\end{equation*}%
Hence $\mu $ and $\nu $ satisfy all the assumptions of Theorem 4.1 in \cite%
{LYL07} and hence we obtain the existence of $F\in \mathcal{F}$ such that $%
\nu (F)=0$ and the decomposition in (\ref{dec}) holds true. From Proposition %
\ref{LemNOpolar}, $\forall A\in \mathcal{F}$ such that $A\subseteq
\Omega _{\ast }$ we have $\nu (A)>0$. Therefore, $F\subseteq
(\Omega _{\ast })^{c}$ and this concludes the proof.
\end{proof}

\begin{remark}
\label{remarkPsing} Observe that if $(\Omega_{\ast })^{c}\neq
\varnothing $ the set of probability measures with non trivial
singular part $P_{s}$ is non-empty. Simply take, for instance, any
convex combination of $\{\delta _{\omega }\mid \omega \in
(\Omega_{\ast })^{c}\}$.
\end{remark}

\paragraph{Preliminary considerations.}

We want to consider now the probabilistic model $(\Omega ,\{\mathcal{F}%
_{t}^{P}\}_{t\in I},S,P)$ and we need therefore to pass from
$\omega $-wise considerations to $P$-a.s considerations. For this
reason we first need to construct an auxiliary process $S_{t}^{P}$
with the property $S_{t}^{P}=S_{t}
$ $P$-a.s for any $t\in I$ in the same spirit of Lemma \ref{propPolari}.%
\newline
Let $P_{\Delta S_{T}}(\cdot ,\cdot ):\Omega \times \mathcal{B}(\mathbb{R}%
^{d})\mapsto \lbrack 0,1]$ be the conditional distribution of
$\Delta S_{T}$ and denote $\Upsilon _{\Delta S_{T}}$ its random
support. Define as in
Rokhlin \cite{Rokhlin} the set $A_{\Delta S_{T}}:=\{0\notin \text{ri(conv}%
\Upsilon _{\Delta S_{T}})\}$. It may happen that $P(A_{\Delta
S_{T}})=0$. In
this case $\mathfrak{B}_{T}$ and $\mathfrak{D}_{T-1}$ as in Lemma \ref%
{propPolari} are subset of $P$-null sets (respectively in
$\mathcal{F}_{T}$
and $\mathcal{F}_{T-1}$). Construct iteratively $X_{t}^P$ and $S_{t}^{P}$ as in %
\eqref{Q_modification2} and \eqref{Q_modification}. Denote $\Delta
X_{t}^{P}:=X_{t}^{P}-S_{t-1}$ and let
\begin{equation}
\tau:=\min \left\{ t\in I_{1}\mid P\left( A_{\Delta
X_{t}^{P}}\right)
>0\right\} .  \label{firstParb}
\end{equation}
Observe that $\tau$ is well defined since, from Lemma \ref{propPolari}, if $%
P(A_{\Delta X_{t}^{P}})=0$ for any $t\geq 1$ we have that
$\bigcup_{t\in I_{1}}\widetilde{\mathfrak{B}}_{t}=(\Omega _{\ast
})^{c}$ is a subset of a $P $-null set (cfr \eqref{ri}). This is a contradiction
since $P$ is not absolutely continuous with respect to $\nu$, henceforth the set $F$ from Lemma \ref{dec} satisfies $F\subseteq (\Omega
_{\ast })^{c}$ and $P(F)>0$. From now on denote $S_{t}=S_{t}\mathbf{1}_{t<\tau%
}+X_{t}^{P}\mathbf{1}_{t\geq \tau}$ which is a $P$-a.s. version of
$S$.

\begin{remark}
For any $t\in I_1$ denote $P_{t-1}(\cdot,\cdot):\left(\Omega,\mathcal{F}%
\right)\mapsto[0,1]$ the conditional probability of $P$ on
$\mathcal{F}_{t-1}$. Recall from Theorem \ref{SV} c] that there exists
$N_1\in\mathcal{F}_{t-1}$
with $P(N_1)=0$ such that for any $\omega\in\Omega\setminus N_1$ we have $%
P_{t-1}\left(\omega,\Sigma_{t-1}^{z(\omega)}\right)=1$ where $%
z(\omega)=S_{0:T}(\omega)$.
\end{remark}

\paragraph{Construction of a $P$-arbitrage from $\mathbb{H}$.}

Fix time $\tau$ and denote $A_{\tau}:=A_{\Delta
S_{\tau}}$. For any $\omega\in \Omega$ the
level set $\Sigma_{\tau-1}^z$ can
be decomposed as $\Sigma_{\tau-1}^z=\cup_{i=1}^{\beta_{\tau,z}}B^i_{\tau,z}\cup B^*_{\tau,z} $%
. Define for any $z\in \mathbf{Z}$
\begin{equation*}
j_z:=\inf\left\{j\in\{1,\dots,\beta_{\tau,z}\}\mid P(\omega,B^j_{\tau,z})>0\ \forall \omega\in\Sigma_{\tau-1}^z\right\}
\end{equation*}
and recall that $P(\cdot,B^j_{\tau,z})$ is constant on $\Sigma_{\tau-1}^z$ (Theorem \ref{SV} b]). Define $N_2:=\bigcup_{z\in \mathbf{Z}_f}\cup_{i=1}^{j_z-1}B^i_{\tau,z}$ where $%
\mathbf{Z}_f:=\left\{z\in \mathbf{Z}\mid j_z<\infty\right\}$. $N_2$ is a $%
\bar{P}$-null set since for any $\omega\in N_1^c$ we have $\bar{P}%
(\omega,N_2)=\bar{P}(\omega,\cup_{i=1}^{j_z-1}B^i_{\tau,z})=0$ hence $\bar{P}(N_2)=%
\bar{P}(N_1\cap N_2)+\bar{P}(N_1^c\cap N_2)=0$ (see also Lemma \ref%
{LemmaPbar} below). Recall that $\bar{P}(\cdot)$ and $\bar{P}(\omega,\cdot)$ denote the completion of $P(\cdot)$ and $P(\omega,\cdot)$ respectively.

\bigskip \noindent Denote $N:=N_1\cup N_2$. We are now able to define the
following multifunction $%
\Psi:\Omega\mapsto2^{\mathbb{R}^d}$ with values in the power set of $\mathbb{%
R}^d$.
\begin{equation}\label{defPSi}
\Psi(\omega):=\left\lbrace%
\begin{array}{lll}
\Delta S_{\tau}\left(\Sigma_{\tau-1}^{z(\omega)}\cap
N^c\right) &\qquad&
\omega\in N^c \\
\varnothing & &\text{otherwise}%
\end{array}%
\right.
\end{equation}%
In Lemma \ref{multiPsimeas} we show that $\Psi$ is $\mathcal{F}^P_{\tau-1}$-measurable. We apply now an argument similar to
\cite{Rokhlin}. Denote $\mathbb{S}^d_1$ the unitary closed ball in
$\mathbb{R}^d$, lin$(\chi)$ the linear space generated by $\chi$
and $\chi^\circ$ the polar cone of $\chi$. By preservation of
measurability (see Proposition \ref{preservation}) the
(closed-valued) multifunction
\begin{equation*}
\omega\mapsto G_0(\omega):=\text{lin}(\Psi(\omega))\cap \left(-\text{cone }%
\Psi(\omega)\right)^\circ\cap\mathbb{S}^d_1
\end{equation*}
is also $\mathcal{F}_{\tau-1}^P$-measurable and
$G_0(\omega)\neq\varnothing$
iff $\omega\in A_{\tau}\cap N^c$, hence $A_{\tau}=\{0\notin\text{%
ri(conv}\Upsilon_{\Delta S^P_{\tau}})\}$ is $\mathcal{F}_{\tau-1}^P$%
-measurable. Note that we already have that $G_0(\omega)\subseteq\mathbb{H}_{%
\tau}(\omega)$ for $P$-a.e. $\omega\in\Omega$. Indeed fix
$\omega\notin N$ and consider the level set
$\Sigma_{\tau-1}^{z(\omega)}$ and its decomposition as in Lemma
\ref{spezzamento}. By construction of $G_0$ we have that any $g\in
G_0(\omega)\neq\varnothing$ satisfies $g\cdot\Delta
S_{\tau}(\omega)\geq 0$ for any $\omega\in \cup_{i=j_z}^{\beta_{\tau,z}}B_{\tau,z}^i\cup
B^*_{\tau,z}$ and thus $g\in\mathbb{H}(\omega)$.\newline

Nevertheless, the random set $G_0(\omega)$ contains those
$g\in\mathbb{S}_1^d$ such that $g\cdot\Delta S_{\tau}(\omega)=0$. Thus,
we will not extract a
measurable selection from $G_0$ but we will rather consider for any $%
n\in\mathbb{N}$ the following closed-valued multifunction
\begin{equation*}
\omega\mapsto G_n(\omega):=\text{lin}(\Psi(\omega))\cap \left\lbrace v\in%
\mathbb{R}^d\mid \langle v,s\rangle\geq \frac{1}{n}\quad \forall
s\in\Psi(\omega)\setminus \{0\}\right\rbrace\cap\mathbb{S}^d_1,\quad n\geq 1
\end{equation*}
and seek for a measurable selection of $G:=\cup_{n=0}^\infty G_n$.
From Lemma \ref{meas-eps} all the random sets $G_n$ are
$\mathcal{F}_{\tau-1}^P$-measurable and therefore the same is true
for $G$. Now, for any $n\geq 0$, let $\tilde{H}_n$ a measurable
selection of $G_n$ on $\{G_n\neq\varnothing\}$ which always exists
for a (measurable)
closed-valued multifunction with $\tilde{H}_n(\omega)=0$ if $%
G_n(\omega)=\varnothing$. Define therefore
\begin{equation}  \label{arbitrage_def}
H_k:=\sum_{n=0}^k \tilde{H}_n\qquad\text{and}\quad
B_k:=\mathcal{V}^+_{H_k}
\end{equation}
By construction $B_k$ is an increasing sequence of sets converging to $%
\cup_z B^{j_z}_{\tau,z}$ which is therefore measurable and it satisfies
\begin{equation*}
P(\cup_z B^{j_z}_{\tau,z})=\int_{\Omega}P(\omega,\cup_z
B^{j_z}_{\tau,z})dP(\omega)=\int_{\Omega\setminus
N}P(\omega,B^{j_{z}}_{\tau,z})dP(\omega)\geq \int_{A_{\tau}\setminus
N}P(\omega,B^{j_{z}}_{\tau,z})dP(\omega)>0
\end{equation*}
which follows from the definition of conditional probability, $P(A_{\tau%
})>0$ and $P(\omega,B^{j_z}_{\tau,z})>0$ for every $\omega \in
A_{\tau}\setminus N$. We can therefore conclude that there
exists $m\geq 0$ such that $P(B_m)>0$
and since obviously $H_m\Delta S_{\tau}\geq 0$ we have that $H_m$ is a $P$%
-arbitrage. The normalized random variable $H_{\tau}^P:=H_m(\omega)/%
\|H_m(\omega)\|$ is a measurable selector of the multifunction
$G_0$ since it satisfies $H_{\tau}^P(\omega)\in\cup_{n=1}^m
G_n(\omega)\subseteq G(\omega)\subseteq
\mathbb{H}_{\tau}(\omega)$ $P$-a.s. and thus the desired
strategy is given by $H_s^P=H_{\tau}^P\mathbf{1}_{\tau}(s)$.

\begin{lemma}\label{multiPsimeas}
The multifunction $\Psi$ defined in \eqref{defPSi} is $\mathcal{F}_{\tau-1}^P$-measurable.\end{lemma}
\begin{proof}
Recall that by definition the multifunction $\Psi$ is measurable
iff for any open set $V\subseteq \mathbb{R}^d$ we have
$\{\omega\mid \Psi(\omega)\cap V\neq \varnothing\} $ is a
measurable set. Observe that
\begin{equation*}
\Psi^{-1}(V):=\{\omega\mid \Psi(\omega)\cap V\neq \varnothing\}=S_{\tau%
-1}^{-1}\left[S_{\tau-1}\left(\Delta S_{\tau}^{-1}(V)\cap N^c\right)%
\right]\cap N^c
\end{equation*}
Let us show that the complement of this set is $\mathcal{F}^P_{\tau-1}$%
-measurable from which the thesis will follow. Observe that for
any function
$f$ and for any set $A$ we have $(f^{-1}(A))^c=f^{-1}(A^c)$ so that%
\begin{eqnarray*}
(\Psi^{-1}(V))^c&=&S_{\tau-1}^{-1}\left[S_{\tau-1}\left(\Delta S_{\bar{%
t}}^{-1}(V)\cap N^c\right)^c\right]\cup N \\
&=&S_{\tau-1}^{-1}\left[S_{\tau-1}\left((\Delta S_{\tau%
}^{-1}(V))^c\cup N\right)\right]\cup N \\
&=&S_{\tau-1}^{-1}\left[S_{\tau-1}\left(\Delta S_{\tau%
}^{-1}(V^c)\cup N\right)\right]\cup N
\end{eqnarray*}
Note now that $A_1:=\Delta S_{\tau}^{-1}(V^c)\cup N$ is an
analytic set since it is union of a Borel set and a $\bar{P}$-null
set. The set $B_1:=S_{\tau-1}(A_1)$ is an analytic subset of
$\mathbb{R}^d$ since $S$ is a Borel function and image of an analytic
set through a Borel measurable function is analytic. Finally
$A_2:=S_{\tau-1}^{-1}(B_1)$ is an analytic subset of $\Omega$
since pre-image of an analytic set through a Borel measurable function
is analytic. Since $P$-completion of $\mathcal{F}$ contains any
analytic set, $A_2\cup N$ is also analytic and belongs to
$\mathcal{F}^P_{\tau-1}$.

\begin{remark} For sure $A_2\cup N$ is analytic and belongs to $\mathcal{F}%
^P$. The heuristic for $A_2\cup N$ belonging to
$\mathcal{F}^P_{\tau-1}$ should be that this set is union of
atoms of $\mathcal{F}^P_{\tau-1}$. More formally, since $B_1$ is analytic in $\mathbb{R}^d$ for any
measure $\mu$ there exists $F,G$ such that $B_1=F\cup G$ with $F$
a Borel set and $G$ a subset of $\mu$-null measure (because
analytic sets are in the completion of $\mathcal{B}$ respect to
any measure $\mu$). If we take
$\mu$ as the distribution of $S_{\tau-1}$ under $P$ we have $A_2=S_{\tau-1}^{-1}(F)\cup S_{\tau-1}^{-1}(G)$. Since $S_{\tau-1}^{-1}(F)\in%
\mathcal{F}_{\tau-1}$ and $S_{\tau-1}^{-1}(G)$ is a subset of a $%
\mathcal{F}_{\tau-1}$-measurable $P$-null set, we have $A_2\in\mathcal{F}%
_{\tau-1}^P$ and hence also $A_2\cup N$.
\end{remark}
\end{proof}

\begin{lemma}
\label{LemmaPbar}Let $(\Omega,\mathcal{F},P)$ a probability space and $%
\mathcal{G}$ a subsigma-algebra of $\mathcal{F}$. Let $P_\mathcal{G}%
(\omega,\cdot)$ the conditional probability of $P$ on
$\mathcal{G}$. Then
\begin{equation}  \label{Pbar}
\bar{P}(A)=\int_{\Omega\setminus
N(A)}\bar{P}_\mathcal{G}(\omega,A)dP(\omega)\qquad
A\in\mathcal{F}^P
\end{equation}
where $\bar{P}_\mathcal{G}(\omega,\cdot)$ is the completion of $P_{\mathcal{G}}(\omega,\cdot)$ and $%
N(A)\in\mathcal{G}$ is a $P$-null set which depends on $A$.
\end{lemma}

\begin{proof}
It is easy to see that every set in $\mathcal{F}^P$ is union of a set $F\in%
\mathcal{F}$ and a subset of a $P$-null set. For any $F\in\mathcal{F}$, $%
\bar{P}(F)=P(F)$ and $P_\mathcal{G}(\omega,F)=\bar{P}_\mathcal{G}(\omega,F)$ so equality (\ref{Pbar}%
) is obvious from the definition of conditional probability (with $%
N(F)=\varnothing$). Let $A$ be a subset of a P-null set $A_1$. $%
0=P(A_1)=\int_{\Omega}P_\mathcal{G}(\omega,A_1)dP(\omega)$ which means that $%
P_\mathcal{G}(\omega,A_1)=0$ $P$-a.s. Thus, we also have $%
\bar{P}_\mathcal{G}(\omega,A)=0$ $P$-a.s. from which the equality (\ref{Pbar}) follows with $%
N(A)=\{\omega\in\Omega:P_\mathcal{G}(\omega,A_1)>0\}\in\mathcal{G}$.
\end{proof}

\paragraph{Measurable selection results.}

\begin{lemma}
\label{meas-eps} Let $(\Omega,\mathcal{A})$ a measurable space and $%
\Psi:\Omega\mapsto 2^{\mathbb{R}^d}$ an $\mathcal{A}$-measurable
multifunction. Let $\varepsilon>0$ then
\begin{equation*}
\Psi^\varepsilon:\omega\mapsto \left\{v\in \mathbb{R}^d\mid\langle
v,s\rangle\geq \varepsilon \quad \forall s\in\Psi(\omega)\setminus
\{0\}\right\}
\end{equation*}
is an $\mathcal{A}$-measurable multifunction.
\end{lemma}

\begin{proof}
Observe first that for $v\in\mathbb{R}^d$
\begin{equation}  \label{dense}
\begin{array}{clcll}
\vspace{0.1in} \langle v,s\rangle\geq \varepsilon\  & \forall
s\in\Psi(\omega)\setminus
\{0\} & \Leftrightarrow & \langle v,s\rangle\geq \varepsilon\  & \forall s\in\overline{%
\Psi}(\omega)\setminus \{0\} \\
&  & \Leftrightarrow & \langle v,s\rangle\geq \varepsilon\  &
\forall s\in
D(\omega)\setminus \{0\}%
\end{array}%
\end{equation}
where $D(\omega)$ is a dense subset of $\Psi(\omega)$. This is
obvious by continuity of the scalar product. With no loss of
generality we can then consider $\Psi$ closed valued and we denote
by $\psi_n$ its Castaing representation (see Theorem 14.5 in
\cite{R} for details). For any $n\in\mathbb{N}$ consider the
following closed-valued multifunction:
\begin{equation*}
\Lambda_n(\omega)=\left\lbrace
\begin{array}{ll}
\left\{v\in \mathbb{R}^d\mid\langle v,\psi_n(\omega) \rangle\geq
\varepsilon \right\} &
\text{if }\omega\in\text{dom }\Psi,\ \psi_n(\omega)\neq 0 \\
\mathbb{R}^d & \text{if }\omega\in\text{dom }\Psi,\ \psi_n(\omega)= 0 \\
\varnothing & \text{otherwise }%
\end{array}%
\right.
\end{equation*}
We claim that $\Lambda_n$ is measurable for any $n\in\mathbb{N}$ from which the map $\omega\mapsto\bigcap_{n\in\mathbb{N}%
}\Lambda_n(\omega)$ is also measurable (cfr Proposition
\ref{preservation}). From (\ref{dense}) we thus conclude that
$\Psi^\varepsilon$ is measurable.\newline We are only left to show
the claim. To this end observe that $\Lambda_n(\omega)$ has
non-empty interior on $\{\Lambda_n\neq\varnothing\}$. Therefore
for any open set $V\subseteq\mathbb{R}^d$ we have
$\{\omega\in\Omega\mid\Lambda_n(\omega)\cap V\neq
\varnothing\}=\{\omega\in\Omega\mid\text{int}(\Lambda_n(\omega))\cap
V\neq \varnothing\}$. Note now that
\begin{equation*}
\{\omega\in\Omega\mid\text{int}(\Lambda_n(\omega))\cap V\neq
\varnothing\}=\psi_n^{-1}\left(\Pi_y\left(\Pi_x^{-1}(V)\cap\langle\cdot,%
\cdot\rangle^{-1}(\varepsilon,\infty)\right)\right)\cup
\psi_n^{-1}(0)
\end{equation*}
which is measurable (when $\psi_n$ is measurable) from the continuity of $%
\langle\cdot,\cdot\rangle$ and from the open mapping property of the projections $%
\Pi_x,\Pi_y:\mathbb{R}^d\times \mathbb{R}^d\mapsto\mathbb{R}^d$.
\end{proof}

\begin{proposition}
\label{preservation}[Proposition 14.2-11-12 \cite{R}] Consider a class of $%
\mathcal{A}$-measurable set-valued functions. The following
operations preserve $\mathcal{A}$-measurability: countable unions,
countable intersections (if the functions are closed-valued),
finite linear combination, convex/linear/affine hull, generated
cone, polar set, closure.
\end{proposition}

\subsection{Complementary results}

Recall that we are assuming that $\Omega $ is a Polish space.

\begin{lemma}
\label{pasting} Let $Q_{i}\in \mathcal{M}$ for any $i\in
\mathbb{N}$. Then
\begin{equation*}
Q:=\sum_{i\in \mathbb{N}}\frac{1}{2^{i}}Q_{i}\in \mathcal{M}
\end{equation*}
\end{lemma}

\begin{proof}
We first observe that $Q\in \mathcal{P}$ hence we just need to
show that is
a martingale measure. Consider the measures $Q_{k}:=\sum_{i=1}^{k}\frac{1}{%
2^{i}}Q_{i}$, which are not probabilities, and note that for each
$k$ we have: $\int_{\Omega }1_{B}\Delta S_{t}dQ_{k}=0$ if $B\in
\mathcal{F}_{t-1}$.
We observe that $\Vert Q_{k}-Q\Vert \rightarrow 0$ for $k\rightarrow \infty $%
, where $\Vert \cdot \Vert $ is the total variation norm. We have
indeed that
\begin{equation*}
\sup_{A\in \mathcal{F}}|Q_{k}(A)-Q(A)|=\sup_{A\in \mathcal{F}%
}\sum_{i=k+1}^{\infty }\frac{1}{2^{i}}Q_{i}(A)=\sum_{i=k+1}^{\infty }\frac{1%
}{2^{i}}\rightarrow 0\text{ as }k\rightarrow \infty .
\end{equation*}%
In particular we have $Q_{k}(A)\uparrow Q(A)$ for any $A\in
\mathcal{F}$.
This implies that for a non negative random variable $X$%
\begin{equation*}
\lim_{k\rightarrow \infty }\int_{\Omega
}XdQ_{k}=\lim_{k\rightarrow \infty }\left( \sup_{f\in
\mathfrak{S}}f_{j}(\omega )Q_{k}(A_{j})\right) =\sup_{k}\sup_{f\in
\mathfrak{S}}f_{j}(\omega )Q_{k}(A_{j})=
\end{equation*}%
\begin{equation*}
\sup_{f\in \mathfrak{S}}\sup_{k}f_{j}(\omega
)Q_{k}(A_{j})=\sup_{f\in \mathfrak{S}}f_{j}(\omega
)Q(A_{j})=\int_{\Omega }XdQ
\end{equation*}%
where $\mathfrak{S}$ are the simple function less or equal than
$X$. For any
$B\in \mathcal{F}_{t-1}$ we then have:%
\begin{eqnarray*}
&&E_{Q}\left[ 1_{B}\Delta S_{t}\right] =\int_{\Omega }(1_{B}\Delta
S_{t})^{+}dQ-\int_{\Omega }(1_{B}\Delta S_{t})^{-}dQ \\
&=&\lim_{k\rightarrow \infty }\int_{\Omega }(1_{B}\Delta
S_{t})^{+}dQ_{k}-\lim_{k\rightarrow \infty }\int_{\Omega
}(1_{B}\Delta S_{t})^{-}dQ_{k}=\lim_{k\rightarrow \infty
}\int_{\Omega }1_{B}\Delta S_{t}dQ_{k}=0.
\end{eqnarray*}
\end{proof}

\begin{lemma}
\label{denseDirac}For any dense set $D\subseteq \Omega $, the set
of
probabilities $co(\{\delta _{\omega }\}_{\omega \in D})$ is $\sigma (%
\mathcal{P},C_{b})$ dense in $\mathcal{P}$.
\end{lemma}

\begin{proof}
Take $\omega ^{\ast }\notin D$ and let $\omega _{n}\rightarrow
\omega ^{\ast }$. Note that for every open set $G$ we have
$\liminf \delta _{\omega _{n}}(G)\geq \delta _{\omega ^{\ast
}}(G)$ and this is equivalent to the weak convergence $\delta
_{\omega _{n}}\overset{w}{\rightarrow }\delta _{\omega ^{\ast }}$.
Observe that for every set $X$ we have
\begin{equation*}
\overline{co(X)}=\overline{co}(X):=\bigcap \left\{ C\mid C\text{
convex closed containing }X\right\} =\overline{co}(\overline{X}).
\end{equation*}%
Hence, by taking $X=\{\delta _{\omega }\}_{\omega \in D}$ and by $\sigma (%
\mathcal{P},C_{b})$ density of the set of measures with finite support in $%
\mathcal{P},$ we obtain the thesis.
\end{proof}

\begin{lemma}
\label{extended} Let $\mathcal{F}=\mathcal{B}(\Omega )$ be the
Borel sigma
algebra and let $\widetilde{\mathcal{F}}$ be a sigma algebra such that $%
\mathcal{F\subseteq }\widetilde{\mathcal{F}}$. The set $\widetilde{\mathcal{P%
}}:=\{\widetilde{P}:\widetilde{\mathcal{F}}\rightarrow \lbrack
0,1]\mid
\widetilde{P}$ is a probability$\}$ is endowed with the topology $\sigma (%
\widetilde{\mathcal{P}},C_{b})$. Then

\begin{enumerate}
\item If $A\subseteq \Omega $ is dense in $\Omega ,$ then
$co(\{\delta _{\omega }\}_{\omega \in A})$ is $\sigma
(\widetilde{\mathcal{P}},C_{b})$ dense in
$\widetilde{\mathcal{P}}$. Notice that any element $Q\in
co(\{\delta _{\omega }\}_{\omega \in A})$ can be extended to $\widetilde{%
\mathcal{F}}$.

\item If $D\subseteq \Omega $ is closed then
\begin{equation*}
\widetilde{\mathcal{P}}(D):=\{\widetilde{P}\in
\widetilde{\mathcal{P}}\mid supp(\widetilde{P})\subseteq D\}
\end{equation*}%
is $\sigma (\widetilde{\mathcal{P}},C_{b})$ closed, where the
support is well-defined by
\begin{equation*}
supp(\widetilde{P}):=\bigcap \{C\in \mathcal{C}\mid
\widetilde{P}(C)=1\}
\end{equation*}%
and $\mathcal{C}$ are the closed sets in $(\Omega ,d)$.
\end{enumerate}
\end{lemma}

\begin{proof}
By construction for any $\widetilde{P}\in \widetilde{\mathcal{P}}$ we have $%
\int fd\widetilde{P}=\int fdP$ for any $f\in C_{b}$ where $P\in
\mathcal{P}$ is the restriction of $\widetilde{P}$ to
$\mathcal{F}$. \newline
To show the first claim we choose any $\widetilde{P}\in \widetilde{\mathcal{P%
}}$. Consider $P\in \mathcal{P}$ the restriction of $\widetilde{P}$ to $%
\mathcal{F}$. Then from Lemma \ref{denseDirac} there exists a sequence $%
Q_{n}\in co(\{\delta _{\omega }\}_{\omega \in A})$ such that $\int
fdQ_{n}\rightarrow \int fdP$ for every $f\in C_{b}$. As a
consequence $\int fdQ_{n}\rightarrow \int fd\widetilde{P}$, for
every $f\in C_{b}$. \newline To show the second claim consider any
net $\{\widetilde{P}_{\alpha
}\}_{\alpha }\subset \widetilde{\mathcal{P}}(D)$ such that $\widetilde{P}%
_{\alpha }\overset{w}{\longrightarrow }\widetilde{P}$. We want to show that $%
\widetilde{P}\in \widetilde{\mathcal{P}}(D)$. Consider $P_{\alpha
},P$ the restriction to $\mathcal{F}$ of $\widetilde{P}_{\alpha
},\widetilde{P}$ respectively. Then $P_{\alpha
}\overset{w}{\longrightarrow }P$. Notice that by definition
$supp(P_{\alpha })=supp(\widetilde{P}_{\alpha })\subseteq D$ and
$supp(P)=supp(\widetilde{P})$. Moreover the set
$\mathcal{P}(D)=\{P\in \mathcal{P}\mid supp(P)\subseteq D\}$ is
$\sigma (\mathcal{P},C_{b})$ closed
(Theorem 15.19 in \cite{Aliprantis}) so that $D\supseteq supp(P)=supp(%
\widetilde{P})$.
\end{proof}

\begin{proof}[Proof of Proposition \protect\ref{WS-characterization}, item
(2)]
Recall that an Open Arbitrage in $\widetilde{\mathcal{H}}$ is a $\{%
\widetilde{\mathcal{F}}\}$-predictable processes $H=[H^{1},\ldots
,H^{d}]$ such that $V_{T}(H)\geq 0$ and
$\mathcal{V}_{H}^{+}=\{V_{T}(H)>0\}$ contains an open set.

First we show that $H\in W(\widetilde{\sigma
},\widetilde{\mathcal{H}})$ implies $V_{T}(H)(\omega )\geq 0$ for
all $\omega \in \Omega $. We need only to show that $B:=\left\{
\omega \in \Omega \ \mid \ V_{T}(H)(\omega )<0\right\} $ is empty.
By contradiction, let $\omega \in B$, take any $P\in \mathcal{U}$
and define the probability $P_{\lambda }:=\lambda \delta
_{\omega }+(1-\lambda )P$. Since $V_{T}(H)\geq 0$ $P$-a.s. we must have $%
P(\omega )=0$, otherwise $P(B)>0$. However, $P_{\lambda }(B)\geq
P_{\lambda }(\omega )=\lambda >0$ for all positive $\lambda $ and
$P_{\lambda }$ will
belongs to $\mathcal{U}$, as $\lambda \downarrow 0$, which contradicts $%
V_{T}(H)\geq 0$ $P$-a.s. for any $P\in \mathcal{U}$.

To prove the equivalence, assume first that $H\in W(\widetilde{\sigma },%
\widetilde{\mathcal{H}})$. We claim that
$(\mathcal{V}_{H}^{+})^{c}=\left\{ \omega \in \Omega \mid
V_{T}(H)=0\right\} $ is not dense in $\Omega $. This will imply
the thesis as the open set $int(\mathcal{V}_{H}^{+})$ will then
be a not empty on which $V_{T}(H)>0$. Suppose by contradiction that $%
\overline{(\mathcal{V}_{H}^{+})^{c}}=\Omega $. We know by Lemma \ref%
{extended} that the corresponding set $\mathcal{Q}$ of embedded
probabilities $co(\{\delta _{\omega }\}_{\omega \in (\mathcal{V}%
_{H}^{+})^{c}})$ is weakly dense in $\widetilde{\mathcal{P}}$ and
hence it intersects, in particular, the weakly open set
$\mathcal{U}$. However, for every $P\in \mathcal{Q}$ we have
$V_{T}(H)=0\ P$-a.s. and so this contradicts the assumption.

Suppose now that $H\in \widetilde{H}$ is an Open Arbitrage. Note
that from
Lemma \ref{extended} if $F$ is a closed subset of $\Omega $, then $%
\widetilde{\mathcal{P}}(F):=\{P\in \widetilde{\mathcal{P}}\mid \
supp(P)\subset F\}$ is $\sigma
(\widetilde{\mathcal{P}},C_{b})$-closed. Since $H$ is an Open
Arbitrage then $\mathcal{V}_{H}^{+}$ contains an open set and in
particular $G:=\overline{(\mathcal{V}_{H}^{+})^{c}}$ is a closed
set strictly contained in $\Omega $. Observe then that $\left( \widetilde{%
\mathcal{P}}(G)\right) ^{c}$ is a non empty $\sigma (\widetilde{\mathcal{P}}%
,C_{b})$-open set of probabilities such that for all $P\in
\mathcal{U}$ we have $V_{T}(H)\geq 0$, $P$-a.s. and
$P(\mathcal{V}_{H}^{+})>0$.
\end{proof}

\end{document}